\documentclass[a4paper,twocolumn,11pt]{quantumarticle}
\pdfoutput=1
\usepackage[utf8]{inputenc}
\usepackage[english]{babel}
\usepackage[T1]{fontenc}

\usepackage[dvips]{graphicx}
\usepackage{bm, amsmath,amssymb,amsthm,mathrsfs,amsfonts,dsfont}
\usepackage{ascmac}
\usepackage{epsfig}
\usepackage{enumerate}
\usepackage{color}
\usepackage{graphicx}
\usepackage{algorithm}
\usepackage{algorithmic}
\usepackage{comment}
\usepackage{appendix}
\usepackage{here}
\usepackage{tabularx}
\usepackage{dcolumn}
\usepackage[caption=false]{subfig}
\usepackage{adjustbox}
\usepackage[hidelinks]{hyperref}
\usepackage{bookmark}
\usepackage{multirow}
\usepackage{dsfont}
\usepackage{hyperref}
\usepackage[numbers]{natbib}

\newcommand\scalemath[2]{\scalebox{#1}{\mbox{\ensuremath{\displaystyle #2}}}}

\newtheorem{claim}{Claim}

\iftrue
  \usepackage{marginnote}
\fi

\begin{document}


\title{Quantum Error Suppression with Subgroup Stabilisation}


\author{Bo Yang}
\email{Bo.Yang@lip6.fr}
\affiliation{Laboratorie d'Informatique de Paris 6, Centre National de la Recherche Scientifique, Sorbonne Université, 4 place Jussieu, 75005 Paris, France}
\orcid{0000-0001-8266-7690}

\author{Elham Kashefi}
\email{Elham.Kashefi@lip6.fr}
\affiliation{Laboratorie d'Informatique de Paris 6, Centre National de la Recherche Scientifique, Sorbonne Université, 4 place Jussieu, 75005 Paris, France}
\affiliation{School of Informatics, University of Edinburgh, 10 Crichton Street, Edinburgh EH8 9AB, United Kingdom}
\orcid{0000-0001-6280-9604}

\author{Dominik Leichtle}
\email{dominik.leichtle@ed.ac.uk}
\affiliation{School of Informatics, University of Edinburgh, 10 Crichton Street, Edinburgh EH8 9AB, United Kingdom}
\orcid{0000-0002-4882-889X}

\author{Harold Ollivier}
\email{harold.ollivier@inria.fr}
\affiliation{Institut National de Recherche en Informatique et en Automatique, 2 Rue Simone Iff, 75012 Paris, France}


\begin{abstract} \label{sec:abstract}
Quantum state purification is the functionality that, given multiple copies of an unknown state, outputs a state with increased purity.
This will be an essential building block for near- and middle-term quantum ecosystems before the availability of full fault tolerance, where one may want to suppress errors not only in expectation values but also in quantum states.
We propose an effective state purification gadget with a moderate quantum overhead by projecting $M$ noisy quantum inputs to their symmetric subspace defined by a set of projectors forming a symmetric subgroup with order $M$.
Our method, applied in every short evolution over $M$ redundant copies of noisy states, can suppress both coherent and stochastic errors by a factor of $1/M$, respectively.
This reduces the circuit implementation cost $M$ times smaller than the state projection to the full symmetric subspace proposed by Barenco et al. more than two decades ago.
We also show that our gadget purifies the depolarised inputs with probability $p$ to asymptotically $O\left(p^{2}\right)$ with an optimal choice of $M$ when $p$ is small.
The sampling cost scales $O\left(p^{-1}\right)$ for small $p$, which is also shown to be asymptotically optimal.
Our method provides flexible choices of state purification depending on the hardware restrictions before fully fault-tolerant computation is available.
\end{abstract}

\maketitle


\section{Introduction \label{sec:introduction}}

\par
The practical execution of quantum algorithms with notable advantages over classical computers~\cite{shor1994algorithms, shor1999polynomial, grover1996fast, harrow2009quantum} may require fault-tolerant quantum computation (FTQC), which can be achieved with quantum error correction (QEC) codes by encoding quantum information into redundant code space to detect and correct errors during the computation~\cite{nielsen2010quantum, lidar2013quantum, roffe2019quantum, horsman2012surface, fowler2012surface, fowler2018low}.
The FTQC procedure, including logical gate operations and error correction, is performed with $O\left(\mathrm{poly}\log\left(N\right)\right)$ times of overhead to the polynomial time computation with problem size $N$ and under the bounded failure probability in each physical gate operation.
Although recent advances in quantum hardware have reached a juncture where the break-even fidelity of logical qubits encoded by certain QEC codes can be witnessed~\cite{ni2023beating, gupta2024encoding}, the realisation of practical FTQC remains elusive since the size of current physical devices is still small and physical operations on them are still too noisy.

\par
Instead, quantum error mitigation (QEM)~\cite{temme2017error, cai2023quantum, koczor2021exponential, huggins2021virtual, yang2022efficient, yang2023dual} attached to hybrid quantum-classical approaches ~\cite{cerezo2021variational, endo2021hybrid, peng2020simulating, yuan2021quantum, harada2023noise, eddins2022doubling} has become mainstream for harnessing the abilities of contemporary and near-term quantum devices.
However, since the QEM methods restore noiseless measurement results (i.e. expectation values) through classical pre- and postprocessing, they are targeted for BQP computation. 
Suppressing errors during sampling problems will require quantum pre- and postprocessing.
Besides, a wide class of QEM methods come with an exponential sampling cost to the desired estimation bias or to the circuit depth~\cite{takagi2022fundamental, takagi2023universal, tsubouchi2023universal, quek2022exponentially}, revealing their limitations in scalability.

\par
The aforementioned issues highlight the demand for suppressing errors in quantum states, particularly to address sampling problems.
In this line, the state purification methods by post-selection find a unique position as being capable of improving the quality of the quantum state before it is sampled.
This scheme suppresses errors for quantum outputs consuming quantum resources in accord with QEM, which is advantageous for quantum computers with small scales before the arrival of full FTQC.
The early works can be found over two decades ago~\cite{bennett1996purification, cirac1999optimal, keyl2001rate, barenco1997stabilization}, sharing the idea of taking noisy quantum states and returning a purified quantum state by state projection and post-selection.

\par
Barenco et al.~\cite{barenco1997stabilization} proposed the stabilisation of multiple inputs to their symmetric subspace by the projectors forming a full symmetric group.
The simplest example is the SWAP test, which maps two inputs into their SWAP invariant subspace.
Barenco et al. generalised this to $M$ inputs and analysed the error suppression rate when repetitively applying this gadget per a short time interval over $M$ redundant states, each under the same computational task.
Projecting $M$ states to their permutation invariant subspace, their gadget reduces errors by a factor of $1/M$ in terms of the probability of the purified output being unwanted states.
Particularly, they demonstrated it for noisy inputs under local coherent errors and first-order perturbed stochastic errors, respectively.

\par
The perfect execution of the gadget by Barenco et al. can effectively stabilise errors growing exponentially to evolution time, such as decoherence.
Supposing we execute a polynomial time computational task which gives the correct result at each step with probability $1-\epsilon$ with a constant $\epsilon$, the success probability is at least $\left(1-\epsilon\right)^{N} \sim e^{-\epsilon N}$ after $N$ steps.
The gadget of Barenco et al. aims to amplify this into $e^{-\epsilon N / M}$ by stabilising $M$ copies at each step.
As a result, the success probability can be kept at the level $1-\delta$ by taking $M = - \epsilon N / \log \left(1-\delta\right)$, which is polynomial in $N$.

\par
However, the quantum circuit for this gadget requires $O\left(M\log\left(M\right)\right)$ measurements (i.e. ancillary qubits) and $O\left(\left(M\log\left(M\right)\right)^{2}\right)$ controlled-SWAP operations, which may not yet be feasible for applying it many times during the computation on the near- and middle-term devices.
In this work, we propose a lighter implementation for suppressing errors by projecting $M$ redundant states to their rotation invariant subspace using a linear combination of projectors forming the cyclic group $C_{M}$, which we name ``Cyclic Group Gadget (CGG)".
The implementation cost of our gadget is reduced to $O\left(\log\left(M\right)\right)$ measurements and $O\left(M\log\left(M\right)\right)$ controlled-SWAP operations.
Yet our gadget still keeps the same purification rate as Barenco et al. when stabilising the errors in a short evolution time.
We also provide another implementation projecting $M$ inputs to a symmetric subspace associated with the group $\left(\mathbb{Z}/2\mathbb{Z}\right)^{\log\left(M\right)}$ inspired by the quantum circuit in Chabaud et al.~\cite{chabaud2018optimal}, which we name ``Generalised SWAP Gadget (GSG)".
The group $\left(\mathbb{Z}/2\mathbb{Z}\right)^{\log\left(M\right)}$ has the same order as the cyclic group $C_{M}$, and the circuit implementation cost and purification rate of GSG is also the same as CGG.

\par
Our gadget can also be used for state purification in a single round, given inputs with a bounded error rate.
The optimality of state purification has also been continuously explored for over two decades~\cite{cirac1999optimal, keyl2001rate, childs2023streaming}.
Recently, Childs et al.~\cite{childs2023streaming} have proved the optimality in sampling cost of the recursive application of SWAP gadgets for purifying the depolarised states with general dimensions.
We also analyse the error suppression rate of single-round application of our method over $M$ depolarised inputs with the depolarising rate $p$.
We observe that there is an optimal number of copies $M$ depending on $p$ for both CGG and GSG that maximises the error suppression rate.
Within the scope of small $p$, the maximal purification rate can be approximated as $O\left(p^{2}\right)$ with $O\left(p^{-1}\right)$ sampling cost, showing the same order as Childs et al.~\cite{childs2023streaming} in purification rate and sampling cost.

\par
We also discuss which factor decides the purification performance by comparing different purification gadgets among existing methods and ours.
We further discuss which hardware implementation our gadget is suitable for.


\section{Preliminaries \label{sec:preliminaries}}


\subsection{SWAP Gadget and Its Recursive Application}

\par
It is widely known that the SWAP test improves the purity of input quantum states by projecting them onto their swapping symmetric subspace.
We refer to this purification procedure as ``SWAP gadget".
The quantum circuit with controlled-SWAP and post-selection on $|0\rangle\langle0|$ state in the ancillary qubit applies a projector
\begin{equation}
    P_{\mathrm{SG}} = \frac {1} {2} (P_{12} + P_{21}) = \frac {1} {2} (I^{\otimes 2} + \mathrm{SWAP})
\end{equation}
to the two input states $\rho\otimes\sigma$, where $P_{i_{1}\cdots i_{M}}$ is a permutation operation of quantum registers that maps the position $k$ to $i_{k}$.
See Appendix~\ref{sec:appendix_swap_test} for the detailed discussion.

\par
When the two inputs $\sigma$ and $\rho$ are the same, such that
\begin{equation}
\begin{split}
\sigma = \rho = \sum_{i}\lambda_{i}|\lambda_{i}\rangle\langle\lambda_{i}|,~\left(\lambda_{0}>\lambda_{1}>\lambda_{2}>\cdots\right),
\end{split}
\end{equation}
the density matrix $\tilde{\rho}$ reduced to the first register in the output (i.e. $\tilde{\rho}^{(1)}$ in Fig.~\ref{fig:qcs_rsg}(a)) and the post-selection probability $\operatorname{P}_{\vec{0}}$ can be simplified to
\begin{equation}\label{eq:outputs_SWAP}
\begin{split}
    \tilde{\rho} 
    &= \frac {1} {2\operatorname{P}_{\vec{0}}} \left(\rho + \rho^{2}\right) 
    = \sum_{i} \frac {\lambda_{i} + \lambda_{i}^{2}} {1 + \operatorname{Tr}\left[\rho^{2}\right]}|\lambda_{i}\rangle\langle\lambda_{i}| , \\
    \operatorname{P}_{\vec{0}} 
    &= \frac {1} {2} \left(1 + \operatorname{Tr}\left[\rho^{2}\right]\right).
\end{split}
\end{equation}
Since $\{\lambda_{i}\}_{i}$ is labelled in descending order, the contribution of larger eigenvalues in $\rho$ is amplified in $\tilde{\rho}$.

\par
We can assume the noise due to the system-environment interaction is stochastic noise with a small probability on quantum hardware.
In this scenario, the target state of interest in the noisy density matrix can be seen as the dominant eigenvector $|\lambda_{0}\rangle\langle\lambda_{0}|$ with a dominant probability (i.e. eigenvalue).
Hence, according to Eq.~\eqref{eq:outputs_SWAP}, the SWAP gadget can "distil" the desired quantum state from two noisy inputs by selecting out unwanted states through post-selection.

\begin{figure*}[htbp]
    \centering
    \subfloat[\label{fig:qc_rsg_recursive}]{
        \raisebox{10pt}{
            \includegraphics[width=0.45\textwidth]{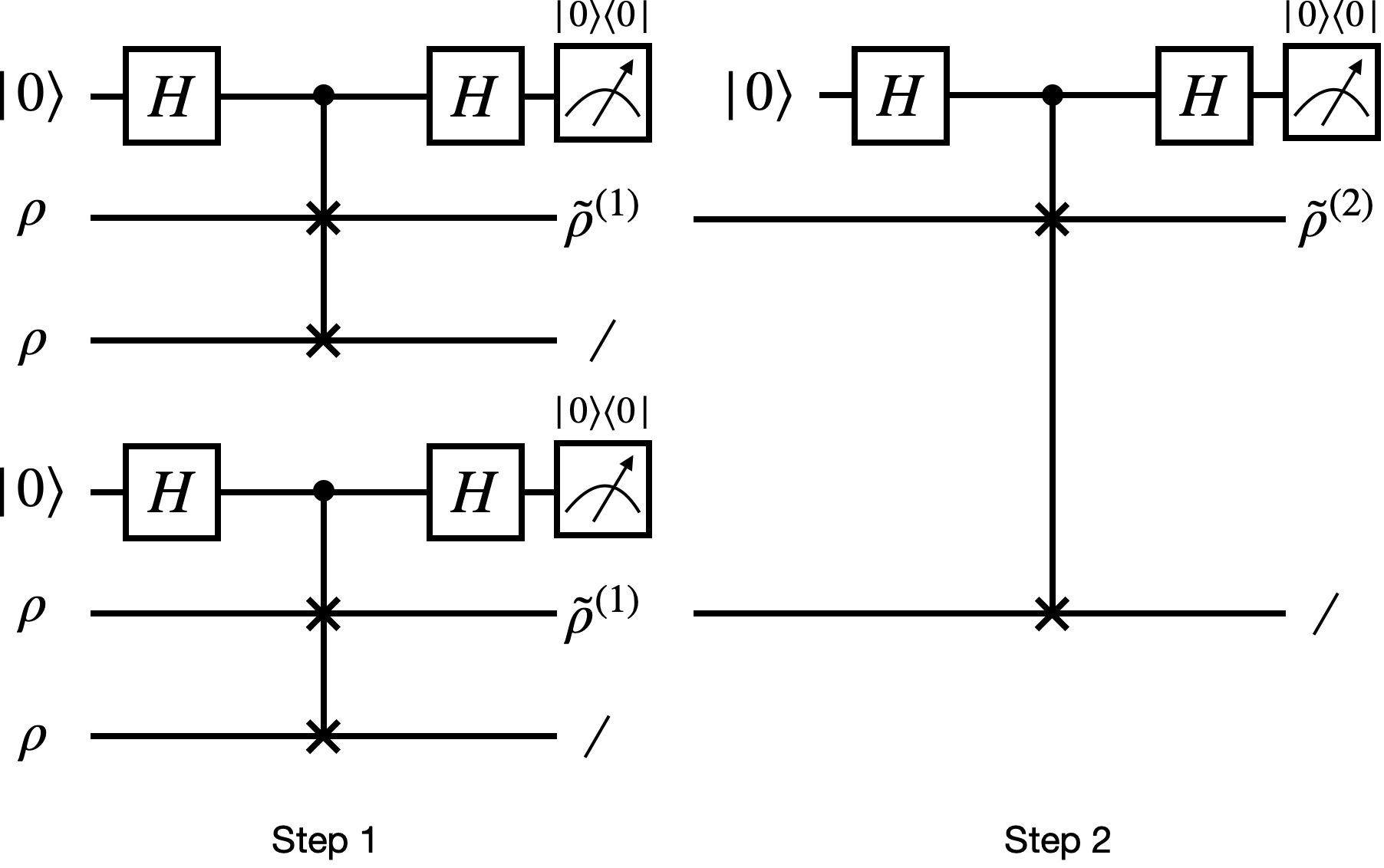}
        }
    }
    \subfloat[\label{fig:qc_rsg_one_m4}]{
        \includegraphics[width=0.45\textwidth]{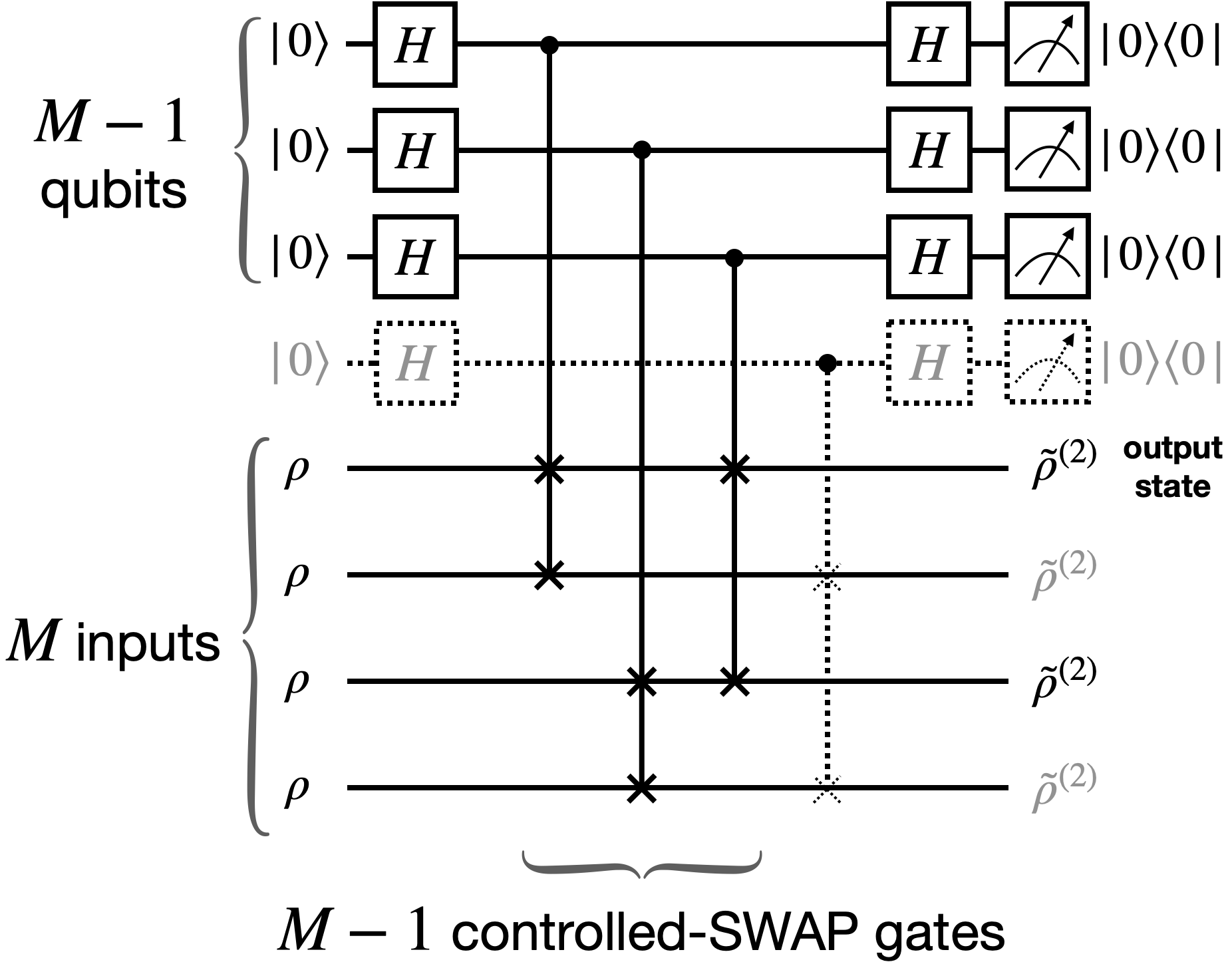}
    }
    \caption{
        (a) The quantum circuits in RSG with depth two, which outputs one purified output $\tilde{\rho}^{(2)}$ from four noisy inputs $\rho$.
        (b) The quantum circuit equivalent to (a), denoted by solid black gates and wires.
            Here, the number of copies $M$ is set to $4$.
            Note that this circuit outputs the purified state only in the quantum register of the first (and the third) copies.
            If one prefers to purify all of the copies to the same extent to $\tilde{\rho}^{(2)}$, an additional ancillary qubit is required to perform the controlled-SWAP operation between the second and the fourth copies.
            The additional quantum circuit of this overhead is shown in dotted gates and wires, along with grey labels of quantum gates and states.
    }
    \label{fig:qcs_rsg}
\end{figure*}

\par
The sequential application of the SWAP gadget will effectively purify multiple copies of noisy inputs.
Childs et al.~\cite{childs2023streaming} show that the cascaded recursive application of the SWAP gadget can keep the quantum advantage of Simon's problem with the depolarised oracle at a constant depolarising rate.
Below, we review their approach and purification performance.

\par
Let us introduce the depolarised state $\rho$ with depolarising rate $p$ as
\begin{equation}
    \rho = \left(1-p\right)\rho_{0} + p\frac{I}{d},
\end{equation}
where $\rho_{0}=|\psi\rangle\langle\psi|$, $I$ is identity matrix, and $d$ is the dimension of $\rho$, $\rho_{0}$ and $I$.
Given two inputs of $\rho$, the SWAP gadget yields $\tilde{\rho}$ as
\begin{equation} \label{eq:outputs_SWAP_depolarising_pure}
\scalemath{0.85}{
\begin{aligned}
    \tilde{\rho} 
    &= \frac {1} {2\operatorname{P}_{\vec{0}}} \left( (1-p) \left(2 - \left(1-\frac{2}{d}\right)p \right)\rho_{0} + p \left(1 + \frac{p}{d}\right) \frac{I}{d}\right), \\
    \operatorname{P}_{\vec{0}} 
    &= 1 - \left(1 - \frac{1}{d}\right) p +\frac{1}{2} \left(1 - \frac{1}{d}\right) p^{2}.
\end{aligned}
}
\end{equation}

\par
According to~\cite{childs2023streaming}, the recursive application of SWAP gadgets with depth $n$ using $M=2^{n}$ inputs yields a purified state 
\begin{equation} \label{eq:outputs_recursive_SWAP_depolarising_pure}
\begin{split}
    \tilde{\rho}^{(n)} &= \left(1-\tilde{p}^{(n)}\right)\rho_{0} + \tilde{p}^{(n)} \frac{I}{d},
\end{split}
\end{equation}
where
\begin{equation} \label{eq:p_tilde_recursive_SWAP_depolarising_pure}
\begin{split}
    \tilde{p}^{(n)} &< \frac{p}{2^{n}(1-2p)+2p} < \frac{p}{2^{n}+1}, \quad p<\frac{1}{2}.
\end{split}
\end{equation}
This means the purification rate scales linearly to the number of copies.

\par
For convenience, we refer to this recursive application of the SWAP gadget as ``Recursive SWAP Gadget (RSG)".
Childs et al. also showed that RSG is optimal in the sampling cost with respect to the expected number of copies required to suppress the depolarising error in the inputs.
To obtain one purified state with the arbitrary depolarising rate $\epsilon$, RSG consumes $O\left(p/\epsilon\right)$ noisy inputs with the depolarising rate $p$.

\par
In terms of implementation cost, distilling a single purified output through RSG from $M$ noisy copies requires $M-1$ times of measurement and post-selection to $|0\rangle\langle0|$ on ancillary qubits and $M-1$ controlled-SWAP gates in total.
Figure~\ref{fig:qcs_rsg}(a) shows specific quantum circuits of RSG with depth two, which can be seen as a single quantum circuit in Fig.~\ref{fig:qcs_rsg}(b) using the delayed measurement.
Note that the original RSG process gives only two purified states out of $M$ inputs.
With the same strategy to obtain $M$ purified states to the same purification rate as Eq.~\eqref{eq:outputs_recursive_SWAP_depolarising_pure}, one should consume $O\left(M\log\left(M\right)\right)$ amount of measurements for post-selection and $O\left(M\log\left(M\right)\right)$ amount of controlled-SWAP operations.
For $M=4$, the additional cost to obtain $4$ purified states is shown as the dotted grey ancillary qubit and controlled-SWAP operation.

\subsection{Full symmetric group Gadget \label{sec:sgg}}

\par
We can also generalise the SWAP gadget to the state projection with the full symmetric group over $M$ copies, which dates back to the work by Barenco et al.~\cite{barenco1997stabilization} during the early days of quantum error correction.
For convenience, we call their method ``Symmetric Group Gadget (SGG)".
The projector of SGG over $M$ copies includes all possible permutations in the symmetric group $S_{M}$, described as
\begin{equation} \label{eq:projector_sgg}
    P_{\mathrm{SGG}} = \frac{1}{M!}\sum_{i=0}^{M!-1}P_{\sigma_{i}},
\end{equation}
where $P_{\sigma_{i}}$ is a permutation operator corresponding to the permutation $\sigma_{i}$.

\par
Following the procedure in~\cite{barenco1997stabilization}, the five steps below complete the projection to the symmetric subspace by Eq.~\eqref{eq:projector_sgg}.
The rough sketch of the quantum circuit for SGG is shown in Fig.~\ref{fig:qc_sgg}.
\begin{enumerate}
    \item Introduce an ancillary Hilbert space with dimension $\log\left(M!\right)$ and initialise it to $|\vec{0}\rangle$.
    \item Prepare the uniform superposition state $\displaystyle \frac{1}{\sqrt{M!}}\sum_{i=0}^{M!-1}|i\rangle$.
          For $M=2^{n}$, the superposition is achieved by $H^{\otimes n}$, and for general $M$, this is achieved by QFT on the ancillary space.
    \item Apply the $i$-th permutation operation $P_{\sigma_{i}}$ in the symmetric group based on each state $|i\rangle$ in the ancillary space.
    \item Apply the inverse operation of step 2 on the ancillary space.
    \item Measure the ancillary space and post-select $|\vec{0}\rangle$ from the measurement result. 
          For the post-selected copies, save the first quantum register by taking the partial trace on all other copies, i.e., throw them away.
          The purified quantum state $\tilde{\rho}$ is then held by the first register.
\end{enumerate}

\par
The circuit construction of SGG consumes $O\left(M\log\left(M\right)\right)$ ancillary qubits in step 1.
In step 2, creating superposition in the ancillary qubits requires at worse $O\left(\left(M\log\left(M\right)\right)^{2}\right)$ gates due to the use of QFT.
Step 4 follows the same circuit overhead as step 2.
Finally, in step 3, projecting $M$ inputs to their full symmetric subspace requires $O\left(M^{2}\log\left(M\right)\right)$ controlled-SWAP gates.
In total, SGG requires $O\left(\left(M\log\left(M\right)\right)^{2} + M^{2}\log\left(M\right)\right)$ steps of non-local gate operations.

\begin{figure}[htbp]
    \includegraphics[width=0.47\textwidth]{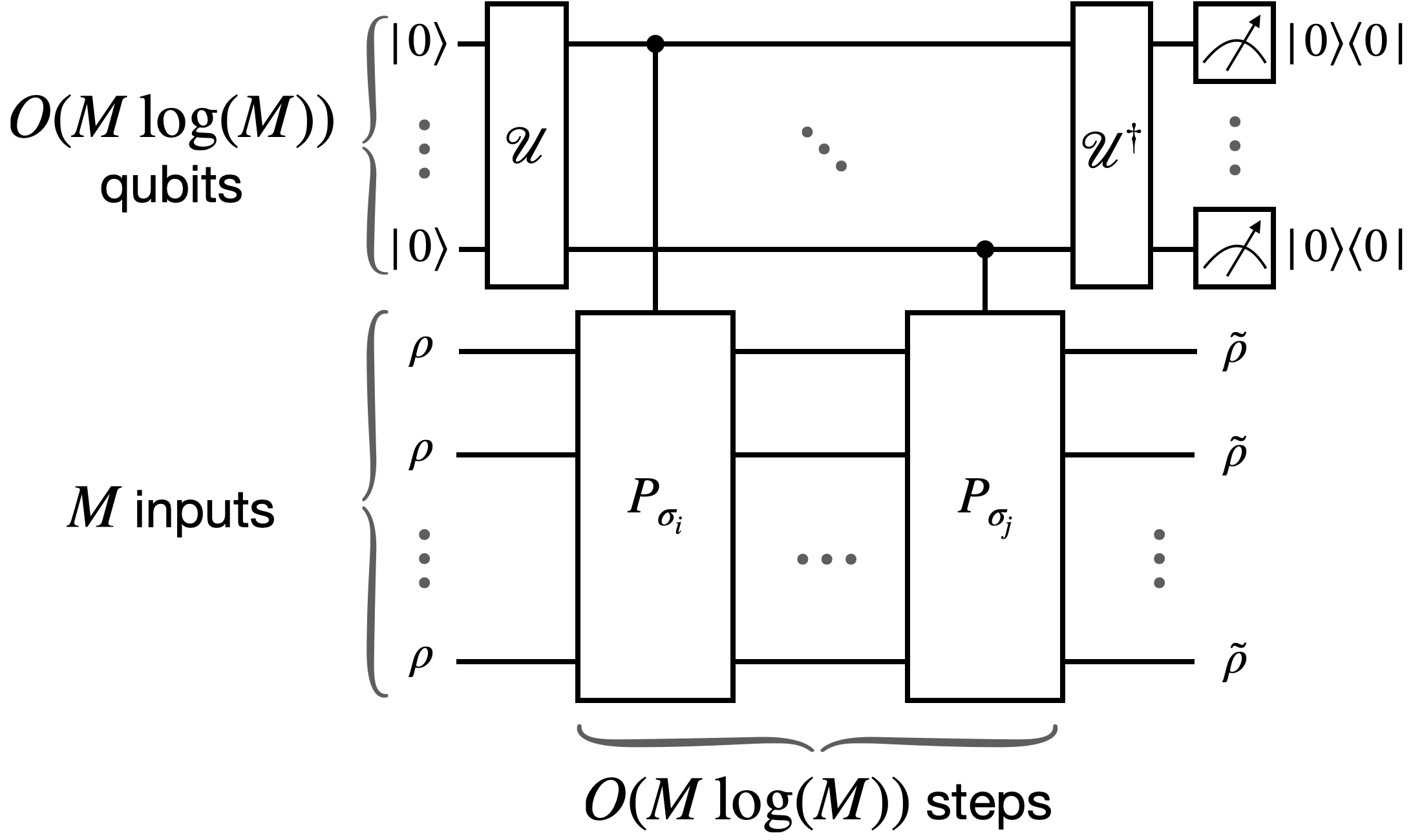}
    \caption{
        The quantum circuit implementing SGG with $M$ inputs.
        The unitary operation $\mathcal{U}$ prepares a uniform superposition state over $\log\left(M!\right)=O\left(M\log \left(M\right)\right)$ qubits.
    }
    \label{fig:qc_sgg}
\end{figure}

\par 
SGG is expected to be applied repetitively in every short-time evolution to project back the noisy evolution of $M$ redundant copies to their noise-free subspace, which is a so-called quantum watched pot or a quantum Zeno effect.
Under this usage, Barenco et al. showed that SGG suppresses coherent errors and stochastic errors into $1/M$ using $M$ copies.
The detailed discussion can be found in Appendix~\ref{sec:appendix_proofs_short_time_evolution}.

\section{Stabilisation by Symmetric Subgroup Projectors \label{sec:cgg_gsg}}

\par
As we have seen in the previous section, projecting $M$ noisy states to their symmetric subspace can linearly amplify the purity of quantum states, while the circuit implementation cost scales $O\left(\left(M\log\left(M\right)\right)^{2}\right)$.
We propose state purification gadgets using the projectors forming a subgroup of the symmetric group with order $M$ instead of the full symmetric group $S_{M}$ with order $M!$.

\subsection{Cyclic Group Gadget \label{sec:cgg}}

\par
First, we project $M$ noisy inputs to their rotation symmetric subspace.
We consider the post-selection projector as a linear combination of permutation operators in the cyclic group to amplify the contribution of higher power degrees of the density matrix.
The projector is designed by a summation of all the projectors generated by every operator in the cyclic group.
Introducing the cyclic group $C_{M}$ of order $M$ for $M$ copies of inputs, the projected unnormalised quantum state $\operatorname{P}_{\vec{0}}\tilde{\rho}$ and the post-selection probability $\operatorname{P}_{\vec{0}}$ are formalised as
\begin{equation} \label{eq:outputs_cgg}
\scalemath{0.85}{
\begin{aligned}
    \operatorname{P}_{\vec{0}} \tilde{\rho} 
    &= \operatorname{Tr}_{2\ldots M}\left[\left(\frac{1}{M}\sum_{P_{i}\in{C_{M}}}P_{i}\right) \left(\rho^{\otimes M}\right) \left(\frac{1}{M}\sum_{P_{i}\in{C_{M}}}P_{i}\right)\right] \\
    &= \frac{1}{M} \sum_{m|M} \varphi(m)\operatorname{Tr}\left[\rho^{m}\right]^{\frac{M}{m}-1}\rho^{m}, \\
    \operatorname{P}_{\vec{0}}
    &= \frac{1}{M} \sum_{m|M} \varphi(m)\operatorname{Tr}\left[\rho^{m}\right]^{\frac{M}{m}}, \\
\end{aligned}
}
\end{equation}
where $m|M$ means $m$ divides $M$, and $\varphi$ is Euler's totient function.
Euler's totient function shows up in Eq.~\eqref{eq:outputs_cgg} because, for the cyclic group with order $m$, the number of its generators corresponds to $\varphi(m)$.

\par
When $M$ is a prime number, \eqref{eq:outputs_cgg} can be simplified as
\begin{equation} \label{eq:outputs_cgg_prime}
\begin{split}
    \operatorname{P}_{\vec{0}} \tilde{\rho}
    &= \frac{1}{M}\rho + \frac{M - 1}{M} \rho^{M}, \\
    \operatorname{P}_{\vec{0}}
    &= \frac{1}{M} + \frac{M - 1}{M} \operatorname{Tr}\left[\rho^{M}\right], \\
\end{split}
\end{equation}
since $\varphi(1) = 1$, $\varphi(M) = M - 1$, and $\varphi(m) = 0$ for all natural number $1<m<M$ when $M$ is prime.

\par 
As Eq.~\eqref{eq:outputs_cgg_prime} is a weighted sum of $\rho$ and $\rho^{M}$, adding $M$ will linearly amplify the contribution of $\rho^{M}$ in which the ratio of the dominant eigenvector is amplified exponentially.
This implies when the desired state $\rho_{0}$ is the dominant eigenspace of $\rho$, we can amplify the contribution of $\rho_{0}$ with CGG.
However, we have to note that $\operatorname{Tr}\left[\rho^{M}\right]$ shrinks to $0$ also in exponential speed with $M$ when $\rho$ is a mixed state.
This implies that there is an optimal number of copies $M^{*}$ that maximises the purification rate of CGG and using more copies than $M^{*}$ will project the copies closer to the initial noisy state.
We analyse the optimal $M$ and its purification rate for the depolarised inputs in Section~\ref{sec:inputs_under_depolarisation}.

\begin{figure}[htbp]
    \includegraphics[width=\columnwidth]{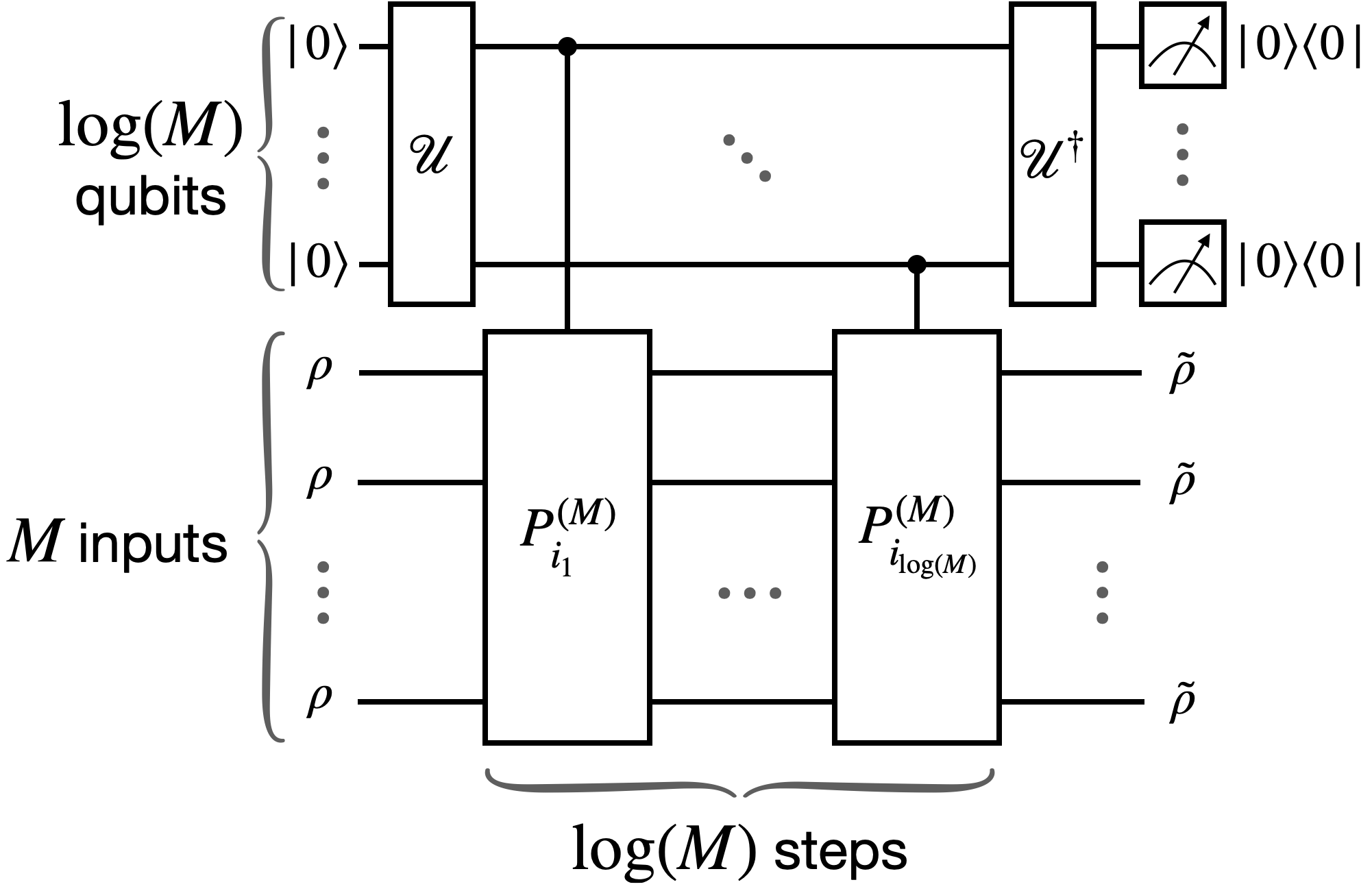}
    \hfill
    \caption{
        The quantum circuit implementing CGG taking $M$ inputs in general.
    }
    \label{fig:qc_cgg}
\end{figure}

\par
To perform CGG in the circuit picture, one can adopt a circuit in Fig.~\ref{fig:qc_cgg}, where $\mathcal{U}$ creates the uniform superposition of $M$ different states over $\log\left(M\right)$ ancillary qubits.
Note that we can reduce the number of ancillary qubits by reusing them.
In general, implementing a controlled-$P_{i}^{(M)}$ gate requires $M-1$ controlled-SWAP gates.
Given the uniform superposition over $O\left(\log\left(M\right)\right)$ ancillary qubits, it is also enough to use $O\left(\log\left(M\right)\right)$ steps of controlled-$P_{i}^{(M)}$ gates.
Therefore, the total controlled-SWAP gates required in the projection is $O\left(M\log\left(M\right)\right)$.

\par
The overhead of preparing uniform superposition on the ancillary qubits is exponentially smaller than the state projection part, which becomes more efficient than Barenco et al.~\cite{barenco1997stabilization}.
For a general $M$, the uniform superposition can also be prepared with QFT consuming $O\left(\left(\log\left(M\right)\right)^{2}\right)$ gate operations.
Therefore, the circuit implementation of CGG has at most $M^{2}$ less quantum overhead than that of~\cite{barenco1997stabilization}, requiring only $O\left(\log\left(M\right)\right)$ ancillary qubits and $O\left(M\log\left(M\right)\right)$ controlled-SWAP operations.
Note the uniform superposition on the ancillary qubits can be prepared with only a single step of local Hadamard gates when $M=2^{n}$, while in this case CGG still consumes $M$ times less quantum operations than that of~\cite{barenco1997stabilization}.

\subsection{Generalised SWAP Gadget \label{sec:gsg}}

\par 
Next, we project $M$ copies with projectors forming the subgroup $\left(\mathbb{Z}/2\mathbb{Z}\right)^{\log_{2}(M)}$, inspired by the generalised SWAP test proposed by Chabaud et al.~\cite{chabaud2018optimal}, which aims to approximate any projective measurement optimally to a given one-sided error using $M$ state inputs.
For the convenience of notation, we refer to the state purification gadget by generalised SWAP test as ``generalised SWAP gadget (GSG)".

\par 
When $M=2^{n}$, the quantum circuit of the generalised parallel SWAP operations shapes as Fig.~\ref{fig:qcs_gsg}(a), where $S_{k}$ for each $k=1,2,\ldots, n$ is the parallel swapping operation defined by
\begin{equation}
\begin{split}
    S_{k} 
    &= \bigotimes_{i \in \left[1,2^{k-1}\right]} \bigotimes_{j \in \left[1,2^{n-k}\right]}\operatorname{SWAP}\left[ (2j-2)2^{k-1} + i, \right. \\
    &\qquad\qquad\qquad\qquad\qquad\qquad~ \left. (2j-1)2^{k-1} + i \right],
\end{split}
\end{equation}
defined by the SWAP operation $\mathrm{SWAP}\left[i, j\right]$ between the $i$-th and $j$-th registers.
Note that the set $\{S_{k}\}_{k}$ is an Abelian group forming a group isomorphic to $\left(\mathbb{Z}/2\mathbb{Z}\right)^{n}$ as mentioned in~\cite{chabaud2018optimal}, and the swapping operations can also be seen as the butterfly diagram appearing in the fast Fourier transformation (FFT).

\begin{figure}[htbp]
    \centering
    \subfloat[\label{fig:qc_gsg}]{
        \includegraphics[width=0.47\textwidth]{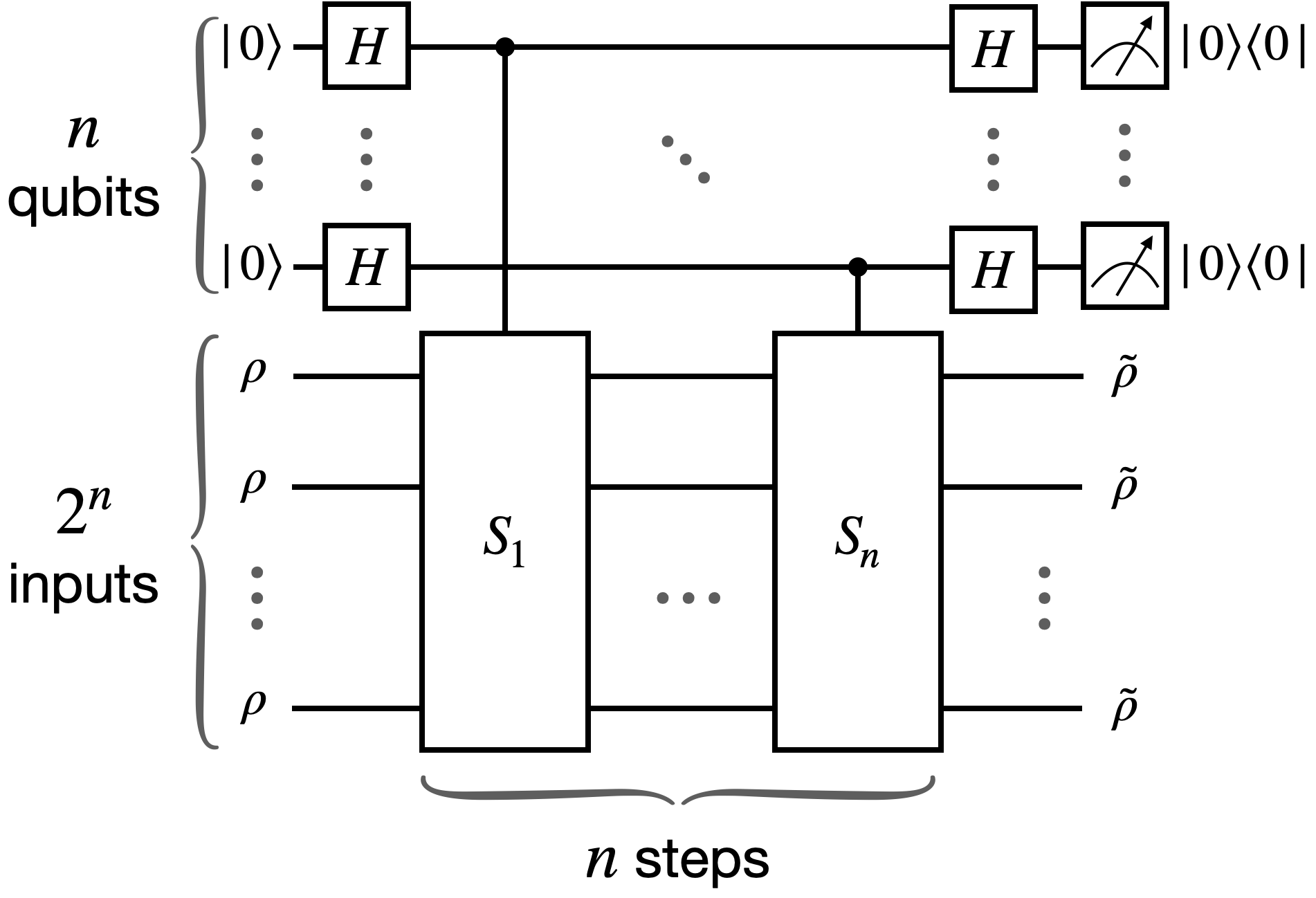}
    }
    \hfill
    \subfloat[\label{fig:qc_gsg_m4}]{
        \includegraphics[width=0.4\textwidth]{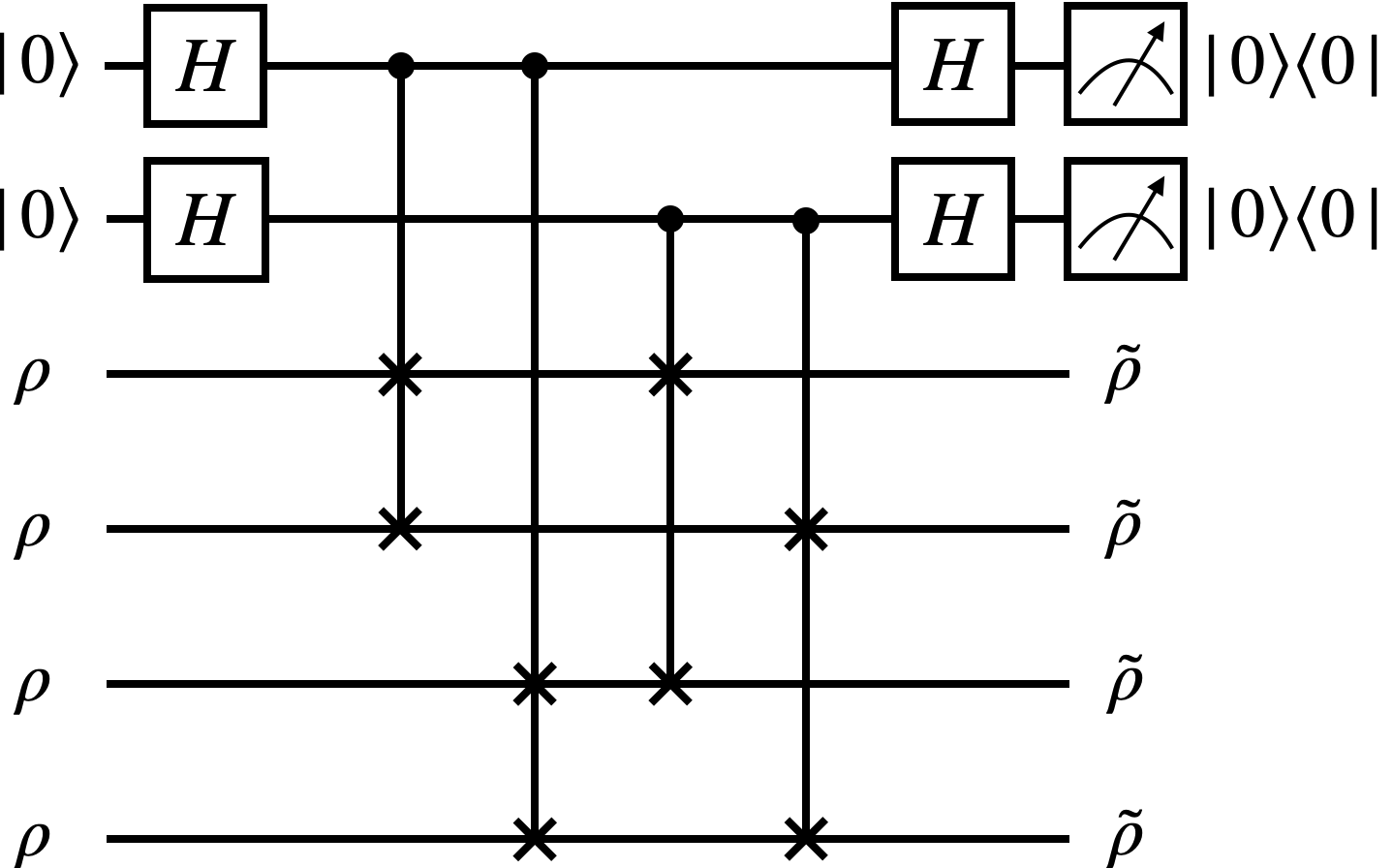}
    }
    \caption{
        (a) The quantum circuit of generalised SWAP gadget (GSG) for $M=2^{n}$ copies.
            The difference between GSG and generalised SWAP gadgets is whether to keep the first quantum register or not.
        (b) The quantum circuit of GSG for $M = 2^{n} = 4, (n=2)$. 
            $S_{1}$, i.e. $k = 1$, has $i \in \left[1,1\right]$, $j \in \left[1,2\right]$, containing SWAP operations between $(1,2)$ and $(3,4)$.
            $S_{2}$, i.e. $k = 2$, has $i \in \left[1,2\right]$, $j \in \left[1,1\right]$, containing SWAP operations between $(1,3)$ and $(2,4)$.
    }
    \label{fig:qcs_gsg}
\end{figure}

\par 
The purified output $\tilde{\rho}$ in each quantum register and the post-selection probability $\operatorname{P}_{\vec{0}}$ becomes
\begin{equation} \label{eq:outputs_gSWAP_m}
\begin{split}
    \operatorname{P}_{\vec{0}}\tilde{\rho} &= \frac {1} {M} \rho + \frac {M - 1} {M} \operatorname{Tr}\left[\rho^{2}\right]^{\frac{M}{2}-1} \rho^{2}, \\
    \operatorname{P}_{\vec{0}} &= \frac {1} {M} + \frac {M - 1} {M} \operatorname{Tr}\left[\rho^{2}\right]^{\frac{M}{2}}.
\end{split}
\end{equation}

\par
Following the original motivation of Chabaud et al.~\cite{chabaud2018optimal}, GSG optimally amplifies the state overlap in terms of the post-selection probability $\operatorname{P}_{\vec{0}}$, which becomes
\begin{equation}
    \operatorname{P}_{\vec{0}} = \frac{1}{M}+\frac{M-1}{M}|\langle\phi|\psi\rangle|^{2},
\end{equation}
given $|\phi\rangle$ as an unknown state and $M-1$ copies of $|\psi\rangle$ as the reference state.
This offers a binary test that outputs $0$ with probability $\displaystyle \frac{1}{M}+\frac{M-1}{M}|\langle\phi|\psi\rangle|^{2}$ and $1$ with probability $\displaystyle \left(\frac{M-1}{M}\right)\left(1-|\langle\phi|\psi\rangle|^{2}\right)$.
If the outcome $0$ (resp. $1$) is obtained, the test concludes that the states $|\phi\rangle$ and $|\psi\rangle$ were identical (resp. different).

\par
The quantum circuit of GSG is realised by implementing all the SWAP operations in $\{S_{k}\}_{k}$ sequentially controlled by each superposition in the ancillary qubits shown in Fig.~\ref{fig:qcs_gsg}(a).
The specific example for $M=4$ is shown in Fig.~\ref{fig:qcs_gsg}(b), which implements the swapping operations forming Klein's four-group.
The circuit implementation of~\cite{chabaud2018optimal} offers the same gate and qubit overhead as CGG, requiring $O\left(\log\left(M\right)\right)$ ancillary qubits and $O\left(M\log\left(M\right)\right)$ controlled-SWAP operations.

\section{Analysis on Purification Rate of CGG and GSG \label{sec:analysis_on_purification_rate_of_CGG_and_GSG}}

\par
In this section, we analyse how well CGG and GSG suppress the noise in input states under two scenarios: the repetitive application per each short time evolution and the single round application given noisy inputs.

\subsection{Purifying Errors in Every Short Evolution Time \label{sec:inputs_with_small_evolution_time}}

\par
We first consider applying our gadgets over $M$ redundant copies per every small time evolution $\delta t$, as depicted in Fig.~\ref{fig:qc_cgg_repetitive}.
Particularly, we analyse the noisy inputs either under the coherent noise drifted by bounded local Hamiltonian or stochastic noise generated by the system-environment interaction, yielding the following claim.

\begin{claim}
    The sequential application of cyclic group gadget (CGG) or generalised SWAP gadget (GSG) over $M$ copies in each small evolution time $\delta t$ under coherent drift suppresses the failure probability of obtaining unwanted output by a factor of $1/M$, under the first-order approximation of $\delta t$.
\end{claim}
\begin{proof}
    The proof follows the same discussion as by Barenco et al.~\cite{barenco1997stabilization}.
    When taking the first order of $\delta t$, it is enough to stabilise $M$ redundancy to their rotation symmetric subspace instead of their permutation symmetric subspace.
    The detailed noise setting and proof are given in Appendix~\ref{sec:appendix_Suppressing_Coherent_Errors_by_CGG_and_GSG}.
\end{proof}

\begin{claim}
    The sequential application of cyclic group gadget (CGG) or generalised SWAP gadget (GSG) over $M$ copies in each small evolution time $\delta t$ suppresses the stochastic error by a factor of $1/M$:
    \begin{equation} \label{eq:fidelity_rho_0_rho_tilde}
    \begin{split}
        \operatorname{Tr} \left[ \rho_{0} \tilde{\rho} \right] = 1 + \frac{1}{M} \operatorname{Tr} \left[ \rho_{0} \tilde{\sigma} \right],
    \end{split}
    \end{equation}
    for the fidelity (i.e. the success probability of obtaining the target state), and 
    \begin{equation} \label{eq:purity_rho_tilde}
    \begin{split}
        \operatorname{Tr}\left[\tilde{\rho}^{2}\right] = 1 + \frac{2}{M} \operatorname{Tr}\left[\rho_{0} \tilde{\sigma}\right],
    \end{split}
    \end{equation}
    for the purity, which are the same error suppression rates as SGG by Barenco et al.~\cite{barenco1997stabilization}.
\end{claim}
\begin{proof}
    Since we assume a small time evolution, the $i$-th noisy input by the system-environment interaction can be seen as $\rho_{i}=\rho_{0} + \sigma_{i}$ with a trace-free local Hamiltonian $\sigma_{i}$.
    This yields a noisy state spanning over $M$ inputs as
    \begin{equation} \label{eq:rho_first_order}
    \begin{split}
        \rho^{(1\ldots M)} 
        = \rho_{0}^{\otimes M}
        &+ \sum_{i=1}^{M} \rho_{0} \otimes \cdots \otimes \sigma_{i} \otimes \cdots \otimes \rho_{0} \\
        &+ O\left(\sigma_{i} \sigma_{j}\right).
    \end{split}
    \end{equation}
    The output state after applying CGG (and GSG) on Eq.~\eqref{eq:rho_first_order} with following only the first-order of $\sigma_{i}$ becomes $\tilde{\rho}$ satisfying Eq.~\eqref{eq:fidelity_rho_0_rho_tilde} and Eq.~\eqref{eq:purity_rho_tilde}.
    The detailed noise setting and proof are given in Appendix~\ref{sec:appendix_Suppressing_Stochastic_Errors_by_CGG_and_GSG}.
\end{proof}

\begin{figure}[htbp]
    \includegraphics[width=\columnwidth]{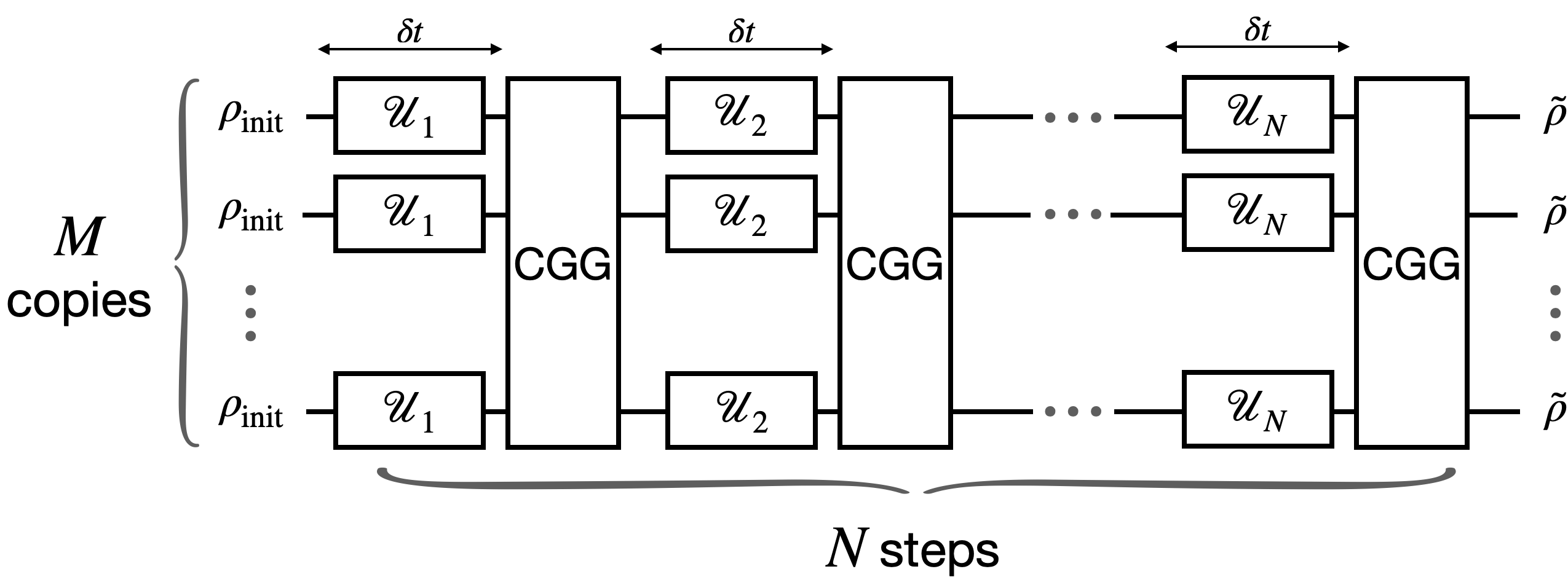}
    \hfill
    \caption{
        The schematic picture of the repetitive application of CGG.
        The quantum operation in each step is assumed to finish in $\delta t$ time, and the CGG process is assumed to be applied perfectly.
    }
    \label{fig:qc_cgg_repetitive}
\end{figure}

\par
As mentioned in~\cite{barenco1997stabilization}, the stabilisation-based state purification gadget, such as CGG, GSG, and SGG, can be used as a tool for quantum watched pot (or quantum Zeno effect)~\cite{peres1997quantum},
i.e. the linear error suppression to the redundancy by repetitively applying these gadgets in every small time $\delta t$ can effectively suppress exponentially decaying errors such as decoherence.
Let us take the dephasing error over a continuous period of time as an example.
Supposing that the probability of a dephasing channel happening per unit time $\delta t$ is $\eta$, the success probability after $N\delta t$ time (i.e. after $N$ steps) is at least $\left(1-\eta\right)^{N} \sim e^{-\eta N}$.
As the probability of obtaining an unwanted state after applying CGG or GSG is suppressed $M$ times smaller, the success probability after $N$ steps is amplified to $e^{-\eta N / M}$.
Therefore, providing $M = - \eta N / \log \left(1-\epsilon\right)$ copies, which is polynomial in $N$, is enough to keep the success probability at the level $1-\epsilon$.

\par
In conclusion, under local interaction with an external environment when the reservoir's coherence length is less than the spatial separation, it is enough to use CGG and GSG to achieve a linear suppression rate to the number of copies $M$ instead of exploiting the symmetry by the full symmetric group.
This reduces the implementation cost of quantum circuits $M$ times lighter regarding the number of ancillary qubits and the number of controlled-SWAP gates from SGG.

\subsection{Purifying Depolarised Inputs \label{sec:inputs_under_depolarisation}}

\par
It is also possible to apply our gadget only once over $M$ redundant copies with a bounded error rate.
Here, we analyse the purification rate of CGG and GSG given depolarised inputs $\rho$ with a depolarising rate $0<p<1$ and dimension $d$ described as
\begin{equation} \label{eq:rho_depolarised}
    \rho = \left(1 - p\right) \rho_{0} + p\frac{I}{d},
\end{equation}
where $\rho_{0}$ is assumed to be pure.
We see that both CGG and GSG suppress the input depolarising rate $p$ to $O\left(p^{2}\right)$ when $p$ is small.


\begin{figure*}[htbp]
    \centering
    \subfloat[\label{fig:CGG_ps-tildes_M-all_d-infty_x-log_y-log}]{
        \includegraphics[width=0.45\textwidth]{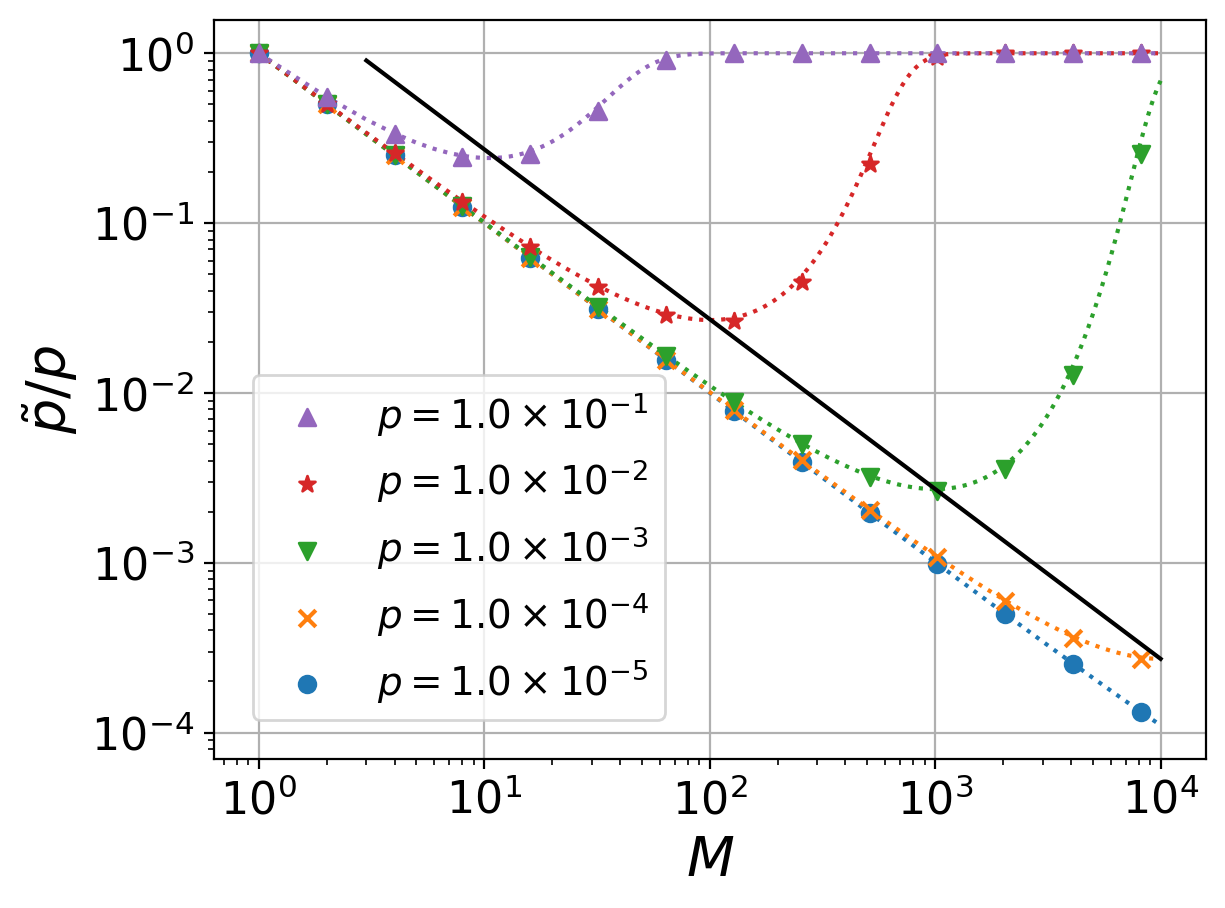}
    }
    \subfloat[\label{fig:CGG_ps-tildes_M-all_p-0.001_ds_x-log_y-log}]{
        \includegraphics[width=0.45\textwidth]{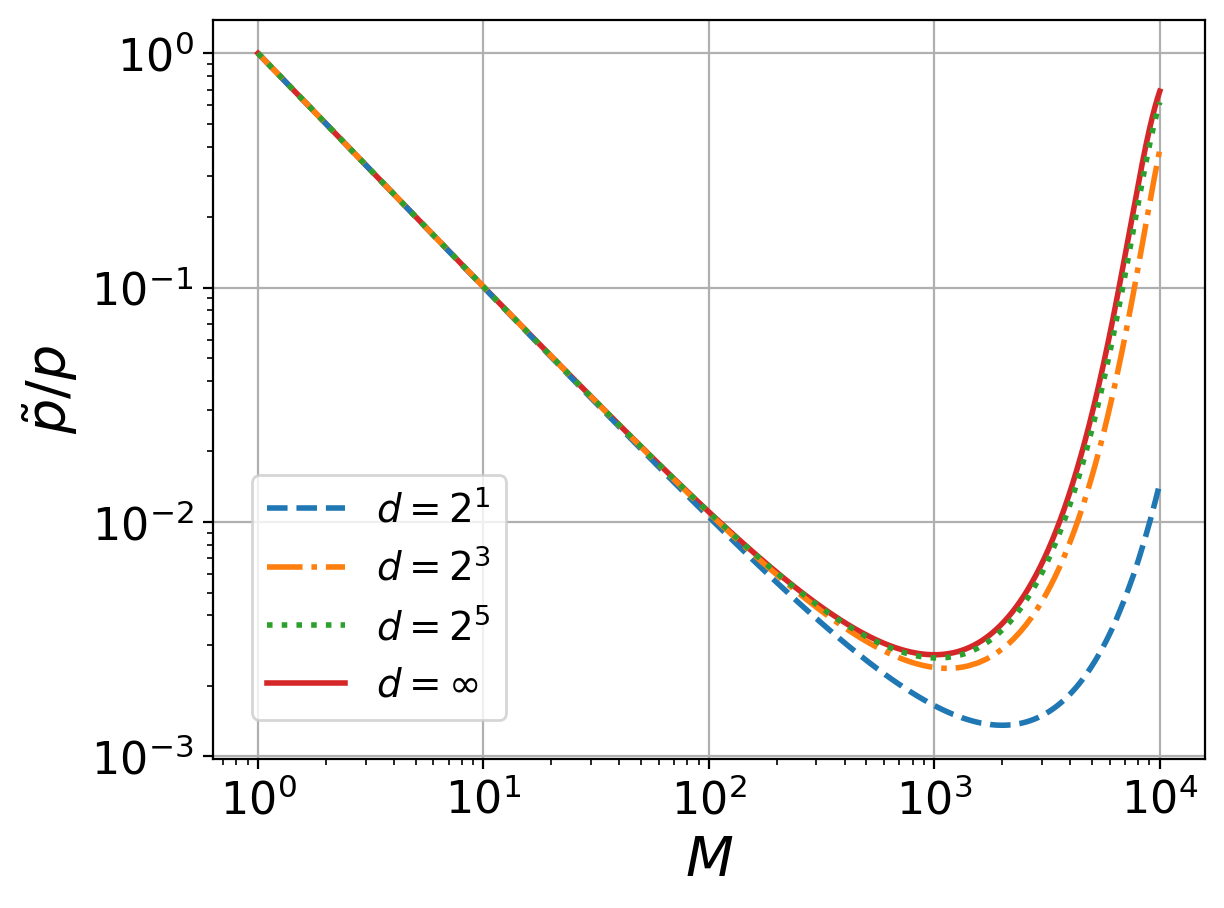}
    }
    \caption{
        The purification ratio $\tilde{p}/p$ of output depolarising rate $\tilde{p}$ to input depolarising rate $p$ using $M$ copies in CGG.
        (a) The ratio with different input depolarising probability $p$ when the state dimension is $d=32$.
            The coloured points, with symbols for different $p$, show the purification ratio for each number of copies when $M=2^{n}$, $n=0,\ldots,13$.
            The dotted curves represent the purification ratio for each input depolarising probability $p$ respectively when the dimension $d$ of input states is infinite.
            The solid black line approximates the valley of the curves according to $\displaystyle \log \left(\tilde{p}/p\right) = - \log \left(M\right) + 1$ in Eq.~\eqref{eq:func_p_tilde_and_M_linear}, deduced in Appendix~\ref{sec:appendix_proof_CGG_maximal_purification_rate_and_condition}.
        (b) The purification ratio with different state dimension $d$ when the input depolarising probability is set to $1.0\times 10^{-3}$.
    }
    \label{fig:CGG_ps-tildes_M-all}
\end{figure*}

\par
We first analyse the purification rate of CGG.
As $\rho_{0}$ is pure and commutative with the maximally mixed state $I/d$, the output state $\tilde{\rho}$ also becomes a depolarised state as
\begin{equation}
\begin{split}
    \tilde{\rho} = \left(1-\tilde{p}\right)\rho_{0} + \tilde{p}\frac{I}{d},
\end{split}
\end{equation}
letting $\tilde{p}$ denote the depolarising probability of $\tilde{\rho}$.
Then, $\tilde{p}$ can be expanded by $p$, $d$, $M$ into
\begin{equation}
\begin{split}
    \tilde{p}
    &= \frac{1}{\operatorname{P}_{\vec{0}}} \frac{1}{M} \sum_{m|M} \varphi(m)\operatorname{Tr}\left[\rho^{m}\right]^{\frac{M}{m}-1} p, \\
    \text{with } \operatorname{P}_{\vec{0}}
    &= \frac{1}{M} \sum_{m|M} \varphi(m)\operatorname{Tr}\left[\rho^{m}\right]^{\frac{M}{m}},
\end{split}
\end{equation}
where
\begin{equation}
\begin{split}
    \operatorname{Tr}\left[\rho^{m}\right] 
    &= \left(1-p+\frac{p}{d}\right)^{m} + p\left(1-\frac{1}{d}\right)\left(\frac{p}{d}\right)^{m-1}.
\end{split}
\end{equation}
As $\tilde{p} < p$ always holds for $d>1$, we see CGG has an error suppression effect.

\par
When $d$ goes to infinity, $\tilde{p}$ converges to
\begin{equation} \label{eq:p_tilde_CGG_d_inf}
\begin{split}
    \tilde{p}_{d=\infty}
    &= \frac{1}{ \displaystyle 1 + \left(M - 1\right) \left(1 - p\right)^{M} }p.
\end{split}
\end{equation}
From now on, we focus on the case when the input dimension is infinity. 



\par
To begin with, we observe that the purification rate $\tilde{p}/p$ is not suppressed monotonically to the increase of $M$.
Figure~\ref{fig:CGG_ps-tildes_M-all}(a) shows $\tilde{p}/p$ decreases at first and increases back to $1$ after a certain number of $M$.
As mentioned in Section.~\ref{sec:cgg}, an intuitive interpretation of this comes from the fact that $\operatorname{Tr}\left[\rho^{M}\right]$ shrinks to $0$ exponentially with $M$.
This appears as the trade-off between $M-1$ and $\left(1-p\right)^{M}$ in Eq.~\eqref{eq:p_tilde_CGG_d_inf}.
Therefore, there seems to be a balancing point with an optimal number of copies $M$ that maximises the purification rate $\tilde{p}/p$.

\par
We analyse this optimal $M$ and the corresponding maximal purification rate, which can be summarised in the following claim.
\begin{claim} \label{thm:claim_p_tilde_CGG}
    For a given input depolarising level $p$, CGG (and GSG) reaches its maximal purification rate
    \begin{equation} \label{eq:p_tilde_optimal_p}
    \begin{split}
        \tilde{p}_{d=\infty}^{*}
        &= \frac{ e\log\left(1-p\right) }
                { e\log\left(1-p\right) + p - 1 } p, \\
    \text{when} \quad
        M_{d=\infty}^{*} &= 1-\frac{1}{\displaystyle \log \left(1-p\right)}.
    \end{split}
    \end{equation}
\end{claim}
\begin{proof}
    We obtain $\tilde{p}_{d=\infty}^{*}$ and $M_{d=\infty}^{*}$ by analysing the monotonicity of $\tilde{p}$ to $M$, particularly, by checking the sign of $\displaystyle \frac{d}{d M} \tilde{p}$.
    We give a detailed proof in Appendix~\ref{sec:appendix_proof_CGG_maximal_purification_rate_and_condition}.
\end{proof}

\par
Since we assume our gadgets are used when the noise rate is not that high (while not negligible), we further analyse the case where we feed inputs with a small depolarising rate $p$.
Comparing CGG with RSG~\cite{childs2023streaming}, we obtain the following claim on the optimality of CGG regarding the sampling cost.
\begin{claim}
    Under the first-order approximation of $p$, the maximal purification rate $\tilde{p}_{d=\infty}^{*}$ of CGG (and GSG) scales in $O\left(p^{2}\right)$.
    The sampling cost to achieve this maximal purification rate in CGG (and GSG) scales in $O\left(p^{-1}\right)$ asymptotically, which is optimal.
\end{claim}
\begin{proof}
    We show that under the first-order approximation of $p$, the purification rate and the sampling cost scales in the same order as RSG~\cite{childs2023streaming}.
    The detailed proof is given in Appendix.~\ref{sec:appendix_proof_CGG_sampling_cost}.
\end{proof}

\par
In addition to this optimal purification rate depending on the input depolarising rate $p$, we also note that the error suppression rate still scales linearly, i.e. by a factor of $1/M$, to the number of copies $M$ until $M$ gets closer to the optimal value $M^{*}$.
Figure~\ref{fig:CGG_ps-tildes_M-all}(a) demonstrates that this linear suppression asymptotically holds until $M$ increases to around $p^{-1}/10$.
Thus, if the desired error rate in the purified output is more than ten times larger than the maximal purification rate, one can extract the linear error suppression ability from CGG.
We also observe from Fig.~\ref{fig:CGG_ps-tildes_M-all}(b) that increasing the dimension of input states makes the purification rate worse.


\par
Next, we analyse the purification rate of GSG when taking depolarised inputs.
We show that GSG has the same purification rate as CGG, particularly when the dimension of input states is infinity.
This owes to the fact that the pair-wise SWAP operations in GSG are equivalent to the rotation operations in CGG up to permutation.

\par
The equivalence in the purification rate can be checked by first assuming the depolarised inputs with dimension $d$ and then taking the infinite limit in $d$.
The output state and post-selection probability of GSG taking inputs as Eq.~\eqref{eq:rho_depolarised} becomes
\begin{equation} \label{eq:outputs_gSWAP_m_depolarising_pure}
\begin{split}
    \operatorname{P}_{\vec{0}}\tilde{\rho} 
    &= \Biggl( \frac{1}{M}\left(1-p\right) + \frac{M-1}{M} \\
    &\quad\quad \times \operatorname{Tr}\left[\rho^{2}\right]^{\frac{M}{2}-1} \left(\operatorname{Tr}\left[\rho^{2}\right] - \frac{p^{2}}{d}\right) \Biggr) \rho_{0} \\
    &\quad + \Biggl( \frac{1}{M}p + \frac{M-1}{M} \operatorname{Tr}\left[\rho^{2}\right]^{\frac{M}{2}-1} \frac{p^{2}}{d} \Biggr) \frac{I}{d}, \\
    \operatorname{P}_{\vec{0}}
    &= \frac {1} {M} + \frac {M - 1} {M} \operatorname{Tr}\left[\rho^{2}\right]^{\frac{M}{2}},
\end{split}
\end{equation}
where $\displaystyle \operatorname{Tr}\left[\rho^{2}\right] = 1-2\left(1-\frac{1}{d}\right)p+\left(1-\frac{1}{d}\right)p^{2}$.

\par
When the dimension $d$ goes to infinity in Eq.~\eqref{eq:outputs_gSWAP_m_depolarising_pure}, 
the purified depolarising probability $\tilde{p}$ becomes
\begin{equation} \label{eq:outputs_gSWAP_m_depolarising_pure_d_inf}
\begin{split}
    \tilde{p} 
    &= \frac {\displaystyle 1 + \left(M-1\right) \operatorname{Tr}\left[\rho^{2}\right]^{\frac{M}{2}-1} \frac{p}{d} } 
             {\displaystyle 1 + \left(M-1\right) \operatorname{Tr}\left[\rho^{2}\right]^{\frac{M}{2}} } p \\
    &\overset{d\rightarrow\infty}{\longrightarrow} \frac{1}{ \displaystyle 1 + \left(M - 1\right) \left(1 - p\right)^{M} }p,
\end{split}
\end{equation}
which corresponds to the depolarising probability of the output state from CGG shown in Eq.~\eqref{eq:p_tilde_CGG_d_inf}.
A more detailed discussion can be found in Appendix~\ref{sec:appendix_proof_GSG}.

\section{Discussion \label{sec:discussion}}

\begin{table*}[htbp]
    \centering
    {\scriptsize
    \begin{tabularx}{1.0\linewidth}{*{5}{X}}
        \hline
                                                      & \#(ancillary qubits)                & \#(controlled-SWAP)                     & \#(output states) & group of projectors \\
        \hline\hline 
         SGG~\cite{barenco1997stabilization}          & $O\left(M\log\left(M\right)\right)$ & $O\left(\left(M\log\left(M\right)\right)^{2}\right)$ & $M$               & $S_{M}$ \\
         RSG~\cite{childs2023streaming}               & $O\left(M\right)$                   & $O\left(M\right)$                                    & $2$               & - \\
         RSG~\cite{childs2023streaming} ($M$ outputs) & $O\left(M\right)$                   & $O\left(M\log\left(M\right)\right)$                  & $M$               & - \\
         ESD~\cite{koczor2021exponential}             & $1$                                 & $M-1$                                                & $M$               & $\{I, D_{M}\}$ \\
         CGG (proposed)                               & $O\left(\log\left(M\right)\right)$  & $O\left(M\log\left(M\right)\right)$                  & $M$               & $C_{M}$ \\
         GSG (proposed)                               & $O\left(\log\left(M\right)\right)$  & $O\left(M\log\left(M\right)\right)$                  & $M$               & $\{\mathbb{Z}/2\mathbb{Z}\}^{\log\left(M\right)}$ \\
        \hline
    \end{tabularx}
    }
    \caption{
        The comparison of different state purification protocols in their implementation cost when using $M$ copies.
        SGG stands for symmetric group gadget by Barenco et al.~\cite{barenco1997stabilization}, RSG for recursive SWAP gadget by Childs et al.~\cite{childs2023streaming}, and ESDG for error suppression by derangement by Koczor~\cite{koczor2021exponential}.
        Two cases are compared for RSG; one is for RSG to distil a single output state, which is their original circuit, and the other is to purify all $M$ state copies, which is further analysed in our work.
    }
    \label{tab:comparision}
\end{table*}

\subsection{The Factors Deciding the Purification Rate}

\par
Until the previous section, we have figured out that CGG and GSG are equivalent in their error suppression rate against several types of noise.
The suppression rate of SGG, CGG, and GSG seems to be decided by the proportion of non-identity operations in their projector.
When the permutation operations form a finite group, there is an identity projector in it.
Therefore, the proportion of non-identity projectors scales in the order of the group.
For example, since the projectors of CGG and GSG form the cyclic group $C_{M}$ and the group $\left(\mathbb{Z}/2\mathbb{Z}\right)^{\log\left(M\right)}$ respectively,
the number of non-identity projectors out of $M$ different projectors is $M-1$.
By post-selection, CGG and GSG are extracting one symmetry subspace from $M$ basis.

\par
The idea of projecting multiple state copies to their symmetry subspace based on finite groups is also used in QEM methods, i.e. exponential suppression by derangement (ESD)~\cite{koczor2021exponential} and virtual distillation (VD)~\cite{huggins2021virtual}, which are referred to as purification-based QEM.
Though these QEM methods suppress the error factors in expectation values of observables exponentially smaller with the number of copies $M$, we remark that they do not suppress errors in quantum states scaling with $M$.
Since they use a set of projectors forming group $\{I, D_{M}\}$ with order $2$, where $D_{M}$ is any derangement operator over $M$ copies, the order of this group does not depend on $M$.
The detailed discussion can be found in Appendix~\ref{sec:appendix_purification_based_qem_gadget}.

\subsection{Comparison of Implementation Cost}

\par
To further suppress the factor of the identity operator, we may have to pomp the entropy out of these $M$ copies.
When more implementation cost is allowed, one can use SGG~\cite{barenco1997stabilization} to stabilise $M$ copies into their full symmetric subspace.
Suppressing the global depolarising noise can be optimally achieved by RSG~\cite{childs2023streaming} using exponentially more ancillary qubits than our method while more moderate than SGG.
We compare the circuit implementation cost of each method in Table.~\ref{tab:comparision}.

\par
From Table.~\ref{tab:comparision}, we observe that CGG and GSG have an intermediate implementation cost between SGG and ESD.
The quantum circuits of CGG and GSG reduce the number of ancillary qubits and post-selection from $O\left(M\log\left(M\right)\right)$ to $O\left(\log\left(M\right)\right)$, and reduce the number of controlled-SWAP operations to implement permutation operations from $O\left(M\log\left(M\right)\right)$ to $O\left(M\log\left(M\right)\right)$.
The reduction of the measurement operations also implies the reduction of the cost for preparing a uniform superposition with dimension $M$, from $O\left(\left(M\log\left(M\right)\right)^{2}\right)$ to $O\left(\left(\log\left(M\right)\right)^{2}\right)$, which is smaller than $O\left(M\log\left(M\right)\right)$.
Compared to the purification-based QEM methods~\cite{koczor2021exponential, huggins2021virtual}, the implementation cost of our gadgets is larger only $O\left(\log\left(M\right)\right)$ times with respect to the number of measurements and controlled-SWAP gates.
Therefore, the implementation cost of our proposed method is still in line with that of QEM schemes, which would be feasible in the near- and middle-term quantum devices.

\par
Compared to RSG, CGG and GSG exponentially reduce the number of ancillary qubits from $O\left(M\right)$ to $O\left(\log\left(M\right)\right)$.
In terms of the number of controlled-SWAP operations, RSG consumes $O\left(M\right)$ of them to obtain a single (or at most two) purified output.
However, when it comes to obtaining $M$ purified quantum states using the same strategy as RSG, it will consume $O\left(M\log\left(M\right)\right)$ SWAP operations controlled by $O\left(\log\left(M\right)\right)$ measurements, which means post-selecting all-zero states out of $O\left(\log\left(M\right)\right)$ different states in the ancillary space.
In contrast, our proposed methods still output $M$ equivalent quantum outputs with $O\left(\log\left(M\right)\right)$ SWAP operations controlled by $O\left(M\right)$ ancillary qubits.

\subsection{Potentially Applicable Hardware}

\par
As we have seen in Section~\ref{sec:inputs_with_small_evolution_time}, CGG and GSG suppress the first-order approximation of coherent errors and stochastic errors by the factor of $1/M$ to the number of copies $M$.
This assumption can be justified when we apply the gadget over $M$ copies for every small evolution time $\delta t$, which is the same strategy as SGG~\cite{barenco1997stabilization}.
Therefore, this type of repetitive stabilisation will be particularly effective when the stabilisation process can be applied quickly with negligibly smaller errors compared to the main noisy computation task.
We still expect our gadget to be achieved before FTQC, as the repetitive application of stabilisation does not require logical qubit encoding.

\par
In a single application of our gadgets over $M$ copies, our analysis of the depolarised inputs demonstrates that the purified depolarising rate $\tilde{p}$ is also suppressed by a factor of $1/M$ up to the optimal number of $M$ associated with the input depolarising rate $p$.
We have also seen that our gadgets achieve the purification rate $O\left(p^{2}\right)$ and sampling cost $O\left(p^{-1}\right)$, which is the same as RSG~\cite{childs2023streaming} when $p$ is small. 
Thus, the sampling cost is optimal for purifying depolarised inputs.
Considering the progress of quantum devices and that both RSG and our gadgets are supposed to be performed perfectly in our analysis, the assumption of a small input depolarising rate is likely to fit the actual situation for applying these purification gadgets.

\par
A specific candidate of hardware where our method can be effectively implemented is photonic devices using Hadamard interferometers as discussed in~\cite{chabaud2018optimal}.
Recent literature also demonstrates the potential of photonic devices to perform permutation operations and nonlocal entanglement~\cite{pont2022quantifying, pont2024high}.
As the controlled operation requires nonlocal operation among multiple state copies and ancillary qubits, architectures combining trapped ions with photons~\cite{stephenson2020high, nichol2022elementary, drmota2023robust, drmota2024verifiable, zhang2024experimental} can also be a potential candidate for incorporating our method.

\subsection{Future Directions}

\par
There are many open questions to be explored in the future.
Lighter implementations and the recursive application of our methods would be an interesting direction.
In our gadgets, we are paying $O\left(M\log\left(M\right)\right)$ controlled-SWAP operations so that all of the $M$ outputs are stabilised equally, i.e. all $M$ outputs are purified.
This allows us to stabilise the errors over $M$ copies of the same computation, while in the case where we want to apply it once and retrieve only one of the purified states, it might be possible to reduce the number of controlled-SWAP gates or recursively apply the gadget by reusing all of the outputs.

\par
We can further consider combining our gadget with other quantum error mitigation or error correction methods.
We point out that the virtual channel purification (VCP) recently proposed by Liu et al.~\cite{liu2024virtual} shares a similar quantum circuit to ours, which aims to restore noise-less expectation values by stabilising noise in quantum channels instead of quantum states.
Thus, a real channel distillation consuming multiple noisy channels can be achieved straightforwardly by combining VCP and our method.
Furthermore, the error suppression methods in early FTQC have also been intensely explored~\cite{suzuki2022quantum, gonzales2023fault, piveteau2021error, lostaglio2021error, xiong2020sampling}, our approach may also find its use in designing partial quantum error correction protocols.

\par
The purification rate of our gadgets under more general errors also remains to be examined analytically and numerically.
With the first-order approximation with respect to the evolution time, we show that our gadgets suppress the coherent and stochastic errors linearly smaller to the number of copies.
The performance of CGG and GSG with inputs affected by more general noise without approximation, such as general stochastic errors or coherent errors, also remains to be examined analytically and numerically.
In Appendix~\ref{sec:appendix_general_stochastic_noise}, we give a preliminary analytical observation on the fidelity of output state taking the inputs under the general stochastic noise $\rho = \left(1-p\right) \rho_{0} + p\sigma$ with any density matrix $\sigma$.
This may come into question when we consider combining our purification gadgets with verification protocols, as discussed below.

\par
Analysing the purification rate against general errors is also of interest from the motivation to apply these purification gadgets in verification protocols.
Quantum verification protocols~\cite{gheorghiu2019verification, hayashi2015verifiable, mahadev2018classical, fitzsimons2017unconditionally, leichtle2021verifying} are designed to guarantee the correctness of results obtained from computations offloaded to untrusted devices.
They are typically described in a client-and-server model, in which a client with very limited quantum abilities delegates a computation to a powerful server, a setting likely to be encountered in a future quantum ecosystem.
While very efficient verification protocols exist for computations in BQP, any protocol that avoids prohibitive hardware overheads has only achieved weak security guarantees for computations with nondeterministic classical or quantum output.
One could hope that purification gadgets could be used to improve the quality of outputs produced by such protocols.
This requires a rigorous analysis of the structure of the errors that may affect the outputs of these verification protocols in the presence of a malicious server and to match these errors with classes of noise for which purification gadgets guaranteedly perform well.

\begin{acknowledgments}
BY acknowledges insightful and fruitful discussions with Suguru Endo from NTT.
This work received funding from the ANR research grant ANR-21-CE47-0014 (SecNISQ),
and from the European Union’s Horizon 2020 research and innovation program through the FET project PHOQUSING (“PHOtonic Quantum SamplING machine” – Grant Agreement No. 899544).
This work was supported by the Quantum Advantage Pathfinder (QAP) research programme within the UK's National Quantum Computing Centre (NQCC), and by the Quantum Computing and Simulation (QCS) Hub.
\end{acknowledgments}

\bibliographystyle{quantum}
\bibliography{main}

\begin{thebibliography}{10}

\bibitem{shor1994algorithms}
P.W. Shor.
\newblock ``Algorithms for quantum computation: discrete logarithms and factoring''.
\newblock In Proceedings 35th Annual Symposium on Foundations of Computer Science.
\newblock \href{https://dx.doi.org/10.1109/sfcs.1994.365700}{SFCS-94}. IEEE Comput. Soc. Press~(1994).

\bibitem{shor1999polynomial}
Peter~W Shor.
\newblock ``Polynomial-time algorithms for prime factorization and discrete logarithms on a quantum computer''.
\newblock \href{https://dx.doi.org/10.1137/S0097539795293172}{SIAM review {\bf 41}, 303--332}~(1999).

\bibitem{grover1996fast}
Lov~K Grover.
\newblock ``A fast quantum mechanical algorithm for database search''.
\newblock In Proceedings of the twenty-eighth annual ACM symposium on Theory of computing.
\newblock \href{https://dx.doi.org/10.1145/237814.237866}{Pages 212--219}.
\newblock ~(1996).

\bibitem{harrow2009quantum}
Aram~W Harrow, Avinatan Hassidim, and Seth Lloyd.
\newblock ``Quantum algorithm for linear systems of equations''.
\newblock \href{https://dx.doi.org/10.1103/PhysRevLett.103.150502}{Physical review letters {\bf 103}, 150502}~(2009).

\bibitem{nielsen2010quantum}
Michael~A Nielsen and Isaac~L Chuang.
\newblock ``Quantum computation and quantum information''.
\newblock \href{https://dx.doi.org/10.1017/CBO9780511976667}{Cambridge university press}. ~(2010).

\bibitem{lidar2013quantum}
Daniel~A Lidar and Todd~A Brun.
\newblock ``Quantum error correction''.
\newblock \href{https://dx.doi.org/10.1017/CBO9781139034807}{Cambridge university press}. ~(2013).

\bibitem{roffe2019quantum}
Joschka Roffe.
\newblock ``Quantum error correction: an introductory guide''.
\newblock \href{https://dx.doi.org/10.1080/00107514.2019.1667078}{Contemporary Physics {\bf 60}, 226--245}~(2019).

\bibitem{horsman2012surface}
Dominic Horsman, Austin~G Fowler, Simon Devitt, and Rodney Van~Meter.
\newblock ``Surface code quantum computing by lattice surgery''.
\newblock \href{https://dx.doi.org/10.1088/1367-2630/14/12/123011}{New Journal of Physics {\bf 14}, 123011}~(2012).

\bibitem{fowler2012surface}
Austin~G Fowler, Matteo Mariantoni, John~M Martinis, and Andrew~N Cleland.
\newblock ``Surface codes: Towards practical large-scale quantum computation''.
\newblock \href{https://dx.doi.org/10.1103/PhysRevA.86.032324}{Physical Review A {\bf 86}, 032324}~(2012).

\bibitem{fowler2018low}
Austin~G Fowler and Craig Gidney.
\newblock ``Low overhead quantum computation using lattice surgery''~(2018).

\bibitem{ni2023beating}
Zhongchu Ni, Sai Li, Xiaowei Deng, Yanyan Cai, Libo Zhang, Weiting Wang, Zhen-Biao Yang, Haifeng Yu, Fei Yan, Song Liu, et~al.
\newblock ``Beating the break-even point with a discrete-variable-encoded logical qubit''.
\newblock \href{https://dx.doi.org/10.1038/s41586-023-05784-4}{Nature {\bf 616}, 56--60}~(2023).

\bibitem{gupta2024encoding}
Riddhi~S Gupta, Neereja Sundaresan, Thomas Alexander, Christopher~J Wood, Seth~T Merkel, Michael~B Healy, Marius Hillenbrand, Tomas Jochym-O’Connor, James~R Wootton, Theodore~J Yoder, et~al.
\newblock ``Encoding a magic state with beyond break-even fidelity''.
\newblock \href{https://dx.doi.org/10.1038/s41586-023-06846-3}{Nature {\bf 625}, 259--263}~(2024).

\bibitem{temme2017error}
Kristan Temme, Sergey Bravyi, and Jay~M Gambetta.
\newblock ``Error mitigation for short-depth quantum circuits''.
\newblock \href{https://dx.doi.org/10.1103/PhysRevLett.119.180509}{Physical review letters {\bf 119}, 180509}~(2017).

\bibitem{cai2023quantum}
Zhenyu Cai, Ryan Babbush, Simon~C Benjamin, Suguru Endo, William~J Huggins, Ying Li, Jarrod~R McClean, and Thomas~E O’Brien.
\newblock ``Quantum error mitigation''.
\newblock \href{https://dx.doi.org/10.1103/RevModPhys.95.045005}{Reviews of Modern Physics {\bf 95}, 045005}~(2023).

\bibitem{koczor2021exponential}
B{\'a}lint Koczor.
\newblock ``Exponential error suppression for near-term quantum devices''.
\newblock \href{https://dx.doi.org/10.1103/PhysRevX.11.031057}{Physical Review X {\bf 11}, 031057}~(2021).

\bibitem{huggins2021virtual}
William~J Huggins, Sam McArdle, Thomas~E O’Brien, Joonho Lee, Nicholas~C Rubin, Sergio Boixo, K~Birgitta Whaley, Ryan Babbush, and Jarrod~R McClean.
\newblock ``Virtual distillation for quantum error mitigation''.
\newblock \href{https://dx.doi.org/10.1103/PhysRevX.11.041036}{Physical Review X {\bf 11}, 041036}~(2021).

\bibitem{yang2022efficient}
Bo~Yang, Rudy Raymond, and Shumpei Uno.
\newblock ``Efficient quantum readout-error mitigation for sparse measurement outcomes of near-term quantum devices''.
\newblock \href{https://dx.doi.org/10.1103/PhysRevA.106.012423}{Physical Review A {\bf 106}, 012423}~(2022).

\bibitem{yang2023dual}
Bo~Yang, Nobuyuki Yoshioka, Hiroyuki Harada, Shigeo Hakkaku, Yuuki Tokunaga, Hideaki Hakoshima, Kaoru Yamamoto, and Suguru Endo.
\newblock ``Dual-gse: Resource-efficient generalized quantum subspace expansion''~(2023).

\bibitem{cerezo2021variational}
Marco Cerezo, Andrew Arrasmith, Ryan Babbush, Simon~C Benjamin, Suguru Endo, Keisuke Fujii, Jarrod~R McClean, Kosuke Mitarai, Xiao Yuan, Lukasz Cincio, et~al.
\newblock ``Variational quantum algorithms''.
\newblock \href{https://dx.doi.org/10.1038/s42254-021-00348-9}{Nature Reviews Physics {\bf 3}, 625--644}~(2021).

\bibitem{endo2021hybrid}
Suguru Endo, Zhenyu Cai, Simon~C Benjamin, and Xiao Yuan.
\newblock ``Hybrid quantum-classical algorithms and quantum error mitigation''.
\newblock \href{https://dx.doi.org/10.7566/JPSJ.90.032001}{Journal of the Physical Society of Japan {\bf 90}, 032001}~(2021).

\bibitem{peng2020simulating}
Tianyi Peng, Aram~W Harrow, Maris Ozols, and Xiaodi Wu.
\newblock ``Simulating large quantum circuits on a small quantum computer''.
\newblock \href{https://dx.doi.org/10.1103/PhysRevLett.125.150504}{Physical review letters {\bf 125}, 150504}~(2020).

\bibitem{yuan2021quantum}
Xiao Yuan, Jinzhao Sun, Junyu Liu, Qi~Zhao, and You Zhou.
\newblock ``Quantum simulation with hybrid tensor networks''.
\newblock \href{https://dx.doi.org/10.1103/PhysRevLett.127.040501}{Physical Review Letters {\bf 127}, 040501}~(2021).

\bibitem{harada2023noise}
Hiroyuki Harada, Yasunari Suzuki, Bo~Yang, Yuuki Tokunaga, and Suguru Endo.
\newblock ``Noise propagation in hybrid tensor networks''~(2023).

\bibitem{eddins2022doubling}
Andrew Eddins, Mario Motta, Tanvi~P Gujarati, Sergey Bravyi, Antonio Mezzacapo, Charles Hadfield, and Sarah Sheldon.
\newblock ``Doubling the size of quantum simulators by entanglement forging''.
\newblock \href{https://dx.doi.org/10.1103/PRXQuantum.3.010309}{PRX Quantum {\bf 3}, 010309}~(2022).

\bibitem{takagi2022fundamental}
Ryuji Takagi, Suguru Endo, Shintaro Minagawa, and Mile Gu.
\newblock ``Fundamental limits of quantum error mitigation''.
\newblock \href{https://dx.doi.org/10.1038/s41534-022-00618-z}{npj Quantum Information {\bf 8}, 114}~(2022).

\bibitem{takagi2023universal}
Ryuji Takagi, Hiroyasu Tajima, and Mile Gu.
\newblock ``Universal sampling lower bounds for quantum error mitigation''.
\newblock \href{https://dx.doi.org/10.1103/PhysRevLett.131.210602}{Physical Review Letters {\bf 131}, 210602}~(2023).

\bibitem{tsubouchi2023universal}
Kento Tsubouchi, Takahiro Sagawa, and Nobuyuki Yoshioka.
\newblock ``Universal cost bound of quantum error mitigation based on quantum estimation theory''.
\newblock \href{https://dx.doi.org/10.1103/PhysRevLett.131.210601}{Physical Review Letters {\bf 131}, 210601}~(2023).

\bibitem{quek2022exponentially}
Yihui Quek, Daniel~Stilck Fran{\c{c}}a, Sumeet Khatri, Johannes~Jakob Meyer, and Jens Eisert.
\newblock ``Exponentially tighter bounds on limitations of quantum error mitigation''~(2022).

\bibitem{bennett1996purification}
Charles~H Bennett, Gilles Brassard, Sandu Popescu, Benjamin Schumacher, John~A Smolin, and William~K Wootters.
\newblock ``Purification of noisy entanglement and faithful teleportation via noisy channels''.
\newblock \href{https://dx.doi.org/10.1103/PhysRevLett.76.722}{Physical review letters {\bf 76}, 722}~(1996).

\bibitem{cirac1999optimal}
J~Ignacio Cirac, AK~Ekert, and Chiara Macchiavello.
\newblock ``Optimal purification of single qubits''.
\newblock \href{https://dx.doi.org/10.1103/PhysRevLett.82.4344}{Physical review letters {\bf 82}, 4344}~(1999).

\bibitem{keyl2001rate}
Michael Keyl and Reinhard~F Werner.
\newblock ``The rate of optimal purification procedures''.
\newblock In Annales Henri Poincare.
\newblock \href{https://dx.doi.org/10.1007/PL00001027}{Volume~2, pages 1--26}.
\newblock Springer~(2001).

\bibitem{barenco1997stabilization}
Adriano Barenco, Andre Berthiaume, David Deutsch, Artur Ekert, Richard Jozsa, and Chiara Macchiavello.
\newblock ``Stabilization of quantum computations by symmetrization''.
\newblock \href{https://dx.doi.org/10.1137/S0097539796302452}{SIAM Journal on Computing {\bf 26}, 1541--1557}~(1997).

\bibitem{chabaud2018optimal}
Ulysse Chabaud, Eleni Diamanti, Damian Markham, Elham Kashefi, and Antoine Joux.
\newblock ``Optimal quantum-programmable projective measurement with linear optics''.
\newblock \href{https://dx.doi.org/10.1103/PhysRevA.98.062318}{Physical Review A {\bf 98}, 062318}~(2018).

\bibitem{childs2023streaming}
Andrew~M Childs, Honghao Fu, Debbie Leung, Zhi Li, Maris Ozols, and Vedang Vyas.
\newblock ``Streaming quantum state purification''~(2023).

\bibitem{peres1997quantum}
Asher Peres.
\newblock ``Quantum theory: concepts and methods''.
\newblock \href{https://dx.doi.org/10.1007/0-306-47120-5}{Volume~72}.
\newblock Springer. ~(1997).

\bibitem{pont2022quantifying}
Mathias Pont, Riccardo Albiero, Sarah~E Thomas, Nicol{\`o} Spagnolo, Francesco Ceccarelli, Giacomo Corrielli, Alexandre Brieussel, Niccolo Somaschi, H{\^e}lio Huet, Abdelmounaim Harouri, et~al.
\newblock ``Quantifying n-photon indistinguishability with a cyclic integrated interferometer''.
\newblock \href{https://dx.doi.org/10.1103/PhysRevX.12.031033}{Physical Review X {\bf 12}, 031033}~(2022).

\bibitem{pont2024high}
Mathias Pont, Giacomo Corrielli, Andreas Fyrillas, Iris Agresti, Gonzalo Carvacho, Nicolas Maring, Pierre-Emmanuel Emeriau, Francesco Ceccarelli, Ricardo Albiero, Paulo~Henrique Dias~Ferreira, et~al.
\newblock ``High-fidelity four-photon ghz states on chip''.
\newblock \href{https://dx.doi.org/10.1038/s41534-024-00830-z}{npj Quantum Information {\bf 10}, 50}~(2024).

\bibitem{stephenson2020high}
LJ~Stephenson, DP~Nadlinger, BC~Nichol, S~An, P~Drmota, TG~Ballance, K~Thirumalai, JF~Goodwin, DM~Lucas, and CJ~Ballance.
\newblock ``High-rate, high-fidelity entanglement of qubits across an elementary quantum network''.
\newblock \href{https://dx.doi.org/10.1103/PhysRevLett.124.110501}{Physical review letters {\bf 124}, 110501}~(2020).

\bibitem{nichol2022elementary}
BC~Nichol, R~Srinivas, DP~Nadlinger, P~Drmota, D~Main, G~Araneda, CJ~Ballance, and DM~Lucas.
\newblock ``An elementary quantum network of entangled optical atomic clocks''.
\newblock \href{https://dx.doi.org/10.1038/s41586-022-05088-z}{Nature {\bf 609}, 689--694}~(2022).

\bibitem{drmota2023robust}
Peter Drmota, Dougal Main, DP~Nadlinger, BC~Nichol, MA~Weber, EM~Ainley, Ayush Agrawal, Raghavendra Srinivas, Gabriel Araneda, CJ~Ballance, et~al.
\newblock ``Robust quantum memory in a trapped-ion quantum network node''.
\newblock \href{https://dx.doi.org/10.1103/PhysRevLett.130.090803}{Physical Review Letters {\bf 130}, 090803}~(2023).

\bibitem{drmota2024verifiable}
P~Drmota, DP~Nadlinger, D~Main, BC~Nichol, EM~Ainley, Dominik Leichtle, A~Mantri, Elham Kashefi, R~Srinivas, G~Araneda, et~al.
\newblock ``Verifiable blind quantum computing with trapped ions and single photons''.
\newblock \href{https://dx.doi.org/10.1103/PhysRevLett.132.150604}{Physical Review Letters {\bf 132}, 150604}~(2024).

\bibitem{zhang2024experimental}
Ting Zhang, Yukun Zhang, Lu~Liu, Xiao-Xu Fang, Qian-Xi Zhang, Xiao Yuan, and He~Lu.
\newblock ``Experimental virtual distillation of entanglement and coherence''.
\newblock \href{https://dx.doi.org/10.1103/PhysRevLett.132.180201}{Physical Review Letters {\bf 132}, 180201}~(2024).

\bibitem{liu2024virtual}
Zhenhuan Liu, Xingjian Zhang, Yue-Yang Fei, and Zhenyu Cai.
\newblock ``Virtual channel purification''~(2024).

\bibitem{suzuki2022quantum}
Yasunari Suzuki, Suguru Endo, Keisuke Fujii, and Yuuki Tokunaga.
\newblock ``Quantum error mitigation as a universal error reduction technique: Applications from the nisq to the fault-tolerant quantum computing eras''.
\newblock \href{https://dx.doi.org/10.1103/PRXQuantum.3.010345}{PRX Quantum {\bf 3}, 010345}~(2022).

\bibitem{gonzales2023fault}
Alvin Gonzales, Anjala~M Babu, Ji~Liu, Zain Saleem, and Mark Byrd.
\newblock ``Fault tolerant quantum error mitigation''~(2023).

\bibitem{piveteau2021error}
Christophe Piveteau, David Sutter, Sergey Bravyi, Jay~M Gambetta, and Kristan Temme.
\newblock ``Error mitigation for universal gates on encoded qubits''.
\newblock \href{https://dx.doi.org/10.1103/PhysRevLett.127.200505}{Physical review letters {\bf 127}, 200505}~(2021).

\bibitem{lostaglio2021error}
Matteo Lostaglio and Alessandro Ciani.
\newblock ``Error mitigation and quantum-assisted simulation in the error corrected regime''.
\newblock \href{https://dx.doi.org/10.1103/PhysRevLett.127.200506}{Physical review letters {\bf 127}, 200506}~(2021).

\bibitem{xiong2020sampling}
Yifeng Xiong, Daryus Chandra, Soon~Xin Ng, and Lajos Hanzo.
\newblock ``Sampling overhead analysis of quantum error mitigation: Uncoded vs. coded systems''.
\newblock \href{https://dx.doi.org/10.1109/ACCESS.2020.3045016}{IEEE Access {\bf 8}, 228967--228991}~(2020).

\bibitem{gheorghiu2019verification}
Alexandru Gheorghiu, Theodoros Kapourniotis, and Elham Kashefi.
\newblock ``Verification of quantum computation: An overview of existing approaches''.
\newblock \href{https://dx.doi.org/10.1007/s00224-018-9872-3}{Theory of computing systems {\bf 63}, 715--808}~(2019).

\bibitem{hayashi2015verifiable}
Masahito Hayashi and Tomoyuki Morimae.
\newblock ``Verifiable measurement-only blind quantum computing with stabilizer testing''.
\newblock \href{https://dx.doi.org/10.1103/PhysRevLett.115.220502}{Physical review letters {\bf 115}, 220502}~(2015).

\bibitem{mahadev2018classical}
Urmila Mahadev.
\newblock ``Classical verification of quantum computations''.
\newblock In 2018 IEEE 59th Annual Symposium on Foundations of Computer Science (FOCS).
\newblock \href{https://dx.doi.org/10.1137/20M1371828}{Pages 259--267}.
\newblock IEEE~(2018).

\bibitem{fitzsimons2017unconditionally}
Joseph~F Fitzsimons and Elham Kashefi.
\newblock ``Unconditionally verifiable blind quantum computation''.
\newblock \href{https://dx.doi.org/https://doi.org/10.1103/PhysRevA.96.012303}{Physical Review A {\bf 96}, 012303}~(2017).

\bibitem{leichtle2021verifying}
Dominik Leichtle, Luka Music, Elham Kashefi, and Harold Ollivier.
\newblock ``Verifying bqp computations on noisy devices with minimal overhead''.
\newblock \href{https://dx.doi.org/10.1103/PRXQuantum.2.040302}{PRX Quantum {\bf 2}, 040302}~(2021).

\end{thebibliography}

\onecolumn

\appendix

\section{Outputs of the SWAP Test \label{sec:appendix_swap_test}}

\begin{figure}
    \centering
    \subfloat[\label{fig:qc_sg}]{
        \raisebox{10pt}{
            \includegraphics[width=0.3\textwidth]{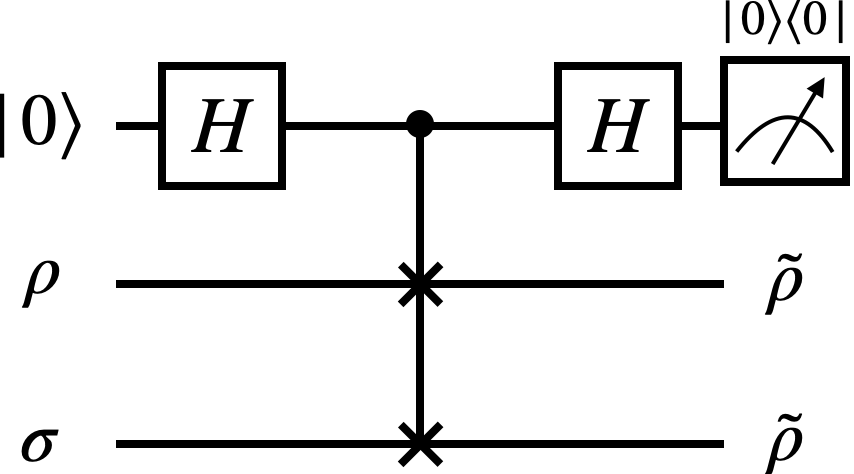}
        }
    }
    \hfill
    \subfloat[\label{fig:tn_sg_loop}]{
        \includegraphics[width=0.25\textwidth]{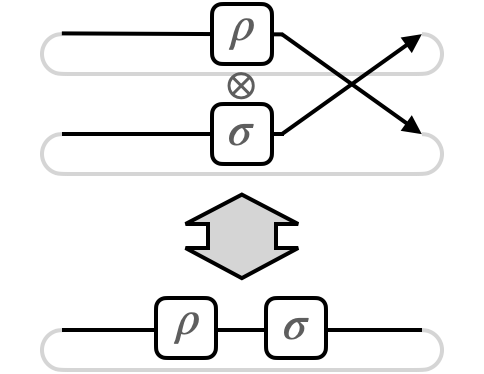}
    }
    \hfill
    \subfloat[\label{fig:tn_sg_line}]{
        \includegraphics[width=0.25\textwidth]{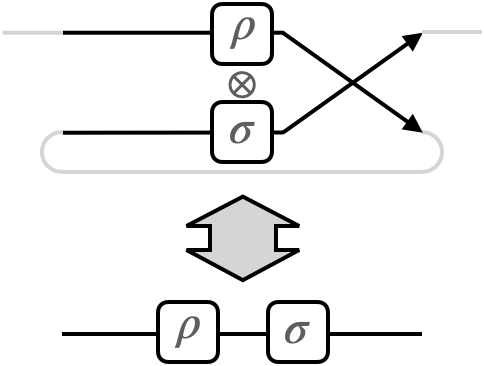}
    }
    \caption{
        (a) The quantum circuit for the SWAP gadget.
            The output state in each wire is denoted by $\tilde{\rho}$.
        (b, c) The tensor network representation of 
        (b): $\operatorname{Tr}\left[\left(\rho\otimes\sigma\right) P_{21}\right] = \operatorname{Tr}\left[\rho\sigma\right]$ and 
        (c): $\operatorname{Tr}_{2}\left[\left(\rho\otimes\sigma\right) P_{21}\right] = \rho\sigma$.
        Taking the partial trace of a quantum register corresponds to drawing a loop around its tensor.
    }
    \label{fig:qcs_tns_sg}
\end{figure}

\par
The quantum circuit of the SWAP gadget is the same as that of the SWAP test, which is described as Fig.~\ref{fig:qcs_tns_sg}(a).
Since this circuit projects the two inputs $\rho$ and $\sigma$ to their swapping invariant subspace, the reduced output in each wire (i.e. register) is the same quantum state.
We let $\tilde{\rho}$ denote this output.
Then, the post-selection probability $\operatorname{P}_{0}$ and the output quantum state $\tilde{\rho}$ can be described as
\begin{equation} \label{eq:outputs_SWAP_rho_sigma}
\begin{split}
    \tilde{\rho} 
    &= \frac {1} {4\operatorname{P}_{0}} \operatorname{Tr}_{2}\left[(P_{12} + P_{21}) \left(\rho\otimes\sigma\right) (P_{12} + P_{21})\right]
    = \frac {1} {4\operatorname{P}_{0}} \operatorname{Tr}_{2}\left[\rho\otimes\sigma + \left(\rho\otimes\sigma\right) P_{21} + P_{21} \left(\rho\otimes\sigma\right) + \sigma\otimes\rho\right] \\
    &= \frac {1} {4\operatorname{P}_{0}} \left(\rho + \operatorname{Tr}_{2}\left[\left(\rho\otimes\sigma\right) P_{21}\right] + \operatorname{Tr}_{2}\left[\left(\rho\otimes\sigma\right) P_{21}\right] + \sigma \right)
    = \frac {1} {4\operatorname{P}_{0}} \left(\rho + \rho\sigma + \sigma\rho + \sigma \right), \\
    \operatorname{P}_{0} 
    &= \operatorname{Tr}\left[\frac {1} {4} \left(\rho + \rho\sigma + \sigma\rho + \sigma \right)\right] = \frac {1} {2} \left(1 + \operatorname{Tr}\left[\rho\sigma\right]\right).
\end{split}
\end{equation}
To obtain the last transform for $\tilde{\rho}$, we use the identity $\operatorname{Tr}_{2}\left[\left(\rho\otimes\sigma\right) P_{21}\right] = \rho\sigma$, which can be checked as below.
Assuming $\displaystyle\rho = \sum_{i}\lambda_{i}|\lambda_{i}\rangle\langle\lambda_{i}|$ and $\displaystyle\sigma = \sum_{j}\lambda_{j}^{'}|\lambda_{j}^{'}\rangle\langle\lambda_{j}^{'}|$,
\begin{equation}
\begin{split}
    \operatorname{Tr}_{2}\left[\left(\rho\otimes\sigma\right) P_{21}\right]
    &= \operatorname{Tr}_{2}\left[\left(\sum_{i,j}\lambda_{i}\lambda_{j}^{'}|\lambda_{i}\rangle\langle\lambda_{i}|\otimes|\lambda_{j}^{'}\rangle\langle\lambda_{j}^{'}|\right) P_{21}\right] 
    = \operatorname{Tr}_{2}\left[\sum_{i,j}\lambda_{i}\lambda_{j}^{'}|\lambda_{i}\rangle\langle\lambda_{j}^{'}|\otimes|\lambda_{j}^{'}\rangle\langle\lambda_{i}|\right] \\
    &= \sum_{i,j}\lambda_{i}\lambda_{j}^{'}|\lambda_{i}\rangle\langle\lambda_{j}^{'}|\operatorname{Tr}_{2}\left[|\lambda_{j}^{'}\rangle\langle\lambda_{i}|\right] 
    = \sum_{i,j}\lambda_{i}\lambda_{j}^{'}|\lambda_{i}\rangle\langle\lambda_{j}^{'}|\langle\lambda_{i}|\lambda_{j}^{'}\rangle 
    = \sum_{i,j}\lambda_{i}\lambda_{j}^{'}|\lambda_{i}\rangle\langle\lambda_{i}|\lambda_{j}^{'}\rangle\langle\lambda_{j}^{'}| \\
    &= \left(\sum_{i}\lambda_{i}|\lambda_{i}\rangle\langle\lambda_{i}|\right)\left(\sum_{j}\lambda_{j}^{'}|\lambda_{j}^{'}\rangle\langle\lambda_{j}^{'}|\right)
    = \rho\sigma.
\end{split}
\end{equation}
This identity can also be easily checked in a tensor network representation, as shown in Fig.~\ref{fig:qcs_tns_sg}(b) and Fig.~\ref{fig:qcs_tns_sg}(c).
Following the same deduction as Eq.~\eqref{eq:outputs_SWAP_rho_sigma}, we also obtain the following rule:
\begin{equation} \label{eq:outputs_SWAP_powers_k_l}
\begin{split}
    \mathcal{E}_{\mathrm{SG}}\left(\rho^{k}\otimes\rho^{l}\right)
    &\sim \operatorname{Tr}_{2}\left[ P_{\mathrm{SG}} \left(\rho^{k}\otimes\rho^{l}\right) P_{\mathrm{SG}}^{\dagger} \right] 
    = \frac{1}{4} \left( \operatorname{Tr}_{2}\left[ \rho^{l} \right] \rho^{k} + \operatorname{Tr}_{2}\left[ \rho^{k} \right] \rho^{l} + 2\rho^{k+l}\right).
\end{split}
\end{equation}

\section{\texorpdfstring{Example of CGG when $M=8$}{} \label{sec:appendix_example_CGG}}

\par
In Section~\ref{sec:cgg_gsg}, we have seen that the general form of each output state from CGG becomes
\begin{equation} \label{eq:rho_cgg_appendix}
\begin{split}
    \operatorname{P}_{\vec{0}} \tilde{\rho}
    &= \operatorname{Tr}_{2\ldots M}\left[\left(\frac{1}{M}\sum_{P_{i}\in{C_{M}}}P_{i}\right) \left(\rho^{\otimes M}\right) \left(\frac{1}{M}\sum_{P_{i}\in{C_{M}}}P_{i}\right)\right]
     = \operatorname{Tr}_{2\ldots M}\left[\left(\rho^{\otimes M}\right) \left(\frac{1}{M}\sum_{P_{i}\in{C_{M}}}P_{i}\right)\right] \\
    &= \frac{1}{M} \sum_{m|M} \varphi(m)\operatorname{Tr}\left[\rho^{m}\right]^{\frac{M}{m}-1}\rho^{m}.
\end{split}
\end{equation}
This operation is applying the projector $\displaystyle P_{\mathrm{CGG}}=\frac{1}{M}\sum_{P_{i}\in{C_{M}}}P_{i}$ to the input states $\rho^{\otimes M}$.

\begin{figure}[htbp]
    \centering
    \subfloat[\label{fig:qc_cgg_m8}]{
        \includegraphics[width=0.30\textwidth]{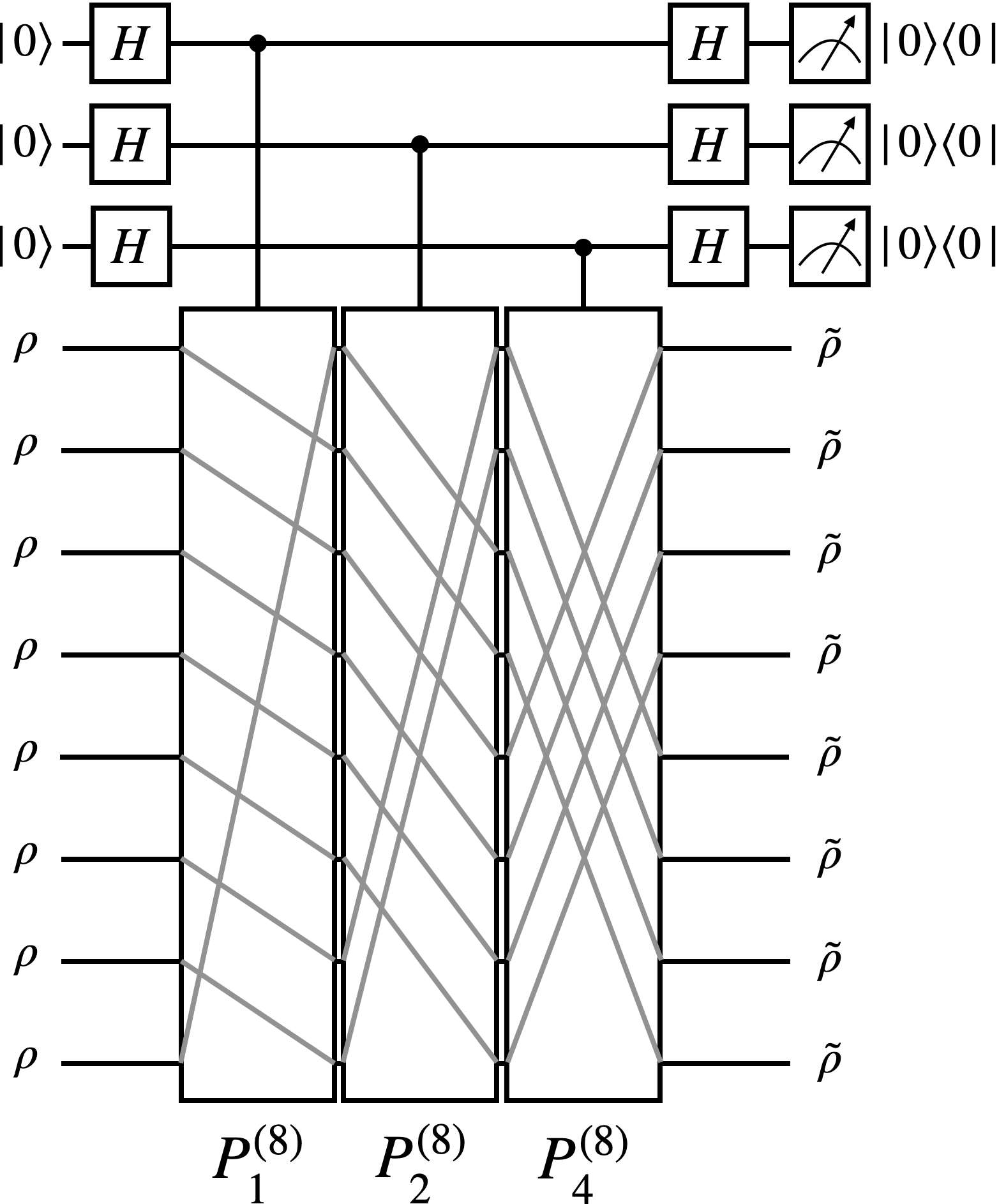}
    }
    \subfloat[\label{fig:qcs_cP_1^4}]{
        \raisebox{20pt}{
            \includegraphics[width=0.25\textwidth]{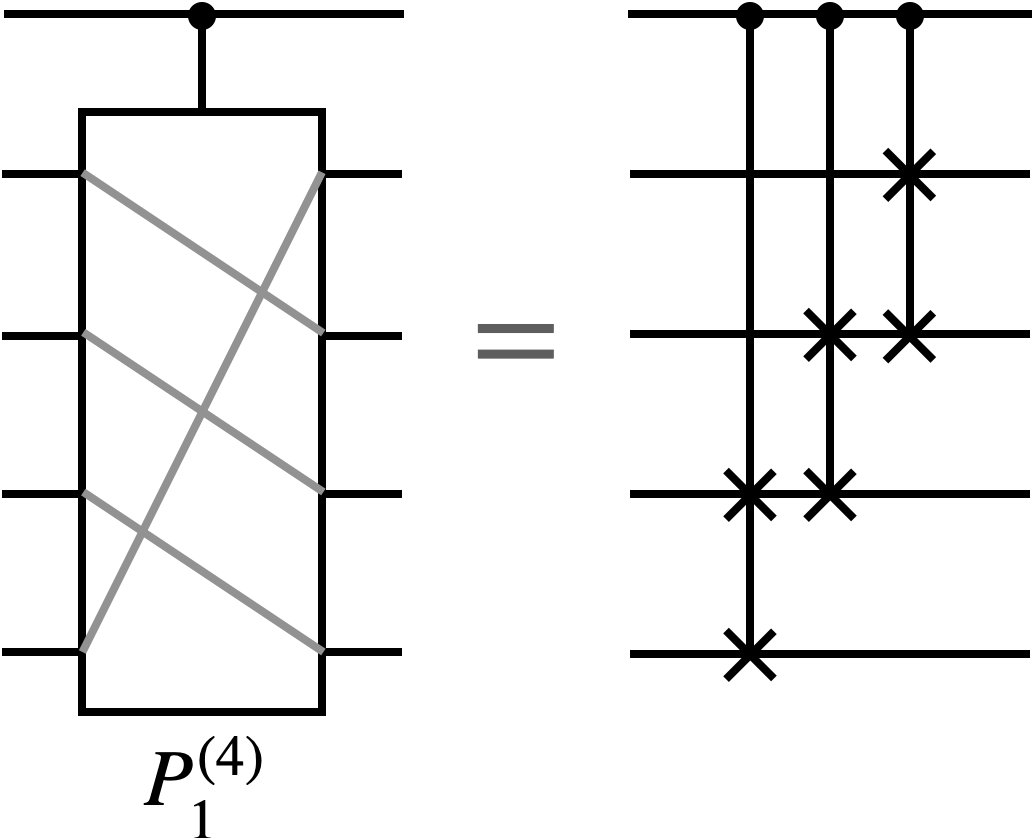}
        }
    }
    \caption{
        (a) The whole quantum circuit of CGG to stabilise $M=8$ noisy inputs.
        (b) The controlled-$P_{1}^{(4)}$ for $M=4$.
            The operator $P_{1}^{(4)}$ represents the cyclic shift to the next element, which can be decomposed into a sequential application of controlled-SWAP gates.
            The grey lines connect the original place and the destination of the quantum state moved by the operation.
    }
    \label{fig:qcs_cgg_m8}
\end{figure}

\par
When $M$ takes the power of a natural number, i.e. $M = {\alpha}^{n}$, the output reduced density matrix becomes
\begin{equation} \label{eq:rho_cyclic_alpha^n}
\begin{split}
    \operatorname{P}_{\vec{0}} \tilde{\rho}
    &= \frac{1}{{\alpha}^{n}}\rho + \frac{1}{{\alpha}^{n}} \sum_{i=1}^{n} \varphi({\alpha}^{i})\operatorname{Tr}\left[\rho^{{\alpha}^{i}}\right]^{{\alpha}^{n-i}-1}\rho^{{\alpha}^{i}}
    = \frac{1}{{\alpha}^{n}}\rho + \frac{1}{{\alpha}^{n}} \sum_{i=1}^{n} {\alpha}^{i-1}\operatorname{Tr}\left[\rho^{{\alpha}^{i}}\right]^{{\alpha}^{n-i}-1}\rho^{{\alpha}^{i}},
\end{split}
\end{equation}
where we use $\varphi({\alpha}^{n}) ={\alpha}^{n-1}$ for two integers $\alpha$ and $n$.
In particular, when $\alpha = 2$, the purified state $\tilde{\rho}$ becomes
\begin{equation} \label{eq:rho_cyclic_2^n}
\begin{split}
    \operatorname{P}_{\vec{0}} \tilde{\rho}
    &= \frac{1}{2^{n}}\rho + \frac{1}{2^{n}} \sum_{i=1}^{n} 2^{i-1}\operatorname{Tr}\left[\rho^{2^{i}}\right]^{2^{n-i}-1}\rho^{2^{i}}.
\end{split}
\end{equation}
For $n=3$, i.e. $M=8$, the purified state in Eq.~\eqref{eq:rho_cyclic_2^n} becomes
\begin{equation} \label{eq:rho_cyclic_2^3}
\begin{split}
    \operatorname{P}_{\vec{0}} \tilde{\rho}
    &= \frac{1}{8} \rho + \frac{1}{8} \sum_{i=1}^{3} 2^{i-1}\operatorname{Tr}\left[\rho^{2^{i}}\right]^{2^{3-i}-1}\rho^{2^{i}}
    = \frac{1}{8} \left( \rho + \operatorname{Tr}\left[\rho^{2}\right]^{3} \rho^{2} + 2 \operatorname{Tr}\left[\rho^{4}\right] \rho^{4} + 4 \rho^{8} \right).
\end{split}
\end{equation}

\par
We take this case as an example.
First, to apply the CGG projector 
\begin{equation} \label{eq:P_CGG^8}
\begin{split}
    P_{\mathrm{CGG}}^{(8)}=\frac{1}{M}\left(P_{0}^{(8)}+P_{1}^{(8)}+P_{2}^{(8)}+P_{3}^{(8)}+P_{4}^{(8)}+P_{5}^{(8)}+P_{6}^{(8)}+P_{7}^{(8)}\right),
\end{split}
\end{equation}
we use the quantum circuit in Fig.~\ref{fig:qcs_cgg_m8}(a).
When using three ancillary qubits, this can be achieved by sequentially applying the cyclic shift operations with $1$, $2$, and $4$ shifts, controlled by each ancillary qubit.
Each shift operation consists of $M-1=7$ SWAP gates controlled by a single ancillary qubit.
Figure~\ref{fig:qcs_cgg_m8}(b) shows the controlled-$P_{1}^{(4)}$ gate, which can be implemented by three (i.e. $4-1$) controlled-SWAP gates.

\begin{table}[htbp]
    \centering
    \begin{tabularx}{0.9\linewidth}{*{2}{X}|*{8}{X}}
        \hline
        &  & \multicolumn{8}{c}{power degree of $\rho$} \\
        &  & $1$ & $2$ & $3$ & $4$ & $5$ & $6$ & $7$ & $8$ \\
        \hline\hline 
        \multirow{8}{10pt}{$M$} & $1$ & $1$ &     &     &   &   &   &   &   \\
                                & $2$ & $1$ & $1$ &     &   &   &   &   &   \\
                                & $3$ & $1$ &     & $2$ &   &   &   &   &   \\
                                & $4$ & $1$ & $\operatorname{Tr}\left[\rho^{2}\right]$ &   & $2$ &   &   &   &   \\
                                & $5$ & $1$ &     &     &   & $4$ &   &   &   \\
                                & $6$ & $1$ & $\operatorname{Tr}\left[\rho^{2}\right]^{2}$ & $2\operatorname{Tr}\left[\rho^{3}\right]$ &   &   & $2$ &   &   \\
                                & $7$ & $1$ &     &     &   &   &   & $6$ &   \\
                                & $8$ & $1$ & $\operatorname{Tr}\left[\rho^{2}\right]^{3}$ &     & $2\operatorname{Tr}\left[\rho^{4}\right]$ &   &   &   & $4$ \\
        \hline
    \end{tabularx}
    \caption{
        The table of coefficients in each term of Eq.~\eqref{eq:outputs_cgg}.
        Ignoring the factor of trace values of powers of $\rho$, the integer coefficients correspond to the values of Euler's totient function $\varphi(m)$ until they reach the maximum degree $M$, i.e. $\varphi\left(M\right)$.
    }
    \label{tab:polynomials_rho}
\end{table}

\begin{figure}[htbp]
    \includegraphics[width=0.98\textwidth]{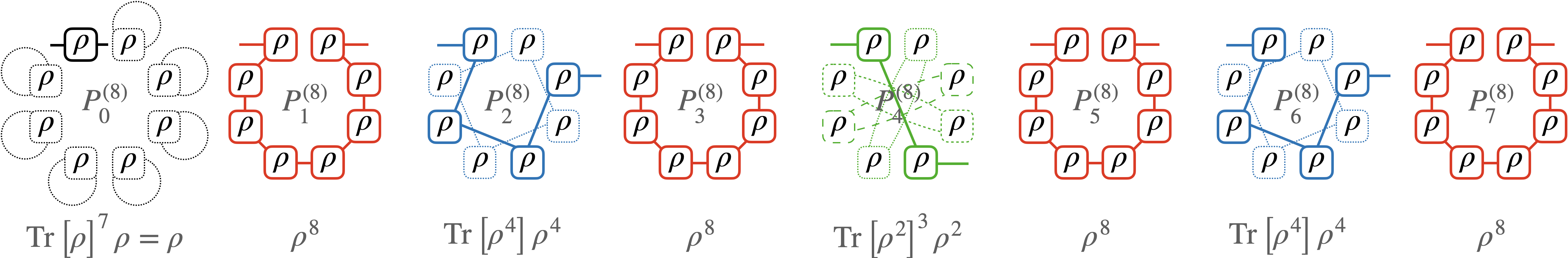}
    \caption{
        The illustration of rotation operations in the cyclic group $C_{8}$ with order $8$.
        $P_{i}^{(8)}$ denotes the rotation which maps teach element $k\in C_{8}$ to $k+i$, and $P_{0}^{(8)}$ represents the identity operation.
    }
    \label{fig:c_8}
\end{figure}

\par
Table~\ref{tab:polynomials_rho} aligns the coefficients of the polynomials of $\rho$ with different $M$, and the coefficient of each term in Eq.~\eqref{eq:rho_cyclic_2^3} can be visually explained by Fig.~\ref{fig:c_8}.
The multiplication of each rotation projector $P_{i}^{(8)}$ and the input states $\rho^{\otimes 8}$ followed by the partial trace to keep the first register can be understood as connecting tensors of $\rho$ and taking loops of $\rho$ accordingly in a tensor network representation.
In Fig.~\ref{fig:c_8}, we observe that $P_{1}^{(8)}$, $P_{3}^{(8)}$, $P_{5}^{(8)}$, and $P_{7}^{(8)}$ make irreducible transition loops with length $8$ since $1,3,5,7$ are coprime to $8=2^{3}$.
These operators contribute to the coefficient of term $\rho^{8}$ in Eq.~\eqref{eq:rho_cyclic_2^3}.
The operators $P_{2}^{(8)}$ and $P_{6}^{(8)}$ contribute to the coefficient of $\rho^{4}$ since the transition of elements by them results in two loops with length $4$.
Finally, $P_{4}^{(8)}$ makes four loops with length $2$ as $P_{4}^{(8)}$ can be decomposed to four independent pair-wise SWAP operations.

\section{Applying SGG, CGG and GSG in A Short-time Evolution \label{sec:appendix_proofs_short_time_evolution}}

\par
We check that CGG and GSG suppress coherent errors and stochastic errors to the same level as SGG when applying them repetitively in every short-time evolution.
Using the same noise assumptions in Barenco et al.~\cite{barenco1997stabilization}, we analyse the error suppression rate of CGG and GSG for coherent errors in Appendix~\ref{sec:appendix_Suppressing_Coherent_Errors_by_CGG_and_GSG} and stochastic errors in Appendix~\ref{sec:appendix_Suppressing_Stochastic_Errors_by_CGG_and_GSG}.

\subsection{Suppressing Coherent Errors \label{sec:appendix_Suppressing_Coherent_Errors_by_CGG_and_GSG}}

\par
We introduce several assumptions on the coherent errors in input states as assumed in~\cite{barenco1997stabilization}.
The coherent error can be seen as the random independent unitary rotation.
Suppose that the $M$ copies of the input state are subject to independent drift of eigenvalues as $e^{i H_{k} t}$ on the $k$-th copy, where the Hamiltonian $H_{k}$ is random and independent.
Hence, the whole quantum state of $M$ noisy inputs is assumed to be
\begin{equation}
\begin{split}
    |\psi\rangle = \bigotimes_{k=1}^{M} e^{iH_{k}\delta t}|\lambda_{0}\rangle
\end{split}
\end{equation}
We can introduce the specific form of Hamiltonian $H_{k}$ as $d\times d$ matrix
\begin{equation}
\begin{split}
    H_{k} =
    \left(\begin{array}{cc}
        a_{k}       & \vec{c}_{k}^{\dagger} \\
        \vec{c}_{k} & B_{k}
    \end{array}\right),
\end{split}
\end{equation}
assuming the dimension of each quantum state in each quantum register is $d$.
Finally, the purification gadget is applied in a short interval $\delta t$ and performed instantaneously.
This introduces the following quantum state for the approximation of $|\psi\rangle$.
\begin{equation}
\begin{split}
    |\psi\left(\delta t\right)\rangle 
    &\sim \bigotimes_{k=1}^{M} \left(I+iH_{k}\delta t\right) |\lambda_{0}\rangle
    = \bigotimes_{k=1}^{M} \left( \left(1+ia_{k}\delta t\right) |\lambda_{0}\rangle + i\delta t |\lambda_{\perp, k}\rangle \right),
\end{split}
\end{equation}
where $\displaystyle |\lambda_{\perp, k}\rangle = \sum_{j>0}c_{kj} |\lambda_{j}\rangle$.
By keeping only the contribution up to $O\left(\delta t\right)$ terms,
\begin{equation} \label{eq:psi_coherent_error_first_order_cgg}
\begin{split}
    |\psi\left(\delta t\right)\rangle 
    &= \left(1+i\delta t\sum_{k=1}^{M}a_{k}\right) |\lambda_{0}\rangle^{\otimes M}
    + i\delta t \sum_{k=1}^{M} \underset{k\text{-th}}{|\lambda_{0}\rangle\cdots|\lambda_{0}\rangle|\lambda_{\perp, k}\rangle|\lambda_{0}\rangle\cdots|\lambda_{0}\rangle}
    + O\left((\delta t)^{2}\right).
\end{split}
\end{equation}

\par
Here we remark that applying to Eq.~\eqref{eq:psi_coherent_error_first_order_cgg} the projector $\displaystyle P_{\mathrm{SGG}}$ in SGG, the projector $\displaystyle P_{\mathrm{CGG}}$ in CGG, and the projector $\displaystyle P_{\mathrm{GSG}}$ in GSG yield the same output, because any permutations of $|\lambda_{0}\rangle$ and $|\lambda_{\perp, k}\rangle$ in each term $\displaystyle |\lambda_{0}\rangle\cdots|\lambda_{0}\rangle|\lambda_{\perp, k}\rangle|\lambda_{0}\rangle\cdots|\lambda_{0}\rangle$ in Eq.~\eqref{eq:psi_coherent_error_first_order_cgg} can be seen as rotation operations.
As a result, $\displaystyle |\lambda_{0}\rangle\cdots|\lambda_{0}\rangle|\lambda_{\perp, k}\rangle|\lambda_{0}\rangle\cdots|\lambda_{0}\rangle$ becomes
\begin{equation}
\begin{split}
    P_{\mathrm{CGG}} \underset{k\text{-th}}{|\lambda_{0}\rangle\cdots|\lambda_{0}\rangle|\lambda_{\perp, k}\rangle|\lambda_{0}\rangle\cdots|\lambda_{0}\rangle} = \left(\frac{1}{M}\sum_{l=1}^{M}P_{l}\right) \underset{k\text{-th}}{|\lambda_{0}\rangle\cdots|\lambda_{0}\rangle|\lambda_{\perp, k}\rangle|\lambda_{0}\rangle\cdots|\lambda_{0}\rangle} = \frac{1}{\sqrt{M}}|E_{1,k}\rangle,
\end{split}
\end{equation}
where we define $\displaystyle |E_{1,k}\rangle$ as
\begin{equation}
\begin{split}
    |E_{1,k}\rangle 
    &= \frac{1}{\sqrt{M}} \left( |\lambda_{\perp, k}\rangle|\lambda_{0}\rangle\cdots|\lambda_{0}\rangle
    + |\lambda_{0}\rangle|\lambda_{\perp, k}\rangle\cdots|\lambda_{0}\rangle
    + \cdots
    + |\lambda_{0}\rangle\cdots|\lambda_{0}\rangle|\lambda_{\perp, k}\rangle \right).
\end{split}
\end{equation}
Using $|E_{1,k}\rangle$, the whole projected states $P_{\mathrm{SGG}}|\psi\left(\delta t\right)\rangle$ by SGG and $P_{\mathrm{CGG}}|\psi\left(\delta t\right)\rangle$ by CGG become
\begin{equation} \label{eq:psi_coherent_error_first_order_projected_cgg}
\begin{split}
    P_{\mathrm{SGG}}|\psi\left(\delta t\right)\rangle 
     = P_{\mathrm{CGG}}|\psi\left(\delta t\right)\rangle 
    &= \left(1+i\delta t\sum_{k=1}^{M}a_{k}\right) |\lambda_{0}\rangle^{\otimes M}
    + \frac{i\delta t}{\sqrt{M}}\sum_{k=1}^{M} |E_{1,k}\rangle + O\left((\delta t)^{2}\right).
\end{split}
\end{equation}

\par
As the projected states are the same between SGG and CGG, we compute only the unnormalised probability of quantum states other than $|\lambda_{0}\rangle$ from $P_{\mathrm{CGG}}|\psi\rangle$ being in the first quantum register after CGG. 
\begin{equation} \label{eq:prob_eigenvector_nondominant_coherent_error_first_order_projected_cgg}
\begin{split}
    & \left\| \left(\left(I-|\lambda_{0}\rangle\langle\lambda_{0}|\right)\otimes I^{\otimes(M-1)}\right) P_{\mathrm{CGG}}|\psi\left(\delta t\right)\rangle \right\|_{2}^{2} \\
    =& \Biggl\| \left(\left(I-|\lambda_{0}\rangle\langle\lambda_{0}|\right)\otimes I^{\otimes(M-1)}\right)
    \frac{i\delta t}{\sqrt{M}}\sum_{k=1}^{M} \sum_{l=1}^{M} \underset{l\text{-th}}{|\lambda_{0}\rangle\cdots|\lambda_{0}\rangle|\lambda_{\perp,k}\rangle|\lambda_{0}\rangle\cdots|\lambda_{0}\rangle} \Biggr\|_{2}^{2} \\
    =& \left(\frac{\delta t}{M}\right)^{2} \left\| \sum_{k=1}^{M} \left( \sum_{l=1}^{M} \underset{l\text{-th}}{|\lambda_{0}\rangle\cdots|\lambda_{0}\rangle|\lambda_{\perp,k}\rangle|\lambda_{0}\rangle\cdots|\lambda_{0}\rangle}
    - \sum_{l=2}^{M} \underset{l\text{-th}}{|\lambda_{0}\rangle\cdots|\lambda_{0}\rangle|\lambda_{\perp,k}\rangle|\lambda_{0}\rangle\cdots|\lambda_{0}\rangle} \right) \right\|_{2}^{2} \\
    =& \left(\frac{\delta t}{M}\right)^{2} \left\| \sum_{k=1}^{M} |\lambda_{\perp,k}\rangle \right\|^{2}
    = \frac{\left(\delta t\right)^{2}}{M} \left\| \frac{1}{\sqrt{M}}\sum_{k=1}^{M} |\lambda_{\perp,k}\rangle \right\|_{2}^{2} \\
    =& \frac{\left(\delta t\right)^{2}}{M} \langle\tilde{\lambda}_{\perp}|\tilde{\lambda}_{\perp}\rangle.
\end{split}
\end{equation}
where $|\tilde{\lambda}_{\perp}\rangle$ is defined as the average deviated state $\displaystyle |\tilde{\lambda}_{\perp}\rangle = \frac{1}{\sqrt{M}}\sum_{k=1}^{M} |\lambda_{\perp,k}\rangle$.

\par
On the other hand, the unnormalised probability of quantum states other than $|\lambda_{0}\rangle$ from $P_{\mathrm{CGG}}|\psi\left(\delta t\right)\rangle$ being in the first quantum register without the subspace projection by CGG is
\begin{equation} \label{eq:prob_eigenvector_nondominant_coherent_error_noisy}
\begin{split}
    & \left\| \left(\left(I-|\lambda_{0}\rangle\langle\lambda_{0}|\right)\otimes I^{\otimes(M-1)}\right) |\psi\left(\delta t\right)\rangle \right\|_{2}^{2} \\
    =& \left(\delta t\right)^{2} \Biggl\| \left(\left(I-|\lambda_{0}\rangle\langle\lambda_{0}|\right)\otimes I^{\otimes(M-1)}\right)
    \sum_{k=1}^{M} \underset{k\text{-th}}{|\lambda_{0}\rangle\cdots|\lambda_{0}\rangle|\lambda_{\perp,k}\rangle|\lambda_{0}\rangle\cdots|\lambda_{0}\rangle} \Biggr\|_{2}^{2} \\
    =& \left(\delta t\right)^{2} \left\| \sum_{k=1}^{M} \underset{k\text{-th}}{|\lambda_{0}\rangle\cdots|\lambda_{0}\rangle|\lambda_{\perp,k}\rangle|\lambda_{0}\rangle\cdots|\lambda_{0}\rangle}
    - \sum_{k=2}^{M} \underset{k\text{-th}}{|\lambda_{0}\rangle\cdots|\lambda_{0}\rangle|\lambda_{\perp,k}\rangle|\lambda_{0}\rangle\cdots|\lambda_{0}\rangle} \right\|_{2}^{2} \\
    =& \left(\delta t\right)^{2} \langle\lambda_{\perp,1}|\lambda_{\perp,1}\rangle.
\end{split}
\end{equation}

\par
When the devices are intended for state storage, the rate of drift is suitably bounded as $\displaystyle \left|\text{eigenvalues of }H_{k}\right| \leq \epsilon$ for $k=1, \ldots, M$ with a small constant $\epsilon$.
With this assumption, both $\langle\lambda_{\perp,1}|\lambda_{\perp,1}\rangle$ and $\langle\tilde{\lambda}_{\perp}|\tilde{\lambda}_{\perp}\rangle$ can be bounded by the Frobenius norm $\left\|H_{k}\right\|_{F} \leq d\epsilon^{2}$ and $\displaystyle \|\tilde{H}\|_{F} = \left\|\frac{1}{M}\sum_{k=1}^{M}H_{k}\right\|_{F} \leq d\epsilon^{2}$ as
\begin{equation}
\begin{split}
    \langle\lambda_{\perp,1}|\lambda_{\perp,1}\rangle 
    &= \|\vec{c}_{1}\|_{2}^{2} \leq \left\|H_{k}\right\|_{F} \leq d\epsilon^{2}, \\
    \langle\tilde{\lambda}_{\perp}|\tilde{\lambda}_{\perp}\rangle 
    &= \|\tilde{c}\|_{2}^{2} \leq \|\tilde{H}\|_{F} \leq d\epsilon^{2}, \\
\end{split}
\end{equation}
where $\tilde{c}$ is the average vector defined by $\displaystyle \tilde{c} = \frac{1}{M}\sum_{k=1}^{M}\vec{c}_{k}$.
Therefore, we observe that the probability of obtaining deviated states after CGG from $|\lambda_{0}\rangle$ is suppressed into $1/M$ of that of the noisy state without CGG.

\par
The success probability of GSG becomes the same as that of CGG because $P_{\mathrm{GSG}}$ also shuffles $|\lambda_{\perp, k}\rangle$ in $|\psi\left(\delta t\right)\rangle$ equally to any positions, resulting in Eq.~\eqref{eq:prob_eigenvector_nondominant_coherent_error_first_order_projected_cgg}.

\subsection{Suppressing Stochastic Errors \label{sec:appendix_Suppressing_Stochastic_Errors_by_CGG_and_GSG}}

\par
The model of stochastic errors can be introduced as the independent interaction between each state and its environment such that after a short period of time $\delta t$, the state of the $i$-th copy will have undergone an evolution $\rho^{(i)} = \rho_{0} + \sigma_{i}$ for some Hermitian $\sigma_{i}$ satisfying $\operatorname{Tr}\left[\sigma_{i}\right] = 0$.
Retaining only the terms of the first-order in the perturbations $\{\sigma_{i}\}_{i=1}^{M}$, the overall $M$ redundant noisy inputs at time $\delta t$ becomes
\begin{equation} \label{eq:rho_stochastic_perturbed}
\begin{split}
    \rho^{(1\ldots M)} = \rho_{0} \otimes \cdots \otimes \rho_{0}
    &+ \sigma_{1} \otimes \rho_{0} \otimes \cdots \otimes \rho_{0} \\
    &+ \rho_{0} \otimes \sigma_{2} \otimes \cdots \otimes \rho_{0} \\
    &+ \cdots \\
    &+ \rho_{0} \otimes \rho_{0} \otimes \cdots \otimes \sigma_{M} \\
    &+ O\left(\sigma_{i} \sigma_{j}\right).
\end{split}
\end{equation}
We are going to apply each CGG and GSG to this state, respectively, keeping only the first-order perturbation of $\{\sigma_{i}\}_{i=1}^{M}$ during the calculation.

\par
According to~\cite{barenco1997stabilization}, applying the full symmetric group projector $P_{\mathrm{SGG}}$ to $\rho^{(1\ldots M)}$ and tracing out the quantum registers except for the first one, the output quantum state $\tilde{\rho}$ is described as
\begin{equation}
\begin{split}
    \tilde{\rho} 
    &= {\left(1-(M-1) \operatorname{Tr}\left[\rho_{0} \tilde{\sigma}\right]\right) \rho_{0}+\frac{1}{M} \tilde{\sigma} }
    +(M-1)\left(A \rho_{0} \tilde{\sigma} \rho_{0}+B\left(\rho_{0} \tilde{\sigma}+\tilde{\sigma} \rho_{0}\right)+C \rho_{0} \operatorname{Tr}\left[\rho_{0}\tilde{\sigma}\right]\right)
    +O\left(\sigma_{i} \sigma_{j}\right),
\end{split}
\end{equation}
where $A + 2B + C = 1$, and $\displaystyle \tilde{\sigma} = \frac{1}{M}\sum_{i=1}^{M}\sigma_{i}$.

\par
Next, we calculate the output state from CGG.
Following the discussion in~\cite{barenco1997stabilization}, we consider the application of the projection operator $\displaystyle \frac{1}{M}\sum_{i=0}^{M-1}P_{i}^{(M)}$ to each of the $M$ terms of $\rho_{0}\otimes\cdots\otimes\sigma_{i}\otimes\cdots\otimes\rho_{0}$ in Eq.~\eqref{eq:rho_stochastic_perturbed} of the form
\begin{equation} \label{eq:term_cgg_stochastic_perturbed}
\begin{split}
    \frac{1}{M^{2}} P_{l}^{(M)} \left( \rho_{0} \otimes \cdots \otimes \sigma_{i} \otimes \cdots \otimes \rho_{0} \right) P_{r}^{(M)},
\end{split}
\end{equation}
where $l, r\in\{0,\ldots,M-1\}$.
For each $\sigma_{i}$, we observe that the $M^{2}$ terms taking the form of~\eqref{eq:term_cgg_stochastic_perturbed} reduce to the following cases.
\begin{enumerate}
    \item One term equal to $\displaystyle \frac{1}{M^{2}}\sigma_{i}$, when $P_{l} = P_{M-i-1}^{(M)}$ and $P_{r} = P_{i}^{(M)}$.
    \item $M-1$ terms equal to $0$, when $P_{l} = P_{M-j-1}^{(M)}$ and $P_{r} = P_{j}^{(M)}$ for $j\neq i$.
    \item $M-1$ terms equal to $\displaystyle \frac{1}{M^{2}}\rho_{0}\sigma_{i}$, when $P_{l} = P_{M-i-1}^{(M)}$ and $P_{r} = P_{j}^{(M)}$ for $j\neq i$.
    \item $M-1$ terms equal to $\displaystyle \frac{1}{M^{2}}\sigma_{i}\rho_{0}$, when $P_{l} = P_{M-j-1}^{(M)}$ for $j\neq i$ and $P_{r} = P_{i}^{(M)}$.
    \item $(M-1)(M-2)$ terms equal to $\displaystyle \frac{1}{M^{2}}\rho_{0}\sigma_{i}\rho_{0}$, when $P_{l} = P_{M-j-1}^{(M)}$ for $j\neq i$ and $P_{r} = P_{k}^{(M)}$ for $k\neq i$ and $P_{l}^{\dagger}\neq P_{r}$.
\end{enumerate}
These five cases are visualised in Fig.~\ref{fig:patterns_cgg}.

\begin{figure*}[htbp]
    \centering
    \includegraphics[width=\textwidth]{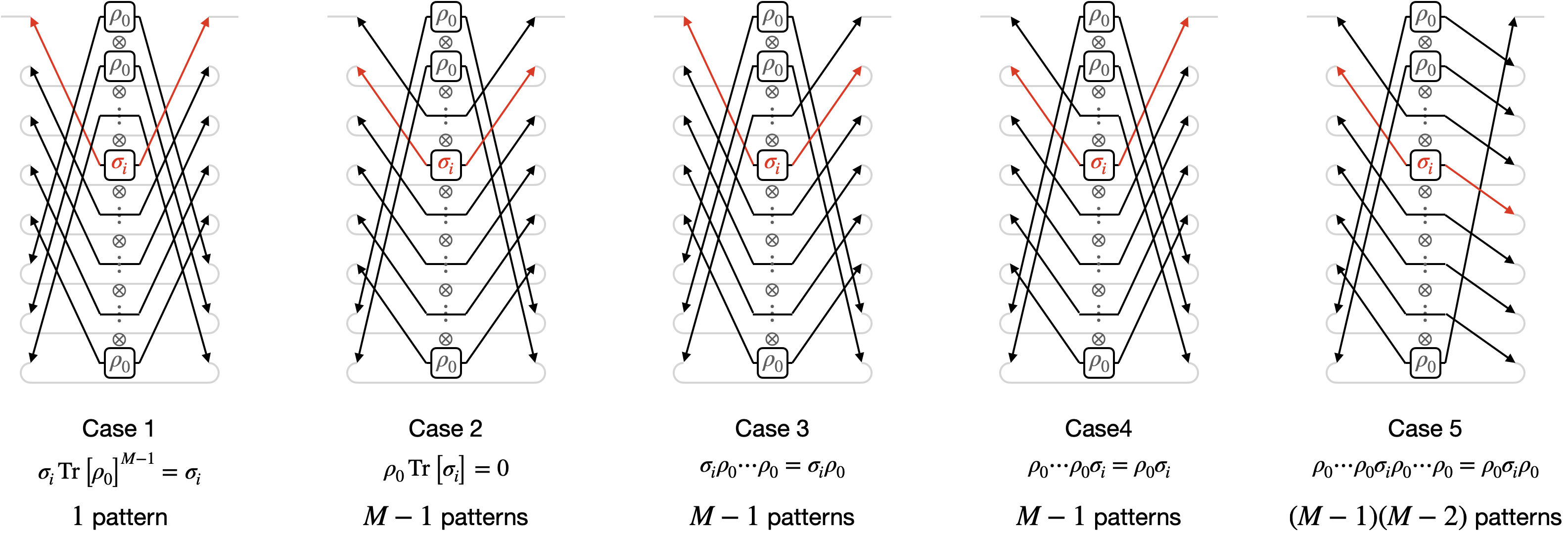}
    \caption{
        The five cases of combinations of cyclic projectors resulting in different quantum states in tensor network representation.
        The effect of the cyclic projectors is shown by the arrows.
        The partial trace operations taken from the second quantum register to the last quantum register are shown by grey loops.
    }
    \label{fig:patterns_cgg}
\end{figure*}

\par
Thus, taking the average $\displaystyle \tilde{\sigma} = \frac{1}{M}\sum_{i=1}^{M}\sigma_{i}$ of the local noise terms $\{\sigma_{i}\}_{i=1}^{M}$, the resulting density matrix $\tilde{\rho}'_{\mathrm{CGG}}$ before normalisation becomes
\begin{equation}
\begin{split}
    \tilde{\rho}'_{\mathrm{CGG}}
    &= \rho_{0}+\sum_{i=1}^{M}\biggl(\frac{1}{M^{2}}\left(\sigma_{i} + (M-1)\left(\rho_{0} \sigma_{i} + \sigma_{i} \rho_{0}\right) + (M-1)(M-2) \rho_{0} \sigma_{i} \rho_{0}\right)\biggr) \\
    &= \rho_{0}+\frac{1}{M} \tilde{\sigma}+\frac{M-1}{M}\left(\rho_{0} \tilde{\sigma}+\tilde{\sigma} \rho_{0}+(M-2) \rho_{0} \tilde{\sigma} \rho_{0}\right).
\end{split}
\end{equation}
This density matrix has a trace given by
\begin{equation}
\begin{split}
    \operatorname{Tr}\left[\tilde{\rho}'_{\mathrm{CGG}}\right]
    &=1+\frac{M-1}{M} \left(1+1+M-2\right) \operatorname{Tr}\left[\rho_{0} \tilde{\sigma}\right]
    =1+(M-1) \operatorname{Tr}\left[\rho_{0} \tilde{\sigma}\right].
\end{split}
\end{equation}
Eventually, we obtain the purified quantum state $\tilde{\rho}$ as
\begin{equation}
\begin{split}
    \tilde{\rho}_{\mathrm{CGG}}
    = \frac { \tilde{\rho}'_{\mathrm{CGG}} } { \operatorname{Tr}\left[\tilde{\rho}'_{\mathrm{CGG}}\right] }
    &=\left(1 - (M-1)\operatorname{Tr}\left[\rho_{0} \tilde{\sigma}\right]\right)\rho_{0} + \frac{1}{M} \tilde{\sigma} + \frac{M-1}{M} \left( \rho_{0}\tilde{\sigma} + \tilde{\sigma}\rho_{0} + (M-2) \rho_{0}\tilde{\sigma}\rho_{0} \right).
\end{split}
\end{equation}

\par
GSG follows a similar discussion to CGG and gives the same output as CGG since the difference between CGG and GSG only lies in the way to shuffle the input quantum registers, which results in slightly different proportion of $\rho_{0}$, $\tilde{\sigma}$, $\rho_{0}\tilde{\sigma}$, $\tilde{\sigma}\rho_{0}$, $\operatorname{Tr}\left[\rho_{0}\tilde{\sigma}\right]\rho_{0}$ and $\rho_{0}\tilde{\sigma}\rho_{0}$, as shown in Fig.~\ref{fig:patterns_gsg}.
The form of the purified state from GSG is described as
\begin{equation}
\begin{split}
    \tilde{\rho}_{\mathrm{GSG}}
    = \frac { \tilde{\rho}'_{\mathrm{GSG}} } { \operatorname{Tr}\left[\tilde{\rho}'_{\mathrm{GSG}}\right] }
    &=\left(1 - 2\frac{M-1}{M}\operatorname{Tr}\left[\rho_{0} \tilde{\sigma}\right]\right)\rho_{0} + \frac{1}{M} \tilde{\sigma} + \frac{M-1}{M} \left( \rho_{0}\tilde{\sigma} + \tilde{\sigma}\rho_{0} \right).
\end{split}
\end{equation}

\begin{figure*}[htbp]
    \centering
    \includegraphics[width=\textwidth]{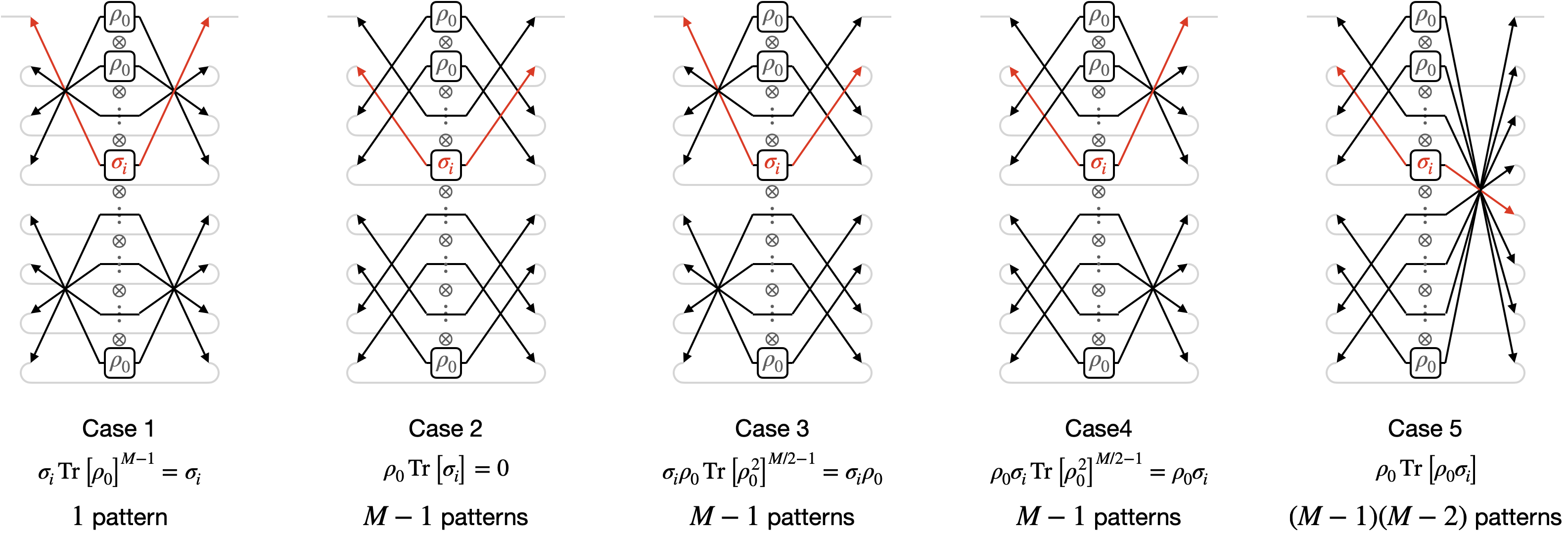}
    \caption{
        The five cases of combinations of parallel SWAP projectors resulting in different quantum states in tensor network representation.
        The effect of the parallel SWAP projectors is shown by the arrows.
        Grey loops show the partial trace operations taken from the second quantum register to the last quantum register.
    }
    \label{fig:patterns_gsg}
\end{figure*}

\par
Here we find that the success probability of obtaining $\rho_{0}$ in $\tilde{\rho}$ after either CGG or GSG is the same as that of SGG under the first-order perturbation of $\{\sigma_{i}\}_{i=0}^{M-1}$.
In fact, the average success probability after SGG, CGG, and GSG all becomes
\begin{equation} \label{eq:fidelity_cgg_gsg_stochastic_perturbated}
\begin{split}
    \operatorname{Tr}\left[\rho_{0} \tilde{\rho}\right] = 1+\frac{1}{M} \operatorname{Tr}\left[\rho_{0} \tilde{\sigma}\right].
\end{split}
\end{equation}
Recall the average success probability among the noisy inputs is represented as 
\begin{equation} 
\begin{split}
    \frac{1}{M} \sum_{i} \operatorname{Tr}\left[\rho_{0}\left(\rho_{0}+\sigma_{i}\right)\right] = 1 + \operatorname{Tr} \left[\rho_{0} \tilde{\sigma} \right],
\end{split}
\end{equation}
these gadgets linearly suppress the probability of accepting unwanted states, which can also be seen as increasing the state fidelity to the desired state by a factor of $1/M$ using $M$ redundant inputs.

\par
In addition, the average impurity is also suppressed by a factor of $1/M$.
The purity of $\tilde{\rho}$ is
\begin{equation} \label{eq:purity_cgg_gsg_stochastic_perturbated}
\begin{split}
    \operatorname{Tr}\left[\tilde{\rho}^2\right]
    = 1+\frac{2}{M} \operatorname{Tr}\left[\rho_{0}\tilde{\sigma}\right]
\end{split}
\end{equation}
which is also enhanced by a factor of $M$ from the purity of the noisy state
\begin{equation}
\begin{split}
    \frac{1}{M} \sum_{i=1}^{M} \operatorname{Tr}\left[\left(\rho_{0}+\sigma_{i}\right)^{2}\right] = 1 + 2 \operatorname{Tr}\left[ \rho_{0} \tilde{\sigma} \right].
\end{split}
\end{equation}
Note that we keep ignoring the contribution $O\left(\sigma_{i}\sigma_{j}\right)$ during all calculations and $\operatorname{Tr}\left[\rho_{0} \tilde{\sigma}\right] < 0$.

\section{Proofs of Purification Rates of CGG Taking Depolarised Inputs \label{sec:appendix_proofs_CGG}}

In this section, we give detailed proofs of the properties of CGG when it takes depolarised inputs.

\subsection{The Purification Rate for Inputs with Infinite Dimension \label{sec:appendix_proof_CGG_p_tilde_inf}}

\par
We compute the output probability $\tilde{p}_{d=\infty}$ for depolarised inputs with infinite dimensions.
To begin with, we remark that the $m$-th power of the depolarised density matrix with dimension $d$ and its trace are formalised in the following way.
\begin{equation}
\begin{split}
    \rho^{m} 
    &= \left(\left(1-p\right)\rho_{0} + p\frac{I}{d}\right)^{m}
    = \left( \left(1-p+\frac{p}{d}\right)^{m} - \left(\frac{p}{d}\right)^{m} \right)\rho_{0} + \left(\frac{p}{d}\right)^{m-1}p\frac{I}{d}, \\
    \operatorname{Tr}\left[\rho^{m}\right]
    &= \left(1-p+\frac{p}{d}\right)^{m} + \left( d-1 \right) \left(\frac{p}{d}\right)^{m}. \\
\end{split}
\end{equation}
Using these properties, the depolarising probability $\tilde{p}$ of the output state $\displaystyle \tilde{\rho}=\left(1-\tilde{p}\right)\rho_{0}+\tilde{p}\frac{I}{d}$ and the post-selection probability $\operatorname{P}_{\vec{0}}$ become
\begin{equation} \label{eq:outputs_CGG}
\begin{split}
    \tilde{p}
    &= \frac{1}{\operatorname{P}_{\vec{0}}} \frac{1}{M} p \sum_{m|M}\varphi\left(m\right) \left(\left(1-p+\frac{p}{d}\right)^{m} + \left( d-1 \right) \left(\frac{p}{d}\right)^{m}\right)^{\frac{M}{m}-1}\left(\frac{p}{d}\right)^{m-1}, \\
    \operatorname{P}_{\vec{0}}
    &= \frac{1}{M}\sum_{m|M}\varphi\left(m\right) \left(\left(1-p+\frac{p}{d}\right)^{m} + \left( d-1 \right) \left(\frac{p}{d}\right)^{m}\right)^{\frac{M}{m}}. \\
\end{split}
\end{equation}
When the dimension $d$ goes to infinity, $\operatorname{P}_{\vec{0}}$ converges to
\begin{equation} \label{eq:P0_CGG_d_inf_proof}
\begin{split}
    \operatorname{P}_{\vec{0}}
    &= \frac{1}{M} + \frac{1}{M}\sum_{m|M, m>1}\varphi\left(m\right) \left(\left(1-p+\frac{p}{d}\right)^{m} + \left( d-1 \right) \left(\frac{p}{d}\right)^{m}\right)^{\frac{M}{m}} \\
    &\overset{d\rightarrow\infty}{\longrightarrow} \frac{1}{M} \left( 1 + \sum_{m|M, m>1}\varphi\left(m\right) \left(\left(1-p\right)^{m}\right)^{\frac{M}{m}} \right) p
    =\frac{1}{M} \left( 1 + \left(1-p\right)^{M} \sum_{m|M, m>1}\varphi\left(m\right) \right) \\
    &=\frac{1}{M} \left( 1 + \left(M-1\right) \left(1-p\right)^{M} \right).
\end{split}
\end{equation}
Likewise, $\tilde{p}$ converges to
\begin{equation} \label{eq:p_tilde_CGG_d_inf_proof}
\begin{split}
    \tilde{p}
    &= \frac{1}{\operatorname{P}_{\vec{0}}} \frac{1}{M} p \left( 1 + \sum_{m|M, m>1}\varphi\left(m\right) \left(\left(1-p+\frac{p}{d}\right)^{m} + \left( d-1 \right) \left(\frac{p}{d}\right)^{m}\right)^{\frac{M}{m}-1} \left(\frac{p}{d}\right)^{m-1} \right) \\
    &\overset{d\rightarrow\infty}{\longrightarrow} \frac{1}{ 1 + \left(M-1\right) \left(1-p\right)^{M} } p,
\end{split}
\end{equation}
which yields Eq.~\eqref{eq:p_tilde_CGG_d_inf}.

\subsection{Maximal Purification Rate and Condition of CGG \label{sec:appendix_proof_CGG_maximal_purification_rate_and_condition}}

\par
We show the maximal purification rate and the condition in the number of copies $M$ required in CGG.
To compute the balance point of the purification rate $\tilde{p}_{d=\infty}/p$, we see the derivative $\displaystyle \frac{d}{d M} \tilde{p}_{d=\infty}$, which can be expanded as
\begin{equation}
\begin{split}
    \frac{d}{d M} \tilde{p}_{d=\infty}
    &= -\frac { \left(1 - p\right)^{M} \left(1 + \left(M - 1\right) \log\left(1 - p\right)\right) } 
              { \left(1 + \left(M - 1\right) \left(1 - p\right)^{M}\right)^{2} } p.
\end{split}
\end{equation}
The condition $\tilde{p}_{d=\infty}$ decreases monotonically to $M$ is $\displaystyle \frac{d}{d M} \tilde{p}_{d=\infty} < 0$, corresponding to $1 + \left(M - 1\right) \log\left(1 - p\right) > 0$.
Solving this for $M$, the condition that $\tilde{p}/p$ is suppressed monotonically to $M$ is
\begin{equation} \label{eq:condition_M_d_inf}
\begin{split}
    M < 1 - \frac{1}{\log\left(1 - p\right)} =: M_{d=\infty}^{*}.
\end{split}
\end{equation}

\par
Next, we compute the purification rate when taking the optimal number of copies $M_{d=\infty}^{*}$.
Assigning $\displaystyle M_{d=\infty}^{*}$ to Eq.~\eqref{eq:p_tilde_CGG_d_inf}, we obtain the optimal purified depolarising rate $\tilde{p}_{d=\infty}^{*}$ as
\begin{equation} \label{eq:p_tilde_optimal_M}
\begin{split}
    \tilde{p}_{d=\infty}^{*}
    &= \frac { e\log\left(1-p\right) }
             { e\log\left(1-p\right) + p - 1 } p.
\end{split}
\end{equation}
We observe that the Taylor expansion of $\tilde{p}_{d=\infty}^{*}$ around $p=0$ is $e p^{2}+O\left(p^{3}\right)$.
Thus, we conclude that for small $p$, CGG has error suppression rate from $p$ to $O\left(p^{2}\right)$ maximally.

\par
We further bound this ratio by observing $e^{-2}<e^{\frac{M}{1-M}} < e^{-1}$ for $M>1$ and $e^{\frac{M}{1-M}}\overset{M\rightarrow\infty}{\longrightarrow}e^{-1}$, which results in the following bound
\begin{equation}
\begin{split}
    \frac {e} {M_{d=\infty}^{*} + \left(e - 1\right)} < \frac {\tilde{p}_{d=\infty}^{*}} {p} < \frac {e^{2}} {M_{d=\infty}^{*} + \left(e^{2}-1\right)}.
\end{split}
\end{equation}
Hence, the error suppression rate scales in $\displaystyle O\left(\frac{1}{M}\right)$ to the number of copies $M$.
Taking the logarithm of Eq.~\eqref{eq:p_tilde_optimal_M} for $d=\infty$, we can deduce the fitting function of the valley point in Fig.~\ref{fig:CGG_ps-tildes_M-all}(a).
\begin{equation} \label{eq:p_tilde_to_p_optimal_M_d_inf}
\begin{split}
    \log\left(\frac{\tilde{p}_{d=\infty}^{*}}{p}\right) 
    &= - \log\left(1+\left(M_{d=\infty}^{*}-1\right)e^{\frac{M_{d=\infty}^{*}}{1-M_{d=\infty}^{*}}}\right) \\
    &< - \log\left(M_{d=\infty}^{*}e^{\frac{M_{d=\infty}^{*}}{1-M_{d=\infty}^{*}}}\right) \\
    &= - \log\left(M_{d=\infty}^{*}\right) + \frac{1}{1-\frac{1}{M_{d=\infty}^{*}}} \\
    &< - \log\left(M_{d=\infty}^{*}\right) + 1.
\end{split}
\end{equation}
Therefore there exists $\displaystyle M < 1 - \frac{1}{\log\left(1-p\right)}$ that achieves the purification rate in logarithmic scale
\begin{equation} \label{eq:func_p_tilde_and_M_linear}
\begin{split}
    \log \left(\frac{\tilde{p}}{p}\right) = - \log \left(M\right) + 1,
\end{split}
\end{equation}
which implies the asymptotical linear relation between the maximal purification rate $\frac{\tilde{p}}{p}$ and the number of copies $M$.
This relation is displayed as the solid black line in Fig.~\ref{fig:CGG_ps-tildes_M-all}(a).

\subsection{\texorpdfstring{Sampling Cost of CGG for Small $p$}{} \label{sec:appendix_proof_CGG_sampling_cost}}

\par
When $p$ is small, our gadget outputs the same purification rate as RSG~\cite{childs2023streaming}, which is shown as optimal in terms of sampling cost.
To see this, we check the purification rate of RSG taking $\displaystyle M_{d=\infty}^{*}=1-\frac{1}{\log\left(1-p\right)}$ copies of depolarised inputs.
Recalling that the purified depolarising rate $\tilde{p}_{RSG}$ is described as Eq.~\eqref{eq:outputs_recursive_SWAP_depolarising_pure} for $\displaystyle p<\frac{1}{2}$, we obtain
\begin{equation} \label{eq:p_tilde_optimal_M_RSG}
\begin{split}
    \tilde{p}_{RSG}
    &< \frac{ \displaystyle 1 }
            { \displaystyle \left(1-\frac{1}{\log\left(1-p\right)}\right)\left(1-2p\right) +2p } p
    = \frac{ p\log\left(1-p\right) }
           { \log\left(1-p\right) + 2p - 1 },
\end{split}
\end{equation}
by replacing $2^{n}$ with $M_{d=\infty}^{*}$.
Taking up to the second order of $p$ around $p=0$ in the Taylor expansion, Eq.~\eqref{eq:p_tilde_optimal_M_RSG} becomes $p^{2}+O\left(p^{3}\right)$.
This implies the purification rate of CGG differs from that of RSG only on a constant scale, with the same number of copies.

\par
We analytically compare the sampling cost of CGG and that of RSG as well.
The acceptance probability of CGG $\operatorname{P}_{\vec{0}}$ with $M$ copies when $d=\infty$ is
\begin{equation} \label{eq:prob_post_select_CGG_d_inf}
\begin{split}
    \operatorname{P}_{\vec{0}} 
    &= \frac{1}{M}\left( 1 + \left(M-1\right) \left(1-p\right)^{M} \right).
\end{split}
\end{equation}
Thus the expected number of copies required to post-select the all-zero state in the ancillary qubits for the optimal $M$ is
\begin{equation} \label{eq:num_samples_M_optimal_d_inf}
\begin{split}
    \frac { M_{d=\infty}^{*} } 
          { \operatorname{P}_{\vec{0}} }
    &= \frac { \displaystyle \left(\log\left(1-p\right)-1\right)^{2} } 
             { \displaystyle \left(\log\left(1-p\right) - \frac{1}{e}\left(1-p\right) \right)\log\left(1-p\right) }.
\end{split}
\end{equation}
The Laurent series of Eq.~\eqref{eq:num_samples_M_optimal_d_inf} up to the first order of $p$ around $0$ is
\begin{equation}
\begin{split}
    \frac{e}{p}+\left(\frac{5}{2}-e\right) e+\left(\frac{41 e}{12}-4 e^{2}+e^{3}\right) p + O\left(p^{2}\right).
\end{split}
\end{equation}
To achieve the same purification quality with RSG, the sampling cost becomes
\begin{equation} \label{eq:num_samples_M_optimal_RSG}
\begin{split}
    C\left(n, d\right) \leq \frac{2p}{\tilde{p}^{*}\left(1-2p\right)^{2}} = \frac { 2\left(e \log \left(1-p\right)+p-1\right) } { e\left(1-2 p\right)^{2} \log \left(1-p\right) }.
\end{split}
\end{equation}
The Laurent series of Eq.~\eqref{eq:num_samples_M_optimal_RSG} up to the first order of $p$ around $0$ is
\begin{equation}
\begin{split}
    \frac{2}{e p}+\left(2+\frac{5}{e}\right)+\left(8+\frac{77}{6 e}\right) p + O\left(p^{2}\right).
\end{split}
\end{equation}
This implies our method achieves the same order in $p$ up to the constant difference as the RSG with fewer measurement operations when $p$ is small enough.

\subsection{\texorpdfstring{Purification Rate of CGG when $M$ is prime}{} \label{sec:appendix_proof_CGG_M_prime}}

\par
For later use, let $\beta$ and $\gamma$ respectively be $\displaystyle \beta  = 1 - p + \frac{p}{d} = 1 - \left(1 - \frac{1}{d}\right)p, \quad \gamma = \frac{p}{d}$.
Note that $0<\gamma<\beta<1$.
Using $\beta$ and $\gamma$, $\rho$ can be rewritten as
\begin{equation} \label{eq:rho_depolarised_beta_gamma}
\begin{split}
    \rho 
    &= (\beta - \gamma) \rho_{0} + \gamma I, \\
    \operatorname{Tr}\left[\rho^{k}\right] 
    &= \operatorname{Tr}\left[\sum_{i=0}^{k}{}_{k}C_{i}(\beta-\gamma)^{k-i}\gamma^{i}\rho_{0}^{k-i}I^{i}\right]
    = \sum_{i=0}^{k}{}_{k}C_{i}(\beta-\gamma)^{k-i}\gamma^{i} + \left(d^{k}-1\right)\gamma^{k} 
    = \beta^{k} + \left(1-\beta\right)\gamma^{k-1}.
\end{split}
\end{equation}
We expand Eq.~\eqref{eq:outputs_cgg_prime} using $\beta$ and $\gamma$ into
\begin{equation}
\begin{split}
    \operatorname{P}_{\vec{0}}\tilde{\rho} 
    &= \frac{1}{M} \left( 1 - p + \left(M-1\right) \left( \left(1-p+\frac{p}{d}\right)^{M} - \left(\frac{p}{d}\right)^{M} \right) \right) \rho_{0}
    + \frac{1}{M} \left( 1 + \left(M-1\right) \left(\frac{p}{d}\right)^{M-1}\right) p\frac{I}{d} \\
    &= \frac{1}{M} \left( \beta - \gamma + \left(M-1\right) \left( \beta^{M} - \gamma^{M} \right) \right) \rho_{0}
    + \frac{1}{M} \left( 1 + \left(M-1\right) \gamma^{M-1}\right) p\frac{I}{d}, \\
    \operatorname{P}_{\vec{0}} 
    &= \frac{1}{M}\left( 1 + \left(M-1\right) \left(\left(1-p+\frac{p}{d}\right)^{M} + \left(\frac{p}{d}\right)^{M}\left(d-1\right)\right) \right) \\
    &= \frac{1}{M}\left( 1 + \left(M-1\right)\left(\beta^{M} + \gamma^{M-1}\left(1 - \beta\right) \right) \right).
\end{split}
\end{equation}
This implies that the output state $\tilde{\rho}$ also takes the depolarised state $\displaystyle \tilde{\rho} = \left(1-\tilde{p}\right)\rho_{0} + \tilde{p}\frac{I}{d}$ with depolarising rate $\tilde{p}$ as 
\begin{equation} \label{eq:p_tilde_m_prime_depolarising}
\begin{split}
    \tilde{p} 
    &= \frac { \displaystyle 1 + \left(M-1\right)\left(\frac{p}{d}\right)^{M-1} } 
             { \displaystyle 1 + \left(M-1\right)\left(\left(1-p+\frac{p}{d}\right)^{M} + \left(\frac{p}{d}\right)^{M}\left(d-1\right) \right) } p
    = \frac { \displaystyle 1 + \left(M-1\right)\gamma^{M-1} } 
             { \displaystyle 1 + \left(M-1\right)\left(\beta^{M} + \gamma^{M-1}\left(1 - \beta\right) \right) } p.
\end{split}
\end{equation}
The derivative $\displaystyle \frac{d}{d M} \tilde{p}$ for prime $M$ then becomes
\begin{equation}
\begin{split}
    \frac{d}{d M} \tilde{p}
    &= \frac 
        { \displaystyle \beta\gamma^{M} - \beta^{M}\gamma + (M-1)\left(\log\left(\gamma\right)\beta\gamma^{M} - \log\left(\beta\right)\beta^{M}\gamma\right) + (M-1)^{2} \left(\log\left(\beta\right) - \log\left(\gamma\right)\right)\beta^{M}\gamma^{M} } 
        { \displaystyle \gamma\left( 1 + \left(M-1\right)\left(\beta^{M} + \gamma^{M-1}\left(1 - \beta\right) \right) \right)^{2} }.
\end{split}
\end{equation}
The condition $\tilde{p}$ decreases monotonically to $M$ is $\displaystyle \frac{d}{d M} \tilde{p} < 0$, which can be described by $\beta$ and $\gamma$ as $f(\beta, \gamma) - f(\gamma, \beta) > 0$ by introducing a function $\displaystyle f(x, y) = x^{M-1} \left(1 + (M - 1) \left(1 + (M - 1) y^{M-1}\right) \log(x) \right)$.
The areas representing $f(\beta, \gamma) - f(\gamma, \beta) > 0$ under different $M$ are visualised in Fig.~\ref{fig:condition_bc}.

\begin{figure}[htbp]
    \centering
    \subfloat[\label{fig:condition_bc_M-2} $M=2$]{
        \includegraphics[width=0.16\textwidth]{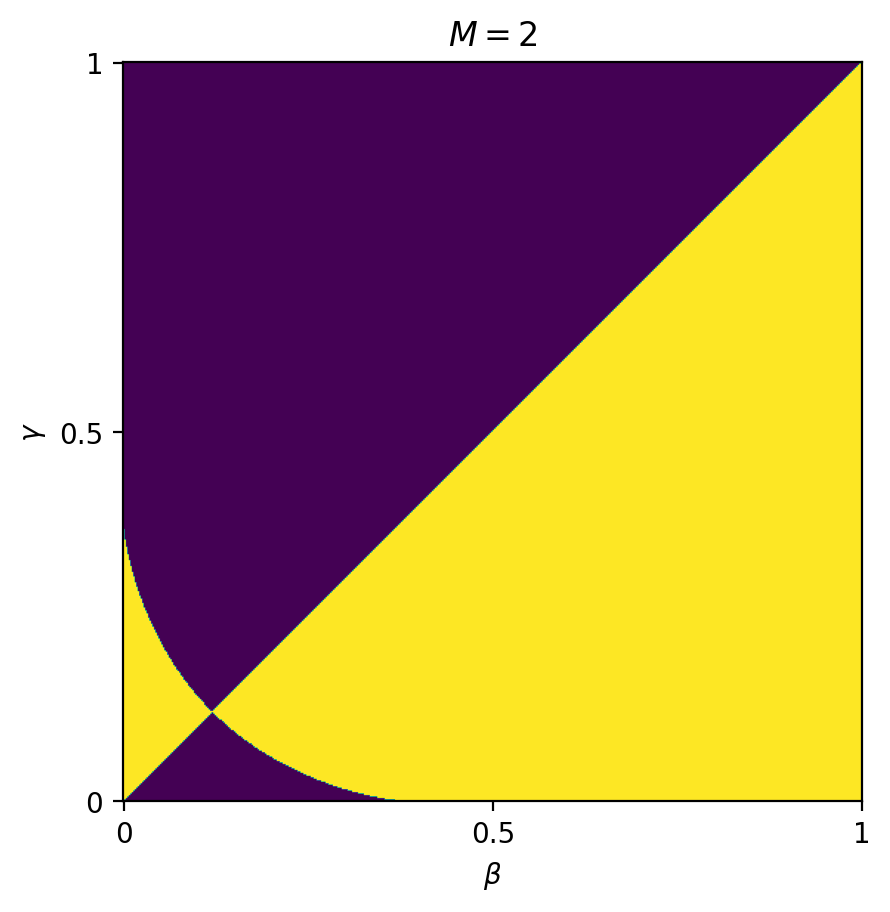}
    }
    \subfloat[\label{fig:condition_bc_M-3} $M=3$]{
        \includegraphics[width=0.16\textwidth]{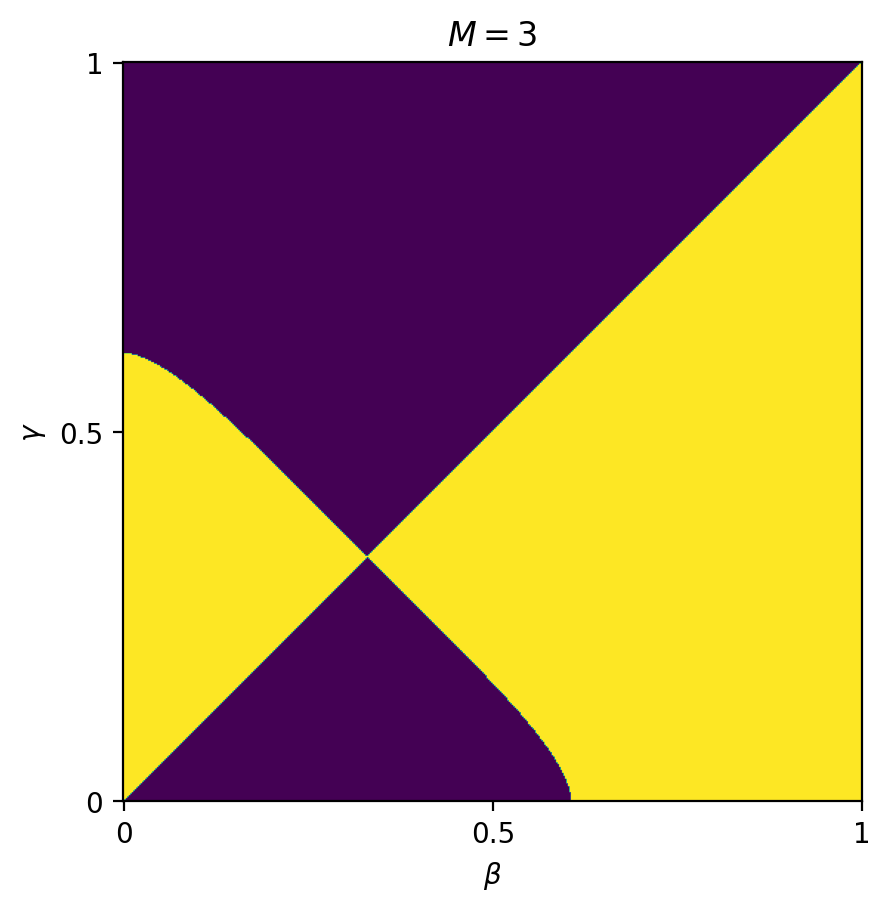}
    }
    \subfloat[\label{fig:condition_bc_M-5} $M=5$]{
        \includegraphics[width=0.16\textwidth]{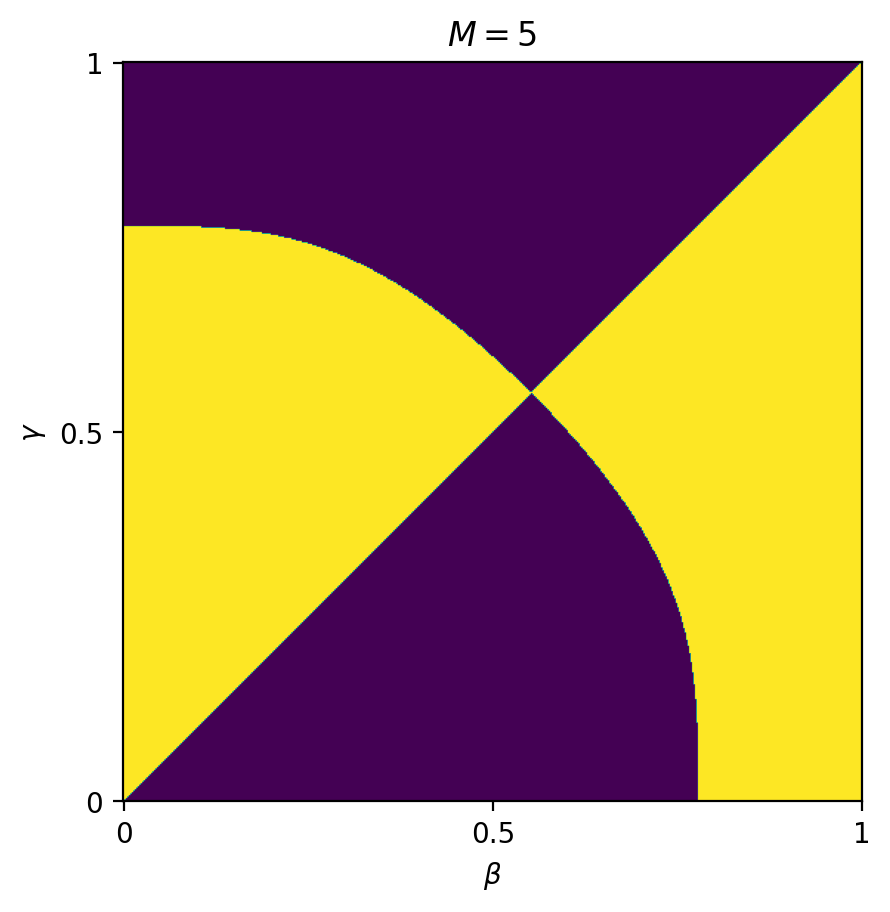}
    }
    \subfloat[\label{fig:condition_bc_M-10} $M=10$]{
        \includegraphics[width=0.16\textwidth]{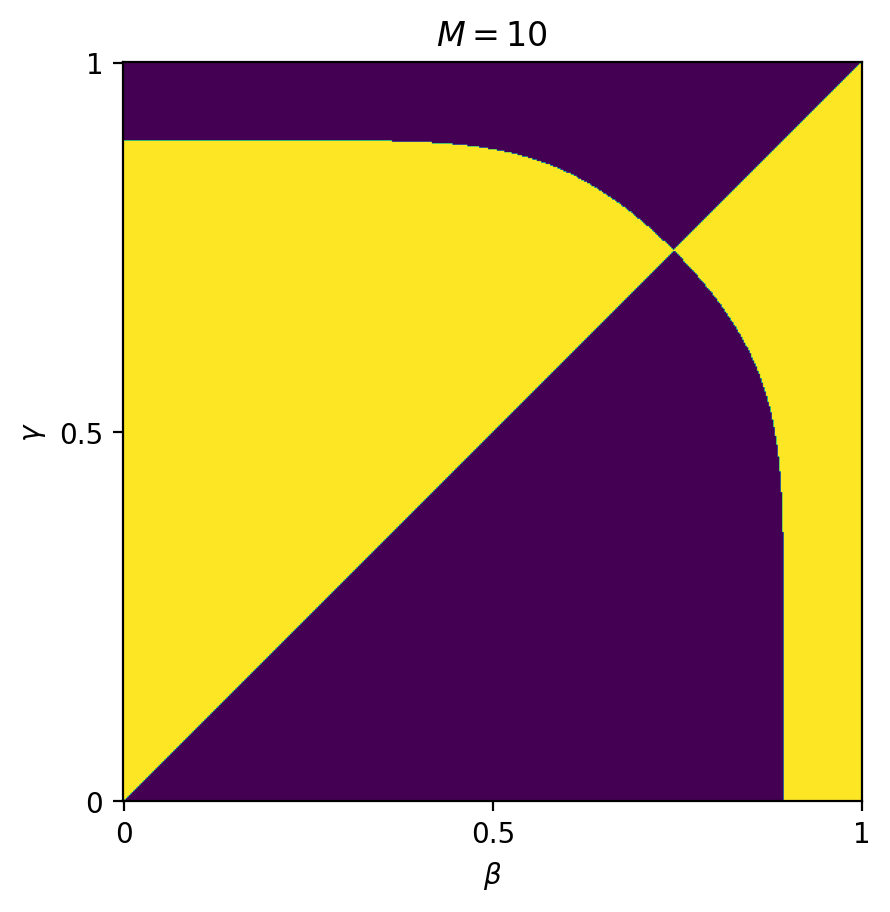}
    }
    \subfloat[\label{fig:condition_bc_M-50} $M=50$]{
        \includegraphics[width=0.16\textwidth]{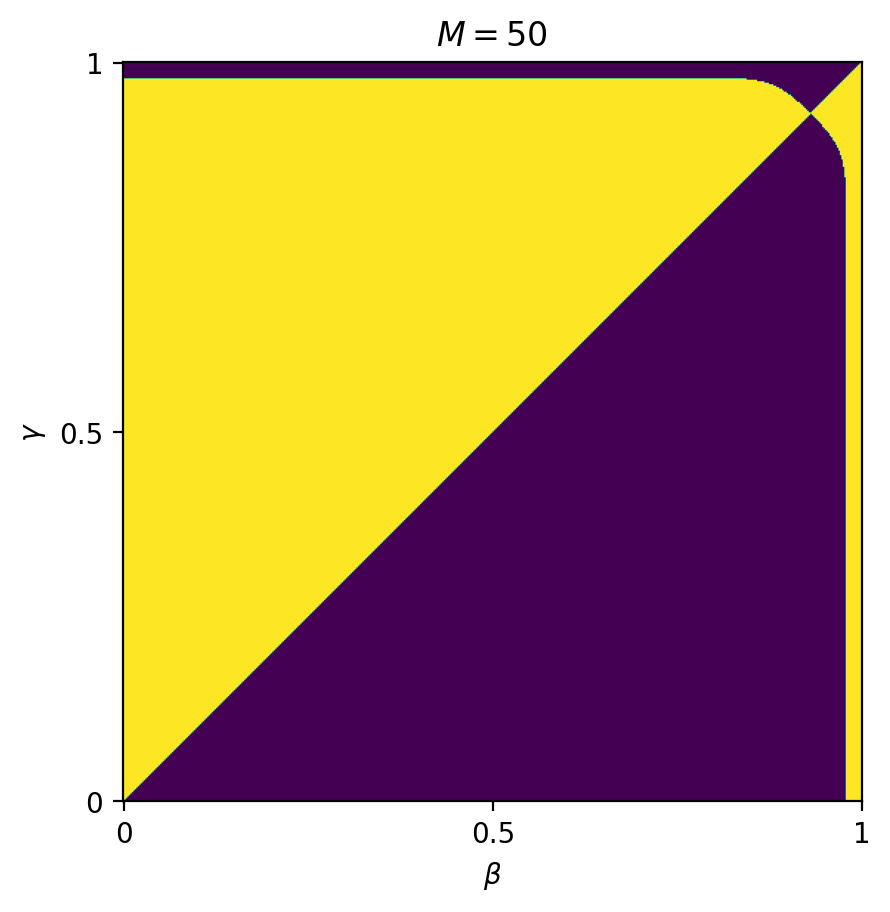}
    }
    \subfloat[\label{fig:condition_bc_M-100} $M=100$]{
        \includegraphics[width=0.16\textwidth]{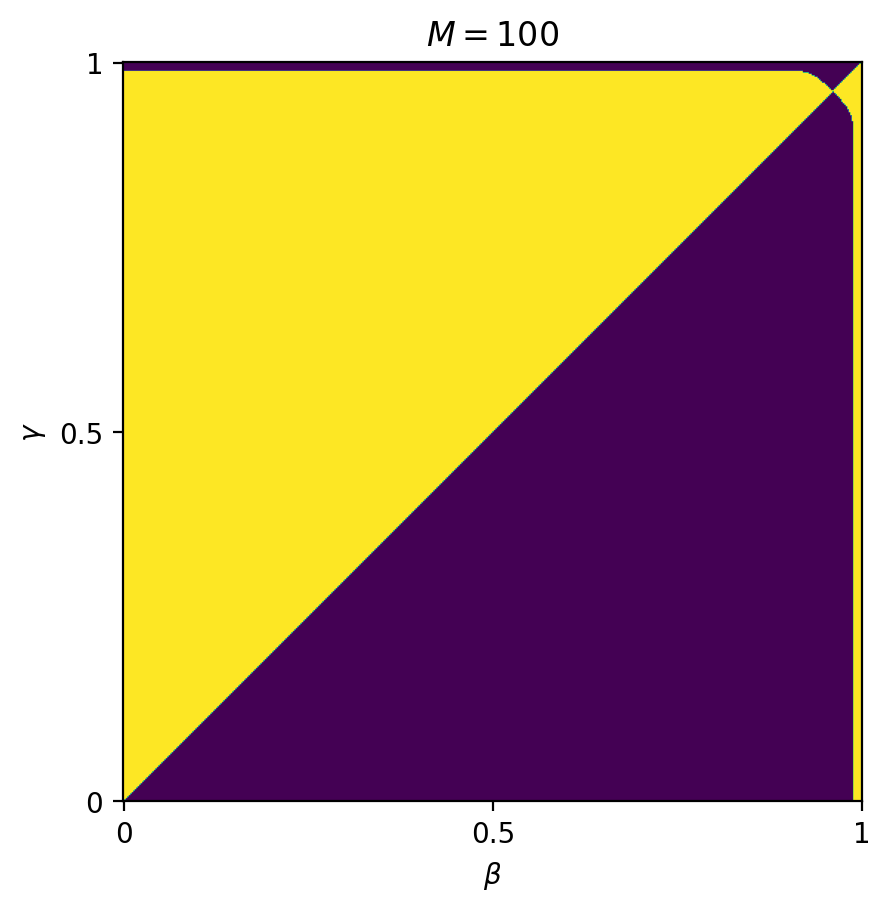}
    }
    \caption{
        The heatmaps that show the sign of $f(\beta, \gamma) - f(\gamma, \beta)$ for $0<\gamma<\beta<1$. 
        In each figure, the areas representing $f(\beta, \gamma) - f(\gamma, \beta) > 0$ are coloured yellow.
        The variable $\beta$ takes from $0$ to $1$ on the $x$-axis, and the variable $\gamma$ takes from $0$ to $1$ on the $y$-axis.
    }
    \label{fig:condition_bc}
\end{figure}

\par
We show the condition for $M$ which satisfies $\displaystyle \frac{d}{d M} \tilde{p} < 0$ for $\gamma=0$ also holds for $0 < \gamma < 1$.
First when $\gamma=0$, the inequality $f(\beta, 0) - f(0, \beta) > 0$ can be simplified into $\displaystyle 1 + (M-1) \log\left(1 - \left(1 - \frac{1}{d}\right)p\right) > 0$.
Solving this inequality for $M$ yields
\begin{equation} \label{eq:condition_M_prime}
\begin{split}
    M &< 1-\frac{1}{\displaystyle \log \left(1-\left(1-\frac{1}{d}\right) p\right)} = 1 - \frac{1}{\log\left(\beta\right)},
\end{split}
\end{equation}
where the factor of dimension $d$ changes the logarithm part only from $1-p$ in Eq.~\eqref{eq:condition_M_d_inf} to $\displaystyle 1 - \left(1-\frac{1}{d}\right)p$.

\par
Next, we check $\displaystyle \frac{d}{d M} \tilde{p} < 0$ also holds for $0 < \gamma < 1$ under the condition Eq.~\eqref{eq:condition_M_prime}.
On the monotonicity of $\tilde{p}$, we prove $f(\beta, \gamma) - f(\gamma, \beta) > 0$.
We first bound $f(\beta, \gamma) - f(\gamma, \beta)$ by using Eq.~\eqref{eq:condition_M_prime} and $0<\gamma<\beta<1$ into
\begin{equation}
\begin{split}
    f(\beta, \gamma) - f(\gamma, \beta)
    &= \beta^{M-1}\left(1 + (M-1) \log \left(\beta\right)+(M-1)^{2} \gamma^{M-1} \left( \log\left(\beta\right) - \log\left(\gamma\right)\right)\right) - \gamma^{M-1} \left(1 + (M-1) \log \left(\gamma\right)\right) \\
    &> \frac{1}{e}\left(1 + (M-1) \frac{1}{1-M}+(M-1)^{2} \gamma^{M-1} \cdot 0\right) -\gamma^{M-1} \left(1 + (M-1) \log \left(\gamma\right)\right) \\
    &= \left(\frac{p}{d}\right)^{M-1} \left( (M-1) \left( \log\left(d\right) - \log\left(p\right) \right) - 1 \right).
\end{split}
\end{equation}
Thus it is sufficient to check $\left(M - 1\right) \left( \log \left(d\right)-\log \left(p\right) \right) > 1$, which is equivalent to $\displaystyle p < d e^{\frac{1}{M-1}}$.
In fact, this holds for $d > 1$ and $M > 1$ since $\displaystyle d e^{\frac{1}{M-1}}$ is always larger than $\displaystyle \frac{d}{d - 1}\left(1 - e^{\frac{1}{1-M}}\right)$ since $\displaystyle M > 1 + \frac {1} {\log\left(d\right)}$.

\subsection{\texorpdfstring{Purification Rate of CGG with $M=2^{n}$ Inputs}{} \label{sec:appendix_proof_CGG_power_of_two}}

\par
For $M = 2^{n}$, noting that we have defined $\displaystyle \beta = 1 - p + \frac{p}{d}$ and $\displaystyle \gamma = \frac{p}{d}$, the unnormalised output quantum state is
\begin{equation}
\begin{split}
    M P_{\vec{0}} \tilde{\rho}
    &= \rho + \sum_{i=1}^{n} \alpha^{i-1} \operatorname{Tr}\left[\rho^{\alpha^{i}}\right]^{\alpha^{n-i}-1} \rho^{\alpha^{i}} \\
    &= \left((1-p)+\sum_{i=1}^{n} \alpha^{i-1}\left(\beta^{\alpha^{i}}+\gamma^{\alpha^{i}-1}\left(1-\beta\right)\right)^{\alpha^{n-i}-1}\left(\beta^{\alpha^{i}}-\gamma^{\alpha^{i}}\right)\right) \rho_{0} \\
    &\quad\quad\quad\quad + \left(p+d\sum_{i=1}^{n} \alpha^{i-1}\left(\beta^{\alpha^{i}}+\gamma^{\alpha^{i}-1}\left(1-\beta\right)\right)^{\alpha^{n-1}-1} \gamma^{\alpha^{i}}\right) \frac{I}{d},
\end{split}
\end{equation}
with the post-selection probability $P_{\vec{0}}$ as
\begin{equation}
\begin{split}
    M P_{\vec{0}}
    &= 1 + \sum_{i=1}^{n} \alpha^{i-1}\left(\beta^{\alpha^{i}}+\gamma^{\alpha^{i}-1}\left(1-\beta\right)\right)^{\alpha^{n-i}-1}\left(\beta^{\alpha^{i}}+\gamma^{\alpha^{i}-1}\left(1-\beta\right)\right) \\
    &= 1 + \sum_{i=1}^{n} \alpha^{i-1}\left(\beta^{\alpha^{i}}+\gamma^{\alpha^{i}-1}\left(1-\beta\right)\right)^{\alpha^{n-i}}.
\end{split}
\end{equation}
Therefore, the normalised output state of CGG becomes $\displaystyle \tilde{\rho} = \left(1 - \tilde{p}\right) \rho_{0} + \tilde{p} \frac{I}{d}$ with the output depolarising rate
\begin{equation} \label{eq:p_tilde_M-power-2_depolarising}
\begin{split}
    \tilde{p} 
    &= \frac {\displaystyle 1 + \sum_{i=1}^{n} \alpha^{i-1}\left(\beta^{\alpha^{i}}+\gamma^{\alpha^{i}-1}\left(1-\beta\right)\right)^{\alpha^{n-1}-1} \gamma^{\alpha^{i}-1}} 
             {\displaystyle 1 + \sum_{i=1}^{n} \alpha^{i-1}\left(\beta^{\alpha^{i}}+\gamma^{\alpha^{i}-1}\left(1-\beta\right)\right)^{\alpha^{n-i}}} p.
\end{split}
\end{equation}
The numerical plots of purification rate when $M=2^{n}$ for $n=1,\ldots,13$ are demonstrated in Fig.~\ref{fig:CGG_ps-tildes_M-all}(a).
Note that taking $M=2^{n}$ copies, the uniform superposition in $n$ ancillary qubits can be trivially prepared by applying local Hadamard gates to $|0\rangle$ in each ancillary qubit.

\section{The Output State and Post-selection Probability of GSG \label{sec:appendix_proof_GSG}}

\par
We first observe the following expansion about the depolarised input $\rho = \left(1-p\right) \rho_{0} + p\frac{I}{d}$ with dimension $d$.
\begin{equation}
\begin{split}
    \rho^{2} &= (1-p)^{2} \operatorname{Tr}\left[\rho_{0}^{2}\right] \frac {\rho_{0}^{2}} {\operatorname{Tr}\left[\rho_{0}^{2}\right]} + \frac {2} {d} p(1-p) \rho_{0} + \frac{p^{2}}{d} \frac{I}{d}, \\
    \operatorname{Tr}\left[\rho^{2}\right] &= (1-p)^{2}\operatorname{Tr}\left[\rho_{0}^{2}\right] + \frac {1} {d} p\left(2-p\right).
\end{split}
\end{equation}
Using this, the output state of GSG and its post-selection probability can be expanded to
\begin{equation} \label{eq:outputs_gSWAP_m_depolarising}
\begin{split}
    \tilde{\rho} 
    &= \frac {1} {\operatorname{P}_{\vec{0}}} \left( \frac {1} {M} \left((1-p)\rho_{0} + p\frac{I}{d}\right) \right. \\
    &\quad\quad\quad + \left. \frac {M - 1} {M} \left((1-p)^{2}\operatorname{Tr}\left[\rho_{0}^{2}\right] + \frac {1} {d} p(2-p)  \right)^{\frac{M}{2}-1} \left((1-p)^{2} \rho_{0}^{2} + \frac {2} {d} p(1-p) \rho_{0} + \frac{p^{2}}{d} \frac{I}{d}\right)\right), \\
    &= \frac {1} {\operatorname{P}_{\vec{0}}} \left( \frac{M-1}{M}\left((1-p)^{2}\operatorname{Tr}\left[\rho_{0}^{2}\right]+\frac{1}{d}p(2-p)\right)^{\frac{M}{2}-1}(1-p)^{2} \operatorname{Tr}\left[\rho_{0}^{2}\right] \frac {\rho_{0}^{2}} {\operatorname{Tr}\left[\rho_{0}^{2}\right]} \right. \\
    &\quad\quad\quad + \left( \frac{1}{M}(1-p) + \frac{M-1}{M}\left((1-p)^{2}\operatorname{Tr}\left[\rho_{0}^{2}\right]+\frac{1}{d}p(2-p)\right)^{\frac{M}{2}-1} \frac{2}{d}p(1-p) \right) \rho_{0} \\
    &\quad\quad\quad + \left. \left( \frac{1}{M}p + \frac{M-1}{M}\left((1-p)^{2}\operatorname{Tr}\left[\rho_{0}^{2}\right]+\frac{1}{d}p(2-p)\right)^{\frac{M}{2}-1} \frac{p^{2}}{d} \right) \frac{I}{d} \right), \\
    \operatorname{P}_{\vec{0}} 
    &= \frac {1} {M} + \frac {M - 1} {M} \left((1-p)^{2}\operatorname{Tr}\left[\rho_{0}^{2}\right] + \frac {1} {d} p(2-p) \right)^{\frac{M}{2}}.
\end{split}
\end{equation}
From \eqref{eq:outputs_gSWAP_m_depolarising}, when $\rho_{0}$ is a pure state,
\begin{equation} \label{eq:outputs_gSWAP_m_depolarising_pure_complete}
\scalemath{0.9}{
\begin{aligned}
    \tilde{\rho} 
    &= \frac {1} {\operatorname{P}_{\vec{0}}} \left( \left( \frac{1}{M}(1-p) + \frac{M-1}{M}\left((1-p)^{2} + \frac{1}{d}p(2-p)\right)^{\frac{M}{2}-1} \left((1-p)^{2} + \frac{2}{d}p(1-p)\right) \right)\right. \rho_{0} \\
    &\quad\quad\quad + \left. \left( \frac{1}{M}p + \frac{M-1}{M}\left((1-p)^{2} + \frac{1}{d}p(2-p)\right)^{\frac{M}{2}-1} \frac{p^{2}}{d} \right) \frac{I}{d} \right) \\
    &= \frac {1} {\operatorname{P}_{\vec{0}}} \left( \left( \frac{1}{M}(1-p) + \frac{M-1}{M}\left( 1-2\left(1-\frac{1}{d}\right)p+\left(1-\frac{1}{d}\right)p^{2} \right)^{\frac{M}{2}-1} \left(1-2\left(1-\frac{1}{d}\right)p+\left(1-\frac{2}{d}\right)p^{2}\right) \right)\right. \rho_{0} \\
    &\quad\quad\quad + \left. \left( \frac{1}{M}p + \frac{M-1}{M}\left( 1-2\left(1-\frac{1}{d}\right)p+\left(1-\frac{1}{d}\right)p^{2} \right)^{\frac{M}{2}-1} \frac{p^{2}}{d} \right) \frac{I}{d} \right), \\
    \operatorname{P}_{\vec{0}}
    &= \frac {1} {M} + \frac {M - 1} {M} \left(1-2\left( 1-\frac{1}{d}\right)p+\left(1-\frac{1}{d}\right)p^{2} \right)^{\frac{M}{2}}.
\end{aligned}
}
\end{equation}
It is clear to see Eq.~\eqref{eq:outputs_gSWAP_m_depolarising_pure_complete} corresponds to Eq.~\eqref{eq:outputs_SWAP_depolarising_pure} when $M=2$.
When the dimension $d$ goes to infinity in Eq.~\eqref{eq:outputs_gSWAP_m_depolarising_pure_complete}, $\tilde{\rho}$ and $\operatorname{P}_{\vec{0}}$ converge to
\begin{equation} \label{eq:outputs_gSWAP_m_depolarising_pure_d_inf_complete}
\begin{split}
    \tilde{\rho} 
    &\overset{d\rightarrow\infty}{\longrightarrow} \frac {1} {\operatorname{P}_{\vec{0}}} \left( \frac{1 - p}{M} \left( 1 + \left(M - 1\right)\left(1 - p\right)^{M-1} \right) \rho_{0} + \frac{p}{M} \frac{I}{d} \right), \\
    \operatorname{P}_{\vec{0}}
    &\overset{d\rightarrow\infty}{\longrightarrow} \frac {1} {M} + \frac {M - 1} {M} \left(1 - p\right)^{M},
\end{split}
\end{equation}
yielding the output depolarising rate in Eq.~\eqref{eq:outputs_gSWAP_m_depolarising_pure_d_inf}.

\section{Purification-based QEM Gadget \label{sec:appendix_purification_based_qem_gadget}}

\par
Recently, with the advent of near-term quantum devices, many quantum error mitigation (QEM) techniques have been proposed to suppress the bias in expectation values of measurement results through classical pre- and postprocessing~\cite{cai2023quantum, endo2021hybrid, huggins2021virtual, koczor2021exponential, yang2023dual}.
The purification-based QEM methods, such as virtual distillation (VD)~\cite{huggins2021virtual} and exponential suppression by derangement (ESD)~\cite{koczor2021exponential}, adopt the same idea as projecting the multiple copies of noisy states into their symmetric subspace.
As those methods aim to improve the expectation value instead of the quantum state itself, they can make their quantum circuits much lighter.

\par
Here, we review the exponential error suppression by derangement (ESD) method~\cite{koczor2021exponential} and design a state purification gadget from it.
The ESD circuit uses $M$ copies of noisy quantum states to suppress the stochastic errors in magnitude to $M$ by applying one derangement operation to $M$ copies.
Then, the mitigated expectation value for an observable $O$ becomes
\begin{equation} \label{eq:expval_esd}
\begin{split}
    \langle O\rangle_{\mathrm{ESD}} = \frac {\operatorname{Tr}\left[\rho^{M} O\right]} {\operatorname{Tr}\left[\rho^{M}\right]}.
\end{split}
\end{equation}
When $\rho$ is given as a classical mixture of the target state $|\lambda_{0}\rangle\langle\lambda_{0}|$ with a dominant probability and unwanted biased states $\{|\lambda_{i}\rangle\langle\lambda_{i}|\}_{i>1}$ with small probabilities, $\rho$ can be expanded as $\displaystyle \rho = \sum_{k} \lambda_{k}|\lambda_{k}\rangle\langle\lambda_{k}|$ whose eigenvalues are aligned in descending order.
The $M$-th power of $\rho$ becomes
\begin{equation}
\begin{split}
    \rho^{M}
    &= \sum_{k} \lambda_{k}^{M}|\lambda_{k}\rangle\langle\lambda_{k}| 
    = \lambda_{0}^{M}\left(|\lambda_{0}\rangle\langle\lambda_{0}| + \sum_{k>0} \left(\frac {\lambda_{k}} {\lambda_{0}}\right)^{M}|\lambda_{k}\rangle\langle\lambda_{k}|\right).
\end{split}
\end{equation}
Then, the mitigated expectation value $\langle O\rangle_{\mathrm{ESD}}$ can be expanded as
\begin{equation} \label{eq:expval_esd_expanded}
\begin{split}
    \langle O\rangle_{\mathrm{ESD}}
    &= \frac 
    {\displaystyle \lambda_{0}^{M}\left(\operatorname{Tr}\left[|\lambda_{0}\rangle\langle\lambda_{0}|O\right] + \sum_{k>0} \left(\frac {\lambda_{k}} {\lambda_{0}}\right)^{M}\operatorname{Tr}\left[|\lambda_{k}\rangle\langle\lambda_{k}|O\right]\right)} 
    {\displaystyle \lambda_{0}^{M}\left(\operatorname{Tr}\left[|\lambda_{0}\rangle\langle\lambda_{0}|\right] + \sum_{k>0} \left(\frac {\lambda_{k}} {\lambda_{0}}\right)^{M}\operatorname{Tr}\left[|\lambda_{k}\rangle\langle\lambda_{k}|\right]\right)} \\
    &=\operatorname{Tr}\left[|\lambda_{0}\rangle\langle\lambda_{0}|O\right] + \sum_{k>0}\left(\frac {\lambda_{k}}{\lambda_{0}}\right)^{M}\left(\operatorname{Tr}\left[|\lambda_{k}\rangle\langle\lambda_{k}|O\right] - \operatorname{Tr}\left[|\lambda_{0}\rangle\langle\lambda_{0}|O\right]\right) + O\left(x^{2}\right)
\end{split}
\end{equation}
where $\displaystyle x = \sum_{k>0}\left(\frac {\lambda_{k}}{\lambda_{0}}\right)^{M}$.
Since $\lambda_{k} < \lambda_{0}$, the factors of unwanted states and the contribution of $O\left(x^{2}\right)$ are suppressed exponentially to $M$ in Eq.~\eqref{eq:expval_esd_expanded}.

\begin{figure}[htbp]
    \centering
    \subfloat[\label{fig:qcs_esd_a}]{
        \includegraphics[width=0.4\textwidth]{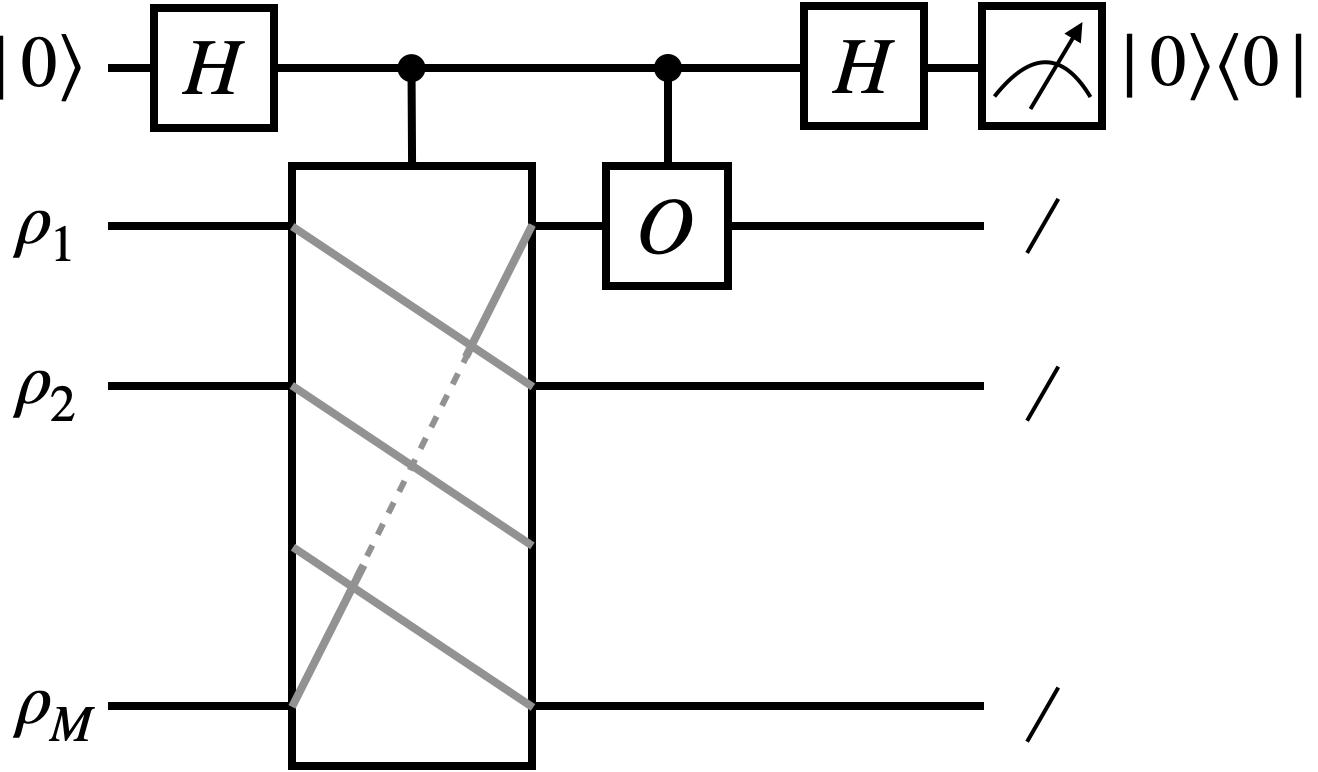}
    }
    \hspace{10pt}
    \subfloat[\label{fig:qcs_esd_b}]{
        \includegraphics[width=0.4\textwidth]{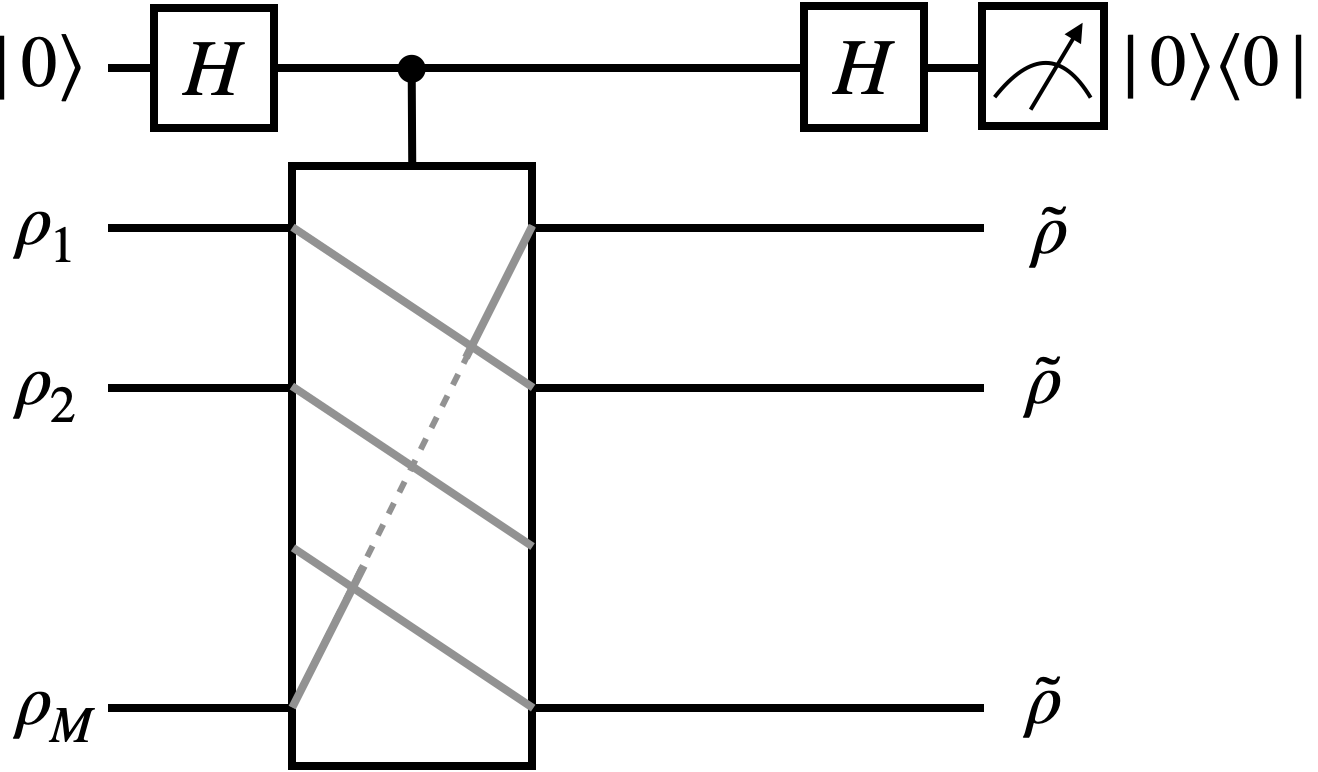}
    }
    \caption{
        (a) The quantum circuit for computing $\operatorname{Tr}\left[\rho_{1}\rho_{2}\cdots\rho_{M}O \right]$ in ESD.
            The cyclic shift $P_{23\cdots M1}$ over $M$ state copies is one of the derangement operators.
            The expectation value $\operatorname{Tr}\left[\rho_{1}\rho_{2}\cdots\rho_{M} \right]$ is obtained by removing the controlled-$O$ operator.
        (b) The quantum circuit which post-selects the first copy after the ESD process.
    }
    \label{fig:qcs_esd}
\end{figure}

\par
In ESD, $\operatorname{Tr}\left[\rho^{M} O\right]$ and $\operatorname{Tr}\left[\rho^{M}\right]$ are obtained from the quantum circuits with a derangement operation that performs a permutation such that no state appears in its original position.
We can pick up cyclic shift $P_{23\cdots M1}$ as an example of the derangement operator.
The quantum circuit for $\operatorname{Tr}\left[\rho^{M} O\right]$ and $\operatorname{Tr}\left[\rho^{M}\right]$ contains only one ancillary qubit and one derangement operation as shown in Fig.~\ref{fig:qcs_esd}(a).
Since the post-selection on $|0\rangle\langle0|$ state in the ancillary qubit gives the post-selection probability $\displaystyle \operatorname{P}_{0} = \frac{1}{2}\left(1 + \operatorname{Tr}\left[\rho^{M}O\right]\right)$ for $\rho_{1}=\cdots=\rho_{M}=\rho$,
$\operatorname{Tr}\left[\rho^{M} O\right]$ can be computed as
\begin{equation}
    \operatorname{Tr}\left[\rho^{M} O\right] = 1 - 2\operatorname{P}_{0}.
\end{equation}

\par
While the ESD circuit is simpler than that of the generalised SWAP test and that of the SGG, the exponential error suppression to the number of state copies $M$ is available only to mitigate the expectation value.
In fact, the state purification performance of the purification gadget inspired by the ESD circuit in Fig.~\ref{fig:qcs_esd}(b) does not amplify the contribution of $\operatorname{Tr}\left[\rho^{M}\right]$ scaling to $M$, which can be immediately checked by computing $\tilde{\rho}$ and $\operatorname{P}_{0}$:
\begin{equation} \label{eq:outputs_esd_m}
\begin{split}
    \tilde{\rho} &= \frac {1} {2\operatorname{P}_{\vec{0}}} \left(\rho + \rho^{M}\right), \\
    \operatorname{P}_{0} &= \frac {1} {2} \left(1 + \operatorname{Tr}\left[\rho^{M}\right]\right).
\end{split}
\end{equation}
Since $\operatorname{Tr}\left[\rho^{M}\right]$ shrinks in magnitude to $M$, the contribution of $\rho^{M}$ vanishes quickly in $\tilde{\rho}$.

\par
We remark that this does not surpass the purification rate of the SWAP gadget with two depolarised inputs.
Taking depolarised inputs $\displaystyle \rho = \left(1 - p\right) \rho_{0} + p\frac{I}{d}$, the output depolarising rate $\tilde{p}$ becomes
\begin{equation} \label{eq:p_tilde_esd_depolarising}
\begin{split}
    \tilde{p} 
    &= 
        \frac { \displaystyle 1 + \left(\frac{p}{d}\right)^{M-1} } 
              { \displaystyle 1 + \left(1 - p + \frac{p}{d}\right)^{M} + \left(\frac{p}{d}\right)^{M}\left(d - 1\right) } p
    \overset{d\rightarrow\infty}{\longrightarrow} 
        \frac { \displaystyle 1 } 
              { \displaystyle 1 + \left(1 - p\right)^{M}} p
    > \frac{1}{2-p}p.
\end{split}
\end{equation}
Therefore, there is no improvement in the purification rate when increasing the number of copies with this gadget.

\section{Fidelity of Output State with Inputs Under Arbitrary Stochastic Noise \label{sec:appendix_general_stochastic_noise}}

\par
In this section, we generalise the noise in input states from the depolarising noise to the general form of stochastic noise described as
\begin{equation}
    \rho\left(p\right) = \left(1 - p\right) \rho_{0} + p\sigma.
\end{equation}
The fidelity $F = \operatorname{Tr}\left[\rho_{0}\rho\right]$ between the target pure state $\rho_{0}=|\psi\rangle\langle\psi|$ and the noisy input state $\rho$ becomes
\begin{equation}
    F = \operatorname{Tr}\left[\rho_{0}\rho\right] = \operatorname{Tr}\left[\rho_{0}\left(\left(1-p\right)\rho_{0}+p\sigma\right)\right] = 1-p+pF_{1} = 1-p\left(1-F_{1}\right) = 1-pI_{1},
\end{equation}
where $F_{1} = \operatorname{Tr}\left[\rho_{0}\sigma\right]$ is the fidelity between $\rho_{0}$ and $\sigma$ and $I_{1} = 1 - F_{1}$ is infidelity.

\par
When we apply CGG over $M$ copies of $\rho$, the fidelity between $\rho_{0}$ and the output state $\tilde{\rho}$ is
\begin{equation} \label{eq:F_tilde}
\begin{split}
    \tilde{F} 
    = F(\rho_{0}, \tilde{\rho})
    &= \frac{1}{\operatorname{P}_{\vec{0}}} \operatorname{Tr}\left[\rho_{0}\left(\frac{1}{M} \sum_{m|M}\varphi\left(m\right)\operatorname{Tr}\left[\rho^{m}\right]^{\frac{M}{m}-1}\rho^{m}\right)\right]
    = \frac{1}{\operatorname{P}_{\vec{0}}} \frac{1}{M} \sum_{m|M}\varphi\left(m\right)\operatorname{Tr}\left[\rho^{m}\right]^{\frac{M}{m}-1} \operatorname{Tr}\left[\rho_{0}\rho^{m}\right].
\end{split}
\end{equation}
Note that $\operatorname{P}_{\vec{0}}$ in the denominator is still the same as Eq.~\eqref{eq:outputs_cgg_prime}, $\displaystyle \operatorname{P}_{\vec{0}} = \frac{1}{M} \sum_{m|M}\varphi\left(m\right)\operatorname{Tr}\left[\rho^{m}\right]^{\frac{M}{m}-1} \operatorname{Tr}\left[\rho^{m}\right]$.
Now, we analyse $\operatorname{Tr}\left[\rho_{0}\rho^{m}\right]$ in Eq.~\eqref{eq:F_tilde} by expanding it to
\begin{equation} \label{eq:trace_rho_0_rho^M}
\begin{split}
    \operatorname{Tr}\left[\rho_{0}\rho^{m}\right]
    &= \operatorname{Tr}\left[\rho_{0}\left(\left(1-p\right)\rho_{0}+p\sigma\right)^{m}\right] \\
    &= (1-p)^{m} + \sum_{i=1}^{m}\left(1-p\right)^{m-i}p^{i} \sum_{l=1}^{\min\{i, m-i+1\}} \sum_{\substack{j\in \operatorname{partition}(i, l) \\ 1^{j_{1}}2^{j_{2}}\cdots i^{j_{i}}}} \prod_{k=1}^{i} {}_{(m-i+1)-(\sum_{k'=1}^{k-1}j_{k'})}C_{j_{k}} F_{k}^{j_{k}} \\
    &= (1-p)^{m} + \sum_{i=1}^{m}\left(1-p\right)^{m-i}p^{i} \sum_{l=1}^{\min\{i, m-i+1\}} {}_{m-i+1}C_{l} \sum_{\substack{j\in \operatorname{partition}(i, l) \\ 1^{j_{1}}2^{j_{2}}\cdots i^{j_{i}}}} \frac{l!}{j_{1}!j_{2}!\cdots j_{i}!} F_{k}^{j_{k}} \\
    &= (1-p)^{m} + \sum_{i=1}^{m}\left(1-p\right)^{m-i}p^{i} \sum_{l=1}^{i} {}_{m-i+1}C_{l} \sum_{\substack{j\in \operatorname{partition}(i, l) \\ 1^{j_{1}}2^{j_{2}}\cdots i^{j_{i}}}} \frac{l!}{j_{1}!j_{2}!\cdots j_{i}!} F_{k}^{j_{k}}, \\
\end{split}
\end{equation}
where $F_{k} = \operatorname{Tr}\left[\rho_{0}\sigma^{k}\right]$.
We obtain the last line of Eq.~\eqref{eq:trace_rho_0_rho^M} by defining ${}_{n}C_{r}=0$ for $n<r$.

\par
Each coefficient of $F_{k}^{j_{k}}$ is associated with each Young diagram.
The case for $M=5$ is shown in Fig.~\ref{fig:table_young_diagrams}.
The coefficient $\displaystyle \prod_{k=1}^{i}{}_{(m-i+1)-(\sum_{k'=1}^{k-1}j_{k'})}C_{j_{k}}$ comes from sequentially choosing $j_{k}$ identical sequences of multiplication of $k$ $\sigma$ from all possible remaining $(m-i+1)-(\sum_{k'=1}^{k-1}j_{k'})$ positions for the sequence of multiplication of $\sigma$, following the stars-and-bars reasoning, so that the sequences of multiplication of $\sigma$ will not be next to each other.
In the third line of Eq.~\eqref{eq:trace_rho_0_rho^M}, we can also interpret the coefficient as first choosing the position of locating $l$ sequences of multiplication of $\sigma$, and then consider the permutation of $l$ sequences by identifying sequences consisting of the same number of $\sigma$.
\begin{figure*}[htbp]
    \centering
    \includegraphics[width=0.7\textwidth]{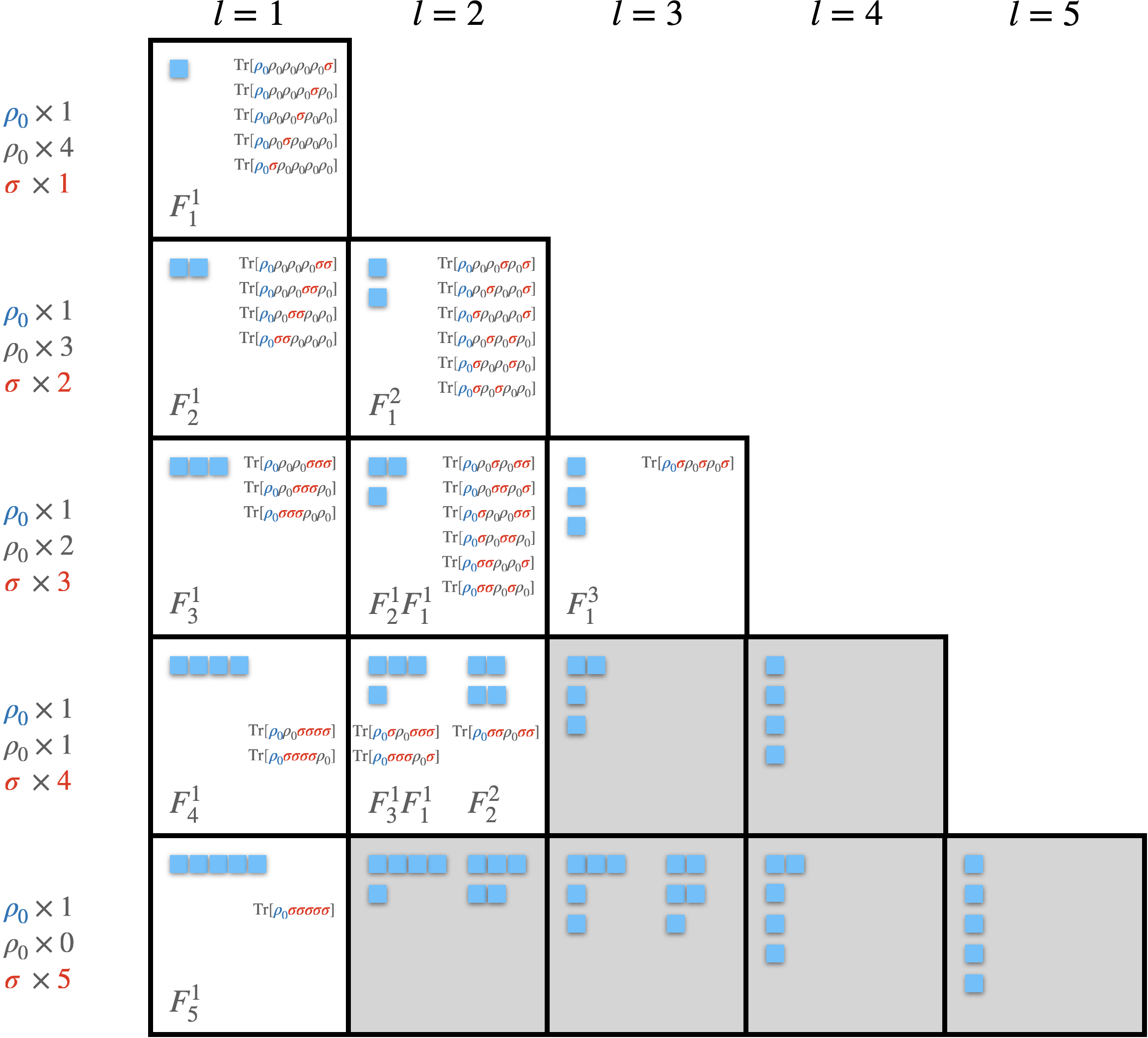}
    \caption{
        The table of Young diagrams for $m=5$.
        The grey cells are not counted (or counted as $0$) in the summation due to the restriction $\#\left(\rho_{0}\right) + \#\left(\sigma\right) = m$ (or to the definition ${}_{n}C_{r} = 0$ for $n<r$).
    }
    \label{fig:table_young_diagrams}
\end{figure*}

\par
We might be able to simplify Eq.~\eqref{eq:trace_rho_0_rho^M} by introducing several different assumptions on the noisy state $\sigma$.
\begin{enumerate}
    \item $\sigma$ is pure: $\sigma^{2}=\sigma$,
    \item $\rho_{0}$ and $\sigma$ are commutative: $\left[\rho_{0}, \sigma\right] = 0$.
\end{enumerate}
These conditions are realistic as condition (1) covers the dephasing error, and condition (2) covers the depolarising error and general orthogonal errors.

\subsection{\texorpdfstring{When $\sigma$ is pure: $\sigma^{2}=\sigma$}{}}

\par
First, when the condition (1) holds, the fidelity can be simplified as
\begin{equation}
    F_{k} = \operatorname{Tr}\left[\rho_{0}\sigma^{k}\right] = \operatorname{Tr}\left[\rho_{0}\sigma\right] = F_{1}.
\end{equation}
Therefore, Eq.~\eqref{eq:trace_rho_0_rho^M} can be simplified into
\begin{equation}
\begin{split}
    \operatorname{Tr}\left[\rho_{0}\rho^{m}\right] 
    &= (1-p)^{m} + \sum_{i=1}^{m}\left(1-p\right)^{m-i}p^{i} \sum_{l=1}^{\min\{i, m-i+1\}} {}_{i-1}C_{l-1} \times {}_{l+1}C_{m-i-(l-1)} F_{1}^{l} \\
    &= (1-p)^{m} + \sum_{i=1}^{m}\left(1-p\right)^{m-i}p^{i} \sum_{l=1}^{i} {}_{i-1}C_{l-1} \times {}_{l+1}C_{m-i-(l-1)} F_{1}^{l},
\end{split}
\end{equation}
if we again define ${}_{n}C_{r}=0$ for $n<r$.
The number of combinations ${}_{i-1}C_{l-1} \times {}_{l+1}C_{m-i-(l-1)}$ for each $i$ and $l$ is obtained by considering the following rule.
Since $l$ represents the number of separated sequences of multiplication of $\sigma$, and each sequence should contain at least one $\sigma$, we have to choose which place to insert $\rho_{0}$ as separators.
There are ${}_{i-1}C_{l-1}$ ways to choose the positions for the separators.
Next, we still have to consider different arrangements of remaining $\rho_{0}$ and separated sequences of multiplication of $\sigma$.
Note that in this step, we can identify the sequences of multiplication with different numbers of $\sigma$.
Therefore, the number of the arrangements can be computed as ${}_{l+1}C_{M-i-(l-1)}$, since we are choosing $m-i-(l-1)$ remaining $\rho_{0}$ from $l+1$ positions following the stars-and-bars reasoning.

\subsection{\texorpdfstring{When $\rho_{0}$ and $\sigma$ are commutative: $\left[\rho_{0}, \sigma\right] = 0$}{}}

\par 
Next, when the condition (2) holds, the fidelity can be simplified as
\begin{equation}
\begin{split}
    \prod_{k=1}^{i} F_{k}^{j_{k}} = \operatorname{Tr}\left[\rho_{0}^{m-i+1}\sigma^{i}\right] 
    = \operatorname{Tr}\left[\rho_{0}\sigma^{i}\right] = F_{i} = \operatorname{Tr}\left[\rho_{0}\sigma\rho_{0}\sigma\cdots\rho_{0}\sigma\right] = F_{1}^{i}, 
\end{split}
\end{equation}
which means $F_{i} = F_{1}^{i}$.
Since $\left[\rho_{0}, \sigma\right] = 0$, $\rho_{0}$ and $\sigma$ can be simultaneously diagonalised.
Thus we can further put $\displaystyle \rho_{0} = |\psi\rangle\langle\psi| = |\lambda_{0}\rangle\langle\lambda_{0}|$ and $\displaystyle \sigma = \sum_{i}\lambda_{i}|\lambda_{i}\rangle\langle\lambda_{i}|$, which gives $F_{1} = \lambda_{0}$.
These settings simplify Eq.~\eqref{eq:trace_rho_0_rho^M} to
\begin{equation} \label{eq:trace_rho_0_rho^M_pure}
\begin{split}
    \operatorname{Tr}\left[\rho_{0}\rho^{m}\right] 
    &= \sum_{i=0}^{m}\left(1-p\right)^{m-i}p^{i} {}_{m}C_{i} F_{i}
    = \sum_{i=0}^{m}\left(1-p\right)^{m-i}\left(pF_{1}\right)^{i} {}_{m}C_{i} \\
    &= \left(1-p+pF_{1}\right)^{m}
    = \left(1-p\left(1-F_{1}\right)\right)^{m} 
    = \left(1-p\left(1-\lambda_{0}\right)\right)^{m}
    = \left(1-pI_{1}\right)^{m},
\end{split}
\end{equation}
where we also defined the infidelity between $\rho_{0}$ and $\sigma$ as $I_{1}=1-F_{1}$.
Note that 
\begin{equation}
\begin{split}
    \operatorname{Tr}\left[\rho^{m}\right]
    &= \operatorname{Tr}\left[\left(\left(1-p\right) \rho_{0} + p\sigma\right)^{m}\right]
    = \operatorname{Tr}\left[\sum_{i=0}^{m} {}_{m}C_{i} \left(1-p\right)^{m-i} p^{i} \rho_{0}^{m-i} \sigma^{i}\right] \\
    &= \sum_{i=0}^{m} {}_{m}C_{i} \left(1-p\right)^{m-i} p^{i} \operatorname{Tr}\left[\rho_{0} \sigma^{i}\right] + p^{m}\left(\operatorname{Tr}\left[\sigma^{m}\right]-\operatorname{Tr}\left[\rho_{0} \sigma^{m}\right]\right) \\
    &= \left(1-p\left(1-F_{1}\right)\right)^{m} + p^{m}\left(\operatorname{Tr}\left[\sigma^{m}\right]-F_{1}^{m}\right) \\
    &= \left(1-p\left(1-\lambda_{0}\right)\right)^{m} + \sum_{i>0}\left(p\lambda_{i}\right)^{m}.
\end{split}
\end{equation}
Hence, replacing $\operatorname{Tr}\left[\rho_{0}\rho^{m}\right]$ in Eq.~\eqref{eq:F_tilde} with Eq.~\eqref{eq:trace_rho_0_rho^M_pure} gives 
\begin{equation} \label{eq:F_tilde_pure}
\begin{split}
    \tilde{F} 
    &= \frac { F_{1} + \left(m-1\right)\operatorname{Tr}\left[\rho_{0}\rho^{m}\right] } 
             { 1 + \left(m-1\right)\operatorname{Tr}\left[\rho^{m}\right] }
    = \frac { \displaystyle \lambda_{0} + \left(m-1\right)\left(1-p\left(1-\lambda_{0}\right)\right)^{m} }
            { \displaystyle 1 + \left(m-1\right)\left(1-p\left(1-\lambda_{0}\right)\right)^{m} + (m-1)\sum_{i>0}\left(p\lambda_{i}\right)^{m} }.
\end{split}
\end{equation}


\end{document}